\documentclass[conference]{IEEEtran}

\usepackage{mathtools,amsthm,amssymb} 
\usepackage{aliascnt} 
\usepackage{algorithm}

\usepackage{booktabs}
\usepackage{array}
\usepackage{multirow}
\usepackage{enumitem} 
\usepackage{xcolor}
\usepackage{tcolorbox}
\usepackage[hidelinks]{hyperref} 
\usepackage{cleveref}
\usepackage{needspace}
\usepackage{fancyhdr}
\usepackage{tikz} 
\usetikzlibrary{shapes,arrows,arrows.meta,positioning,calc,fit,shapes.geometric}

\usepackage{etoolbox}

\usepackage{balance}    
\usepackage{needspace}  
\usepackage{etoolbox}   
\preto\section{\needspace{6\baselineskip}}
\preto\subsection{\needspace{5\baselineskip}}
\preto\subsubsection{\needspace{6\baselineskip}}

\theoremstyle{definition}
\newtheorem{theorem}{Theorem}[section]   
\newaliascnt{lemma}{theorem}
\newtheorem{lemma}[lemma]{Lemma}
\aliascntresetthe{lemma}
\newaliascnt{proposition}{theorem}
\newtheorem{proposition}[proposition]{Proposition}
\aliascntresetthe{proposition}
\newaliascnt{corollary}{theorem}

\aliascntresetthe{corollary}
\newaliascnt{observation}{theorem}

\aliascntresetthe{observation}
\newaliascnt{definition}{theorem}
\newtheorem{definition}[definition]{Definition}
\aliascntresetthe{definition}
\newaliascnt{remark}{theorem}
\newtheorem{remark}[remark]{Remark}
\aliascntresetthe{remark}
\newaliascnt{note}{theorem}

\aliascntresetthe{note}
\newaliascnt{assumption}{theorem}

\aliascntresetthe{assumption}
\newaliascnt{condition}{theorem}
\newtheorem{condition}[condition]{Condition}
\aliascntresetthe{condition}
\newaliascnt{example}{theorem}
\newtheorem{example}[example]{Example}
\aliascntresetthe{example}

\crefname{theorem}{Theorem}{Theorems}
\Crefname{theorem}{Theorem}{Theorems}
\crefname{lemma}{Lemma}{Lemmas}
\Crefname{lemma}{Lemma}{Lemmas}
\crefname{proposition}{Proposition}{Propositions}
\Crefname{proposition}{Proposition}{Propositions}
\crefname{corollary}{Corollary}{Corollaries}
\Crefname{corollary}{Corollary}{Corollaries}
\crefname{observation}{Observation}{Observations}
\Crefname{observation}{Observation}{Observations}
\crefname{definition}{Definition}{Definitions}
\Crefname{definition}{Definition}{Definitions}
\crefname{remark}{Remark}{Remarks}
\Crefname{remark}{Remark}{Remarks}
\crefname{example}{Example}{Examples}
\Crefname{example}{Example}{Examples}
\crefname{note}{Note}{Notes}
\Crefname{note}{Note}{Notes}
\crefname{assumption}{Assumption}{Assumptions}
\Crefname{assumption}{Assumption}{Assumptions}
\crefname{condition}{Condition}{Conditions}
\Crefname{condition}{Condition}{Conditions}
\crefname{algorithm}{Algorithm}{Algorithms}
\Crefname{algorithm}{Algorithm}{Algorithms}
\crefname{table}{Table}{Tables}
\Crefname{table}{Table}{Tables}
\crefname{figure}{Figure}{Figures}
\Crefname{figure}{Figure}{Figures}
\crefname{appendix}{Appendix}{Appendices}
\Crefname{appendix}{Appendix}{Appendices}

\begin{document}
\bstctlcite{IEEEexample:BSTcontrol}

\title{Certificate-Governed CRT Sparse FFT: Verified Global-Label Candidate Construction under Explicit Decoding Models}

\author{
  Aaron R. Flouro and Shawn P. Chadwick, PhD. \\
  SparseTech, Iowa, USA \\
  research@sparse-tech.com
}

\maketitle

\begin{abstract}
We present a deterministic CRT-based sparse Fourier architecture for exact recovery in a noiseless, on-grid, at-most-$k$-sparse model with an engineered transform length. The length $N$ is selected at design time as an exact product of small pairwise-coprime stage primes, so each stage transform is sized by the sparsity rather than by the ambient dimension. At each stage, three co-prime-increment time shifts provide a one-sided singleton-consistency test: every genuine singleton passes, but a passing bin is only a candidate-singleton. Each candidate-singleton is submitted to a total label-emission rule that emits at most one global frequency label, and the per-stage label sets are intersected to form a candidate set $\mathcal{G}$ of deterministically bounded size $O(k)$ for every input in the engineered exact-product family.

The sparse-attempt cost depends on the declared decoding model. In a comparison real-RAM model with exact real arithmetic and exact $\arg(\cdot)$, but without unit-cost root-index reads, candidate construction uses $O(k \log N)$ samples and $O(k \log^2 N / \log k)$ arithmetic. In the stronger root-index-oracle model, with no precomputed root table, the arithmetic tightens to $O(k \log N)$. A verified sparse output adds a consecutive-sample residual verifier, costing $O(k^2)$ arithmetic in the oracle model. Neither model is bit-level, and no bit-complexity claim is made.

Correctness is verifier-gated rather than screen-gated. All-stage genuine-singleton survival of a nonempty support is sufficient for the candidate set to contain the true support and for the sparse path to complete; without that condition, the sparse attempt can omit true tones or admit structured phantom labels. The final verifier reads $k + |\mathcal{G}|$ consecutive residual samples: a PASS proves the sparse output exact, while any sparse-path failure routes to a dense FFT. Thus the hybrid algorithm returns the exact spectrum for every input in the noiseless, on-grid, at-most-$k$-sparse model, with $O(N \log N)$ worst-case runtime and no randomized sampling or hashing.
\end{abstract}

\begin{IEEEkeywords}
Sparse FFT, Verified Recovery, Chinese Remainder Theorem, Global-Label Engine, Deterministic Algorithms
\end{IEEEkeywords}

\section{Introduction}

Sparse Fast Fourier Transform (sFFT) methods trade among randomness, algebraic structure, transform-length constraints, noise assumptions, and the form of their recovery guarantees. This paper studies a narrow operating model: known very-low sparsity, exact on-grid support, noiseless observations, and design-time selection of the transform length $N$.

The architecture is a single chain. An engineered transform length, built as the exact product of small pairwise-coprime stage primes, gives every stage a transform sized by the sparsity rather than by $N$. Co-prime-increment time shifts then apply a one-sided screen that admits provisional candidate-singletons without proving singleton status, each candidate-singleton is submitted to the total emission rule, which emits at most one global frequency label, and the per-stage label sets are intersected rather than multiplied.

A sound verifier discharges the surviving candidates, and any failure routes to a dense fallback. The verifier's PASS on its window of $k + |\mathcal{G}|$ consecutive residual samples, where $\mathcal{G}$ is the constructed candidate set, is the acceptance \emph{certificate} of the title: the sparse output is emitted only under that certificate, and a PASS implies the emitted spectrum is exact (\Cref{thm:verifier-soundness}). Within that operating model the core result is a conditional sparse-candidate architecture: candidate construction runs in $O(k \log^2 N / \log k)$ in a comparison model without unit-cost decoding reads, tightening to $O(k \log N)$ in an explicitly stated exact label-decoding cost model (\Cref{thm:cost-comparison-model,thm:sparse-path-complexity}), while a standard dense FFT supplies an $O(N \log N)$ fallback whenever the verified sparse path does not complete.

\subsection{Relationship to Prior Work}

\noindent The multi-stage CRT framework adopted here builds on the closed-form noise-tolerant CRT line developed by Xia and collaborators, beginning with the undersampled multi-frequency estimation analysis of \cite{xia2000estimation}, the closed-form noise-tolerant CRT formulation of \cite{wang2010closed}, the improved-performance variant of \cite{xiao2014robust}, the multi-stage noise-tolerant CRT of \cite{xiao2015multistage}, and the truly sub-Nyquist treatment of \cite{xiao2017frequency}. The present work differs from this lineage by (i) targeting exact on-grid noiseless recovery rather than noise-tolerant frequency estimation in additive noise, (ii) integrating a per-stage one-sided co-prime shift phase-consistency screen (with an arithmetic-shift Hankel/Prony rank test available as an optional deterministic escalator) and a deterministic dense FFT fallback, and (iii) operating at very low sparsity (\Cref{rem:engineering-regime}) under an explicit prime-availability hypothesis $L \le \pi([k,ck])$.

\emph{Randomized hashing and subsampling.} For a broad survey of sublinear sparse-FFT methods we refer to Gilbert, Indyk, Iwen \& Schmidt~\cite{gilbert2014sfftsurvey}. Probabilistic approaches, notably Hassanieh, Indyk, Katabi \& Price~\cite{hassanieh2012nearly,hassanieh2012simple}, use randomized subsampling and permutation hashing to achieve sublinear $O(k \log N)$ complexity with high-probability correctness. Indyk, Kapralov, and Price~\cite{ikp2014soda} achieved nearly sample-optimal $O(k \log N)$ complexity, and Kapralov~\cite{kapralov2016stoc} extended the lineage to any constant dimension with nearly-optimal sample complexity. A CRT-based variant, Pawar and Ramchandran's FFAST~\cite{pawar2018ffast}, exploits coprime-factor aliasing on signals of length $N = k \cdot 2^t$ to reach $O(k \log k)$ runtime with sample-optimal $O(k)$ subsampling. These methods have probabilistic ($1-\delta$) correctness with limited worst-case guarantees, and FFAST requires a structural constraint on $N$ that rules out arbitrary DFT lengths.

\emph{Deterministic algebraic and structured-support methods.} The Fourier-learning line traces to Kushilevitz \& Mansour~\cite{kushilevitz1993learning}, whose KM algorithm recovers sparse Fourier spectra over the Boolean cube and seeded the sparse-Fourier program. Iwen~\cite{iwen2013improved} introduced algebraic hashing using coprime modular decimation. Plonka \& Wannenwetsch~\cite{plonka2016numalg,plonka2017jcam} gave deterministic sublinear sparse FFTs for vectors whose support lies in a short interval of length $m < N$: $O(m \log m)$ arithmetic for exactly measured data, and, for real non-negative vectors, $O(m \log m \log(N/m))$ arithmetic on $O(m \log(N/m))$ Fourier samples with no a priori knowledge of the support length. Bittens, Zhang \& Iwen~\cite{bittens2017acha} gave deterministic sublinear-time algorithms for structured Fourier sparsity. These frameworks rely on hash-style bucket isolation or short-support periodization and do not use the co-prime-shift singleton screen plus sound-verifier mechanism studied here.

\emph{Fully discrete deterministic sparse FFTs.} Merhi, Zhang, Iwen \& Christlieb~\cite{merhi2019discrete} widened that line to arbitrary inputs and arbitrary lengths: their deterministic algorithm returns a near-best $s$-term approximation of the DFT in $O(s^2 \log^{11/2} N)$ time for any vector in $\mathbb{C}^N$ and any $N$, with randomized implementations reported alongside it.

Merhi et al.'s guarantee~\cite{merhi2019discrete} is unconditional where the present one is not: their $O(s^2 \log^{11/2} N)$ time bound carries no survival condition and no unit-cost decoding assumption, while the $O(k \log N)$ stated here depends on both a survival-conditioned sparse completion and the declared exact label-decoding cost model. Their algorithm~\cite{merhi2019discrete} also accepts any vector in $\mathbb{C}^N$ for any $N$, where the present construction requires $N$ to be selected at design time as an exact product of pairwise-coprime stage primes. The comparison therefore does not favor the present paper on generality. What the restrictions buy is a sparsity dependence that is linear rather than quadratic, and exact recovery within the noiseless on-grid model in place of a near-best $s$-term approximation.

\emph{Noise-tolerant and multiscale sublinear methods.} The closed-form noise-tolerant CRT line of \cite{xia2000estimation,wang2010closed,xiao2014robust,xiao2015multistage,xiao2017frequency} targets frequency estimation in additive noise rather than exact on-grid noiseless recovery. Christlieb, Lawlor \& Wang~\cite{christlieb2016multiscale} reach the noisy setting from a mechanism close to the one used here: sampling at two equispaced sets offset in time makes the offset appear in the decimated spectrum as a per-frequency modulation, which both detects aliasing and reads an isolated frequency directly. Their noisy extension~\cite{christlieb2016multiscale} performs that read progressively at geometrically spaced shifts and reports $O(k \log k \log(N/k))$ time on average when the noise is not overwhelming. Reading a frequency out of a shifted decimated observation is therefore not new here; what differs is the surrounding contract, since the present paper fixes the transform length by design, treats the shift signature as a one-sided screen rather than a decision, and takes its correctness from a verifier rather than the read.

\emph{Orthogonal measurement models.} A separate body of work varies the measurement model rather than the recovery algebra: compressed sensing recovers sparse signals from random linear projections~\cite{donoho2006compressed}, quantized acquisition pushes this to a single bit per measurement~\cite{boufounos2008onebit}, and sparse-operator identification recovers structured operators from few observations~\cite{heckel2013identification}. These directions are complementary to the present exact, full-precision, on-grid setting and are not used here; the one-bit and compressed-measurement regimes are noted as robustness extensions of the noiseless screen.

\emph{Prior keyed-gating and safety-screened CRT work.} Flouro \& Chadwick~\cite{flouro2026keyed}, ``Keyed Gating for Multi-View Sparse FFT,'' introduced a deterministic two-of-three keyed CRT-gating mechanism that reduced candidate sets from $O(k^2)$ to $O(k)$ and validated frequencies through iterative residue refinement on \mbox{$\sqrt{N}$-sized} decimated views. Flouro \& Chadwick~\cite{flouro2026safety}, ``Safety-Certified CRT Sparse FFT,'' added safety certificates and an adaptive dense-FFT fallback giving $O(N \log N)$ worst-case runtime.

\emph{Distinction of the present paper.} The present paper studies the complementary very-low-sparsity end of the deterministic CRT sparse FFT line (\Cref{rem:engineering-regime}). It replaces \mbox{$\sqrt{N}$-sized} decimations with $\Theta(k)$-sized stage DFTs and uses co-prime time shifts $\{0, 31, 63\}$ for a low-order singleton-consistency screen. Unlike the earlier keyed-gating and safety-screened constructions, which address larger decimation scales and worst-case safety control, this paper isolates the engineered very-low-sparsity fast path and states its complexity conditionally (see \Cref{thm:sparse-path-complexity,thm:global-label-completeness}). The construction is complementary rather than dominant: it does not displace randomized $O(k \log N)$ schemes or \mbox{$\sqrt{N}$-scaled} deterministic schemes at larger $k$, but extends the deterministic CRT family into the very-low-sparsity regime where $k$-sized stages dominate, and does so without weakening the correctness guarantee.

\begin{table*}[t]
\centering
\caption{Comparison of Operating Models and Complexity Guarantees}
\label{tab:comparison}
\small
\setlength{\tabcolsep}{4pt}
\newcommand{\LMPrr}{\raggedright\arraybackslash}
\begin{tabular}{>{\LMPrr}p{2.8cm}>{\LMPrr}p{2.6cm}>{\LMPrr}p{2.7cm}>{\LMPrr}p{2.1cm}>{\LMPrr}p{3.5cm}>{\LMPrr}p{2.2cm}}
\toprule
\textbf{Algorithm} & \textbf{Complexity} & \textbf{Cost model}$^\S$ & \textbf{Correctness Model} & \textbf{Design} & \textbf{Deterministic?} \\
\midrule
\textbf{This paper} & $O(k \log^2\! N/\log k)$ conditional sparse-path$^\dagger$; tightens to $O(k \log N)$ with the root-index oracle & Comparison real-RAM (b) / root-index-oracle real-RAM (a); arithmetic ops & Deterministic, exact$^*$ & Multi-stage CRT with $\Theta(k)$-sized stage DFTs & Deterministic CRT \\
\addlinespace
Keyed Gating~\cite{flouro2026keyed} & $O(\sqrt{N} \log k)$ conditional sparse-path & Not stated & Deterministic$^*$ & 3-view CRT with \mbox{$\sqrt{N}$-sized} FFTs & Deterministic CRT \\
\addlinespace
Safety-Certified CRT~\cite{flouro2026safety} & $O(N \log N)$ worst-case & Not stated & Deterministic & \mbox{$\sqrt{N}$-sized} FFTs + safety certificates & Deterministic CRT \\
\addlinespace
MIT sFFT~\cite{hassanieh2012nearly,hassanieh2012simple} & $O(k \log N)$ & Not stated; time & $1-\delta$ & Random hashing & No \\
\addlinespace
IKP-SODA~\cite{ikp2014soda} & $O(k \log N)$ sample-optimal & Not stated; samples and time & $1-\delta$ & Random hashing + sample optimality & No \\
\addlinespace
FFAST~\cite{pawar2018ffast} & $O(k \log k)$, $N = k \cdot 2^t$ only & Unit-cost RAM I/O; arithmetic ops & $1-\delta$ & Coprime-factor aliasing & No \\
\addlinespace
AAFFT~\cite{iwen2013improved} & $O(k \log k \log N)$ & Not stated; time & Deterministic & Algebraic & Yes \\
\addlinespace
Plonka-Wannenwetsch~\cite{plonka2016numalg,plonka2017jcam} & $O(m \log m)$, support interval of length $m$$^\ddagger$ & Not stated; arithmetic ops & Deterministic & Short-support periodization & Yes \\
\addlinespace
Bittens et al.~\cite{bittens2017acha} & $O(k \log k \log N)$ & Not stated; runtime & Deterministic & Structured Fourier sparsity & Yes \\
\addlinespace
Merhi et al.\ DSFT~\cite{merhi2019discrete} & $O(s^2 \log^{11/2}\! N)$, any $N$ & Not stated; time & Deterministic, near-best $s$-term approx. & Periodized-Gaussian conversion of USSFT methods to fully discrete SFT & Yes \\
\bottomrule
\end{tabular}

\vspace{0.3em}
{\footnotesize\raggedright
$^*$Always returns the exact spectrum within the noiseless, on-grid, at-most-$k$ model.\par
\vspace{0.2em}
$^\dagger$Candidate-construction cost only (models (a), (b) of \Cref{def:label-decoding-model}, the no-oracle figure being \Cref{thm:cost-comparison-model}); $O(N \log N)$ worst case via the dense fallback when the exact-product envelope is not used or sparse completion fails. The verified path adds the sound verifier, for $O(k \log N) + O(k^2)$ arithmetic on $O(k \log N)$ samples (\Cref{prop:verifier-cost}).\par
\vspace{0.2em}
$^\ddagger$Structurally different rather than dominated: Plonka-Wannenwetsch needs no survival condition, but assumes the support lies in one interval of length $m$; this paper takes arbitrary scattered at-most-$k$ supports at the price of that condition and the dense fallback.\par
\vspace{0.2em}
$^\S$The machine model each row's own source declares for its own bound, with the unit that bound counts. Only this paper and FFAST name a model at all, and the units are not uniform, so the complexity column is not a single ruler.\par
}
\end{table*}

\subsection{Contribution and Scope}

The architecture itself is narrated once, in the chain description above; the list below states the claims, each carried by a stated result:
\begin{enumerate}
\item Engineered divisor-compatible stage construction: the transform length is design-selected as $N = \prod_{\ell=1}^{L} p_\ell$ with pairwise-coprime stage primes in $[k, ck]$, giving $L = O(\log_k N)$ stages of $\Theta(k)$-sized FFTs in place of square-root-scale decimations (\Cref{def:transform-family}).
\item One-sided candidate-singleton screen from the shift signature: the co-prime-increment shifts give every true singleton exact magnitude invariance and exact modular phase progression, a fast necessary screen whose passes are provisional (\Cref{thm:phase-consistency-certificate}).
\item Proved deterministic $O(k)$ candidate count: each candidate-singleton emits at most one global frequency label (possibly none) by the total rule of \Cref{def:total-recovery}, and the per-stage label sets are intersected, so the $O(k)$ candidate count holds for every input (\Cref{thm:global-label-count}).
\item Candidate-construction cost in two explicit decoding models: $O(k \log^2 N / \log k)$ arithmetic in the comparison real-RAM model, tightening to $O(k \log N)$ in the exact label-decoding (root-index oracle) model of \Cref{def:label-decoding-model}, which costs the per-bin root-index read explicitly and uses no precomputed root table (\Cref{thm:cost-comparison-model,thm:sparse-path-complexity}).
\item Survival-conditioned containment and verified sparse completion: under all-stage genuine-singleton survival the constructed candidate set contains $\mathcal{S}$, so the verified sparse path returns $X$ exactly (\Cref{thm:global-label-completeness}); its verified cost is $O(k \log N) + O(k^2)$ arithmetic on $O(k \log N)$ samples in the oracle model (\Cref{prop:verifier-cost}).
\item Verified fallback architecture with proved soundness: exact recovery is made unconditional within the noiseless, on-grid, at-most-$k$ model by a sound final verifier constructed and proved sound here (\Cref{def:verifier}, \Cref{thm:verifier-soundness}); sparse candidates that fail any operational condition or final verification are routed to a dense FFT, giving $O(N \log N)$ worst-case runtime (\Cref{thm:hybrid-correctness}).
\end{enumerate}

The formal results are stated for every sparsity $k \ge 3$ under the engineered exact-product model of \Cref{def:transform-family} (fixed absolute constant $c$, pairwise-distinct primes $p_\ell \in [k, ck]$, $L \le \pi([k,ck])$), and every $O(\cdot)$ cost expression is a genuine asymptotic bound over this unbounded family. The intended engineering regime is the very-low-sparsity range $3 \le k < 20$ (\Cref{rem:engineering-regime}), where $\Theta(k)$-scaled stages remain smaller than \mbox{$\Theta(\sqrt{N})$-scaled} decimations; instances at $k = 20$ and beyond are formally admissible and are shown as illustrations of the regime boundary, where the \mbox{$\sqrt{N}$-scaled} schemes amortize more effectively (see \Cref{tab:comparison}). A single-frequency worked trace at $N = 46{,}189$, $k=8$ appears in \Cref{ex:n1024_k8}.

\subsection{Roadmap}

\noindent \Cref{sec:problem} fixes the operating envelope for sparse completion and \Cref{sec:preliminaries} the notation; \Cref{sec:algorithm-overview} then presents the multi-stage algorithm with a single-frequency walkthrough (\Cref{ex:n1024_k8}). \Cref{sec:screens} develops the one-sided singleton screen and its core lemmas, and \Cref{sec:main-theorem} assembles the main results: the sound final verifier, the two explicit-model conditional sparse-path cost theorems (the $O(k \log N)$ oracle-model bound and the $O(k \log^2 N / \log k)$ comparison-real-RAM bound), and the deterministic $O(k)$ global-label candidate count. \Cref{sec:computational} covers parameter selection and memory, \Cref{sec:adversarial} exercises the dense fallback on adversarial supports, and \Cref{sec:limitations} states the limitations before the conclusion. \Cref{app:tuple-negative} records the superseded residue-tuple analysis as a negative result, \Cref{app:screen-analysis} collects the deferred screen analysis, \Cref{app:detailed-example} gives the full worked trace, \Cref{app:cost-models} collects the further cost-model discussion, and \Cref{app:keyed-multiview} records the keyed multi-view extension as design history and future work.

\noindent Unless stated otherwise, every unqualified $O(k \log N)$ sparse-path arithmetic figure in this paper is stated in the root-index-oracle model of \Cref{def:label-decoding-model}(a); comparison-model figures are explicitly marked.

\section{Problem Statement and Operating Envelope}
\label{sec:problem}

\noindent This section presents the formal model of the paper. Two objects are design-selected and stated here as definitions (the signal model and the engineered exact-product transform family), and one is a premise of the conditional results and stated here as a condition. Later sections cite these three objects rather than restating them.

\begin{definition}[Signal Model]
\label{def:signal-model}
The recovery setting is exact, noiseless, exact-arithmetic recovery of an at-most-$k$-sparse on-grid spectrum. The signal is $x \in \mathbb{C}^N$ whose $N$-point DFT $X$ satisfies $|\mathrm{supp}(X)| \le k$ with support frequencies on the integer grid $f \in \{0, 1, \ldots, N-1\}$; the sparsity $k$ is known and $k \ge 3$; observations are noiseless and all arithmetic is exact. All formal results of this paper are stated for every such $k$; the very-low-sparsity range in which the construction is intended to be fielded is prose context (\Cref{rem:engineering-regime}), never a hypothesis of a theorem.
\end{definition}

\begin{definition}[Engineered Exact-Product Transform Family (Operating Model)]
\label{def:transform-family}
The transform length is \emph{design-selected}, not given: it is the exact product
\[
  N \;=\; Q_L \;=\; \prod_{\ell=1}^{L} p_\ell
\]
of $L$ pairwise distinct stage primes $p_\ell \in [k, ck]$, fixed at design time to realize that product rather than selected at runtime, where $c \ge 2$ is an absolute constant fixed independently of $k$ and $N$ (the widened interval $c > 2$ is used whenever $\pi([k,2k]) < L$; $c \le 4$ suffices for every configuration of \Cref{tab:prime-availability}), subject to the stage-count bound $2 \le L \le \pi([k,ck])$, where $\pi(\cdot)$ counts primes in the interval; the lower bound is a standing hypothesis of this model, excluding the degenerate single-stage instance $L = 1$ (where $N = p_\ell \le ck$ is itself of order $k$ and a multi-stage construction is vacuous), and every result stated under this model inherits it. Since $p_\ell \ge k \ge 3$, every stage prime is odd, so $N$ is odd and $\gcd(32, N) = 1$, as required by the B\'ezout label inversion of \Cref{lem:global-label-recovery}(3). Distinct primes in $[k,ck]$ are automatically pairwise coprime and each divides $N$, so the decimation of every stage is divisor-compatible. The exact-product relation gives the stage count $\log_{ck} N \le L \le \log_k N$, i.e.\ $L = O(\log_k N)$, and per-stage transform sizes $m_\ell = \Theta(p_\ell) = \Theta(k)$. Because $L \ge 2$, the length is a product of at least two distinct primes $\ge k$, so every admissible instance satisfies
\[
  N \;\ge\; k(k+1) \;>\; k^2 ,
\]
whence $k < \sqrt{N}$ and $k^2 = O(N)$; this is the fact the runtime accounting of \Cref{thm:hybrid-correctness} uses.

Acquisition uses the fixed co-prime time-domain shift triple $\{0, \Delta_1, \Delta_1 + \Delta_2\}$ with $\gcd(\Delta_1, \Delta_2) = 1$, specifically $\{0, 31, 63\}$, fixed deterministically per the keyed gating construction. Candidate-singleton screening and label emission follow the full exact emission predicate of \Cref{def:exact-screen} (the occupied equal-magnitude pre-screen, then the total label-emission rule: root membership and class consistency); normalized MAD, unwrapped-phase regression, and frequency consistency are its implementation surrogates (\Cref{rem:impl-surrogates}) and are not part of this model. Likewise the engineered mixed-radix batching factor $2^t$ (in this paper $t = 6$, so an implementation length $64 p_\ell$) is a cache-alignment convenience outside the formal model (\Cref{rem:64k-implementation}); the formal per-stage transform is the $p_\ell$-point DFT.
\end{definition}

\begin{condition}[All-Stage Genuine-Singleton Survival (Fast-Path Condition)]
\label{cond:all-stage-survival}
Every support frequency $f \in \mathcal{S} = \mathrm{supp}(X)$ is the unique support frequency in its class $f \bmod p_\ell$ at every stage $\ell = 1, \ldots, L$, so its bin holds a genuine singleton there and passes the screen. This is a premise of the conditional complexity and containment results, not a property enjoyed by every input, and it is never assumed for correctness. Under it the candidate set contains $\mathcal{S}$ (\Cref{thm:global-label-completeness}) and the verified sparse path returns $X$ exactly; a support tone that collides at some stage is not guaranteed in the candidate set and is recovered by the dense fallback. (For true tones this survival implies that every stage emits their true labels along the sparse path, since an isolated bin always satisfies the full exact emission predicate of \Cref{def:exact-screen}; the converse is not certified, the screen being one-sided, since a collided bin may also pass, cf.\ \Cref{prop:two-tone-phantom}.) A violation of this condition is not detectable from the data: the one-sided screen cannot certify it stage-locally, and it surfaces only operationally, as a stage emitting no class-consistent global label, a budget overflow, or a final-verifier FAIL.

\emph{Nonempty support is part of the condition.} The zero spectrum $X = 0$ is admissible under \Cref{def:signal-model} and satisfies the per-frequency clause above vacuously, yet it yields no occupied bins and no labels, so every stage emits no class-consistent global label and \Cref{alg:multistage_oklogn} routes to the dense fallback, which returns $X = 0$ exactly. This condition therefore additionally requires $1 \le |\mathcal{S}|$, and every completion claim conditioned on it concerns nonempty supports.

\emph{The condition is satisfiable at every $k$, so the fast path is not asymptotically vacuous.} If the support lies in an integer interval of length at most $\min_\ell p_\ell$ inside $[0,N)$, then any two distinct support frequencies differ by less than every stage prime, so no stage prime divides their difference and every tone is alone in its class at every stage. Such supports exist at every admissible instance, since $\min_\ell p_\ell \ge k$ and $N > k^2$. This short-support family is sufficient for the condition, not necessary: arbitrary scattered supports also satisfy it whenever they happen to avoid all-stage collisions, which is why the condition is strictly weaker than the short-support hypothesis of the deterministic interval methods (\Cref{tab:comparison}).
\end{condition}

\begin{remark}[Intended Engineering Regime]
\label{rem:engineering-regime}
The results of this paper are stated asymptotically for every $k \ge 3$ under \Cref{def:signal-model,def:transform-family}: for fixed absolute $c$ the prime supply $\pi([k,ck])$ grows without bound with $k$ (\Cref{lem:prime-existence}), so the admissible engineered lengths $N = \prod_{\ell=1}^{L} p_\ell$ are unbounded and every $O(\cdot)$ expression is a genuine asymptotic bound over this unbounded family.

The \emph{intended engineering regime}, by contrast, is the very-low-sparsity range $3 \le k < 20$, and the reason is structural: because the stage primes are proportional to the sparsity ($p_\ell \in [k, ck]$), the per-stage per-bin collision probability under a uniform-random-support heuristic stays near a nonzero constant ($1 - (1-1/p_\ell)^{k-1} \approx 1 - e^{-(k-1)/p_\ell}$, between $1 - e^{-1/c}$ and $1 - e^{-1}$) rather than vanishing at larger $k$, so all-stage genuine-singleton survival of all $k$ tones grows increasingly demanding as $k$ rises and the sparse path's advantage concentrates at very low sparsity (\S\ref{sec:limitations}). The regime bound is an engineering target, never a theorem hypothesis: instances at $k \ge 20$ are formally admissible and appear in \Cref{tab:prime-availability} as boundary illustrations.
\end{remark}

\noindent\emph{Operating envelope.} We refer to \Cref{def:signal-model,def:transform-family} together with \Cref{cond:all-stage-survival} as the \emph{operating envelope for sparse completion}: the full set of conditions under which the verified sparse path completes rather than falling back. Cost and count alone require less: the $O(k \log N)$ candidate-construction cost is stated in the cost model of \Cref{def:label-decoding-model} and holds under the engineered exact-product transform family alone; the survival condition is what additionally places $\mathcal{S}$ inside the $O(k)$ candidate set, and hence what keeps the algorithm on the sparse path instead of the dense fallback. The $O(k)$ candidate count is not an envelope condition but the proved \Cref{thm:global-label-count}, holding for every input; exactness of an accepted sparse output is the proved verifier of \Cref{thm:verifier-soundness} rather than an assumption. Any input that violates a complexity condition is routed to the $O(N \log N)$ dense FFT fallback, which is exact on every input of \Cref{def:signal-model}; the violation itself is not certified stage-locally (the screen is one-sided), but it surfaces operationally as a stage emitting no class-consistent global label, a budget overflow, or a final-verifier FAIL, each of which triggers the routing.

\noindent\emph{Scope.}
This paper analyzes exact on-grid, noiseless recovery for known sparsity $k \ge 3$ under the model of \Cref{def:signal-model}. The screened sparse path is engineered for the very-low-sparsity regime, its intended engineering regime being $3 \le k < 20$ (\Cref{rem:engineering-regime}), where $k$-scaled mixed-radix stages remain smaller than competing \mbox{$\sqrt{N}$-scaled} decimations; illustrations at $k \ge 20$ in this paper are admissible instances marking the boundary of that regime and are labeled as such. Larger $k$ values are covered by the same formal theorems, but the sparse-path advantage is expected to degrade because all-stage genuine-singleton survival becomes less frequent; such inputs remain exact through verifier-gated acceptance or dense fallback.

The formal sparse-path theorem (\Cref{thm:sampling-bluestein}) applies when the transform length is chosen as $N = Q_L$ exactly, the engineered product of the stage primes; non-product transform lengths are handled by the dense-fallback option of \Cref{subsec:sampling-operator}. The prime-availability constraint is $L \le \pi([k,ck])$ for an explicitly selected constant $c \ge 2$, which with the exact-product model $N = \prod_\ell p_\ell$ gives the stage relations $\log_{ck} N \le L \le \log_k N$. The special case $c = 2$ gives the tightest stage-size envelope, while larger $c$ is used for very small $k$ when $[k,2k]$ does not contain enough primes. Off-grid frequencies, additive noise, and out-of-envelope sparsity ranges are admitted only as discussed extensions, with the dense FFT fallback handling out-of-envelope inputs at $O(N\log N)$ cost.

\noindent\emph{Problem.} Within this scope, the recovery task is stated as follows: given time-domain samples $x[n]$ for a length-$N$ signal whose frequency-domain representation $X$ has at most $k \ll N$ non-zero on-grid components, recover $X$ exactly. The standard FFT~\cite{cooley1965algorithm} achieves $O(N \log N)$ complexity unconditionally, while sparse FFT algorithms~\cite{hassanieh2012nearly,hassanieh2012simple,iwen2013improved} aim for sublinear complexity by exploiting sparsity. The recovery problem we address is the noiseless, on-grid version of this sparse FFT task: the support of $X$ is unknown but lies on the integer DFT grid, and the goal is exact reconstruction with provable correctness rather than high-probability success. The information-theoretic floor that any support-identifying procedure must respect is recorded next.

\begin{remark}[Information-Theoretic Counting Bound]
\label{rem:lower-bound}
Identifying which of the $\binom{N}{k}$ possible $k$-sparse supports a signal belongs to requires conveying at least $\log_2\binom{N}{k}$ bits of information~\cite{cover2006elements}. By Stirling's approximation,
\[
\log_2\binom{N}{k} \approx k \log_2(N/k) + k,
\]
so $\Omega(k \log(N/k))$ bits of input must be read by any support-identifying procedure. This is an information-theoretic bound, not a computational lower bound: converting it to a lower bound on operations requires fixing a specific machine model (e.g., algebraic decision tree, comparison-based, real-RAM with $O(1)$-time arithmetic), which is not asserted here.
\end{remark}

\subsection{Detection vs Avoidance}

A distinction exists between two approaches to handling collisions in sparse FFT algorithms:

\emph{Collision Avoidance:} Traditional sparse FFT algorithms use hashing schemes to probabilistically prevent multiple frequencies from mapping to the same bin. The Fredman-Koml\'os lower bound of $\Omega(k \log(N/k)/\log k)$~\cite{fredman1984separating} applies to specific collision-avoidance / hash-table models; the model-free counting floor that no procedure escapes is \Cref{rem:lower-bound}.

\emph{Engineered Divisor-Compatible Reconstruction (This Work):} The present framework does not claim to bypass the Fredman-Koml\'os bound as a generic lower bound. Instead, the algorithm operates in a different model: collisions are allowed to occur in the decimated spectra, and a bin failing the multi-dimensional geometric tests grounded in co-prime time shifts (\Cref{def:coprime-shifts}) emits no label, the input reaching the dense-FFT fallback only through the stage- and chain-level routing rule of \Cref{thm:hybrid-correctness}; a measure-zero screen-passing phantom (\Cref{prop:two-tone-phantom}) is admitted to the candidate set and, if it survives the global-label intersection and is used in the reconstructed candidate, is rejected there by the residual verifier (\Cref{rem:one-sided-normative}). The lower bound therefore does not constrain the present construction on inputs whose sparse attempt completes (every stage emits a class-consistent global label and both budgets hold); on inputs whose sparse attempt aborts, the algorithm bills $O(N \log N)$ work and no claim is made about beating any sparse-FFT lower bound.

\noindent The screen tests use exact magnitude equality (implementation surrogate: normalized MAD) and exact ratio/root-membership conditions (\Cref{def:exact-screen}; surrogate: unwrapped phase regression~\cite{tribolet1977new}), with an optional frequency-consistency cross-check (Quinn/Jacobsen estimators~\cite{quinn1994estimating,jacobsen2007fast}) as an implementation surrogate outside the formal screen; their behavior is summarized as: true singletons pass; generic collisions fail the screen and emit no label (optionally after higher-order Hankel/Prony escalation, \Cref{lem:arith-shift-hankel-rank}), the dense-FFT fallback being a stage- and chain-level decision (\Cref{thm:hybrid-correctness}); but a measure-zero coefficient slice of collisions passes the screen at every $m \ge 2$ (\Cref{prop:two-tone-phantom,prop:higher-order-phantom}), and such screen-passing phantoms are handled as in \Cref{rem:one-sided-normative}.

The construction underlying the screen is the use of co-prime time shifts $\{0, \Delta_1, \Delta_1+\Delta_2\}$ with $\gcd(\Delta_1, \Delta_2) = 1$. By B\'ezout's identity, two distinct on-grid frequencies cannot simultaneously satisfy both increment congruences, so the two individual frequencies of a colliding pair cannot each equal one tone that is consistent at all three shifts. This is a componentwise, pairwise statement only: it does not prevent the phasor sum of the colliding tones from mimicking, on a measure-zero coefficient slice, the three observations of a genuine singleton (the explicit two-tone phantom of \Cref{prop:two-tone-phantom}), so it is not a two-tone false-accept impossibility. Combined with the occupied equal-magnitude pre-screen and the total label-emission rule (\Cref{def:exact-screen}), this gives a one-sided screen (\Cref{rem:one-sided-normative}): every true singleton is accepted and a passing bin is only a candidate.

\needspace{12\baselineskip}\section{Preliminaries and Notation}
\label{sec:preliminaries}

\subsection{Discrete Fourier Transform}

\begin{definition}[Discrete Fourier Transform~\cite{oppenheim2010discrete}]
\label{def:dft}
The Discrete Fourier Transform (DFT) of a length-$N$ time-domain signal $x[n]$ is defined as:
\begin{equation}
X[f] = \sum_{n=0}^{N-1} x[n] \cdot e^{-j2\pi fn/N}, \quad f = 0, 1, \ldots, N-1
\end{equation}
where $X[f]$ represents the frequency-domain coefficient at frequency $f$, and $j = \sqrt{-1}$. The corresponding inverse transform reconstructs $x[n]$ from its coefficients with the conjugate ($+$ sign) kernel, consistent with the unified sign convention used throughout:
\begin{equation}
x[n] = \frac{1}{N} \sum_{f=0}^{N-1} X[f] \cdot e^{+j2\pi fn/N}, \quad n = 0, 1, \ldots, N-1 .
\end{equation}
\end{definition}

\begin{definition}[k-Sparse Spectrum]
A frequency-domain spectrum $X$ is \emph{k-sparse} if the number of non-zero frequency components is at most $k$:
\begin{equation}
|\text{supp}(X)| \leq k, \quad \text{where } \text{supp}(X) = \{f : X[f] \neq 0\}
\end{equation}
\end{definition}

\begin{definition}[On-Grid Frequencies]
Frequencies $f \in \{0, 1, \ldots, N-1\}$ are \emph{on-grid}. Off-grid frequencies (fractional frequencies between grid points) require separate analysis due to spectral leakage effects.
\end{definition}

\begin{definition}[Co-prime Time Shifts]
\label{def:coprime-shifts}
Three time shifts $\{0, \Delta_1, \Delta_1 + \Delta_2\}$ are \emph{co-prime} if $\gcd(\Delta_1, \Delta_2) = 1$. In this case B\'ezout's identity (over the integers) guarantees that there exist integers $a, b \in \mathbb{Z}$ with $a \Delta_1 + b \Delta_2 = 1$. As an angular-congruence corollary, the only solution to both $(\nu_1 - \nu_2) \cdot \Delta_1 \equiv 0 \pmod{2\pi}$ and $(\nu_1 - \nu_2) \cdot \Delta_2 \equiv 0 \pmod{2\pi}$ is $\nu_1 = \nu_2$~\cite{vaidyanathan2001theory}. Consequently, no single colliding pair of distinct frequencies can individually produce phase alignment at all three shifts. This is a componentwise, pairwise statement only; the shift design yields a one-sided screen (\Cref{rem:one-sided-normative}).
\end{definition}

\begin{definition}[Shifted Decimated Spectrum and Screen Statistics (Notation)]
\label{def:multidim-isolation}
For a stage modulus $M \mid N$ and time shift $s$, the \emph{shifted decimated spectrum} is $X_M[r;s] \coloneqq \tfrac{M}{N} \sum_{f \equiv r \,(\bmod\, M)} X[f]\, e^{+j 2\pi s f/N}$, the size-$M$ DFT of the $s$-shifted decimated samples $x[(s + m N/M) \bmod N]$, $m = 0, \ldots, M-1$ (derived in \Cref{subsec:sampling-operator}); at $s = 0$ it carries the aliasing prefactor, $X_M[r;0] = (M/N)\, X_M[r]$ in the unshifted notation of \Cref{def:m-decimated-dft}.

For a bin $r$ read at the operative shift triple $\{0, 31, 63\}$ ($\Delta_1 = 31$, $\Delta_2 = 32$), the shift-valued observations are $z_0, z_{31}, z_{63}$ with $z_\sigma = X_M[r;\sigma]$; the normalized magnitude range is $D_{\max}(r) = (\max_\sigma |z_\sigma| - \min_\sigma |z_\sigma|)/\max_\sigma |z_\sigma|$, well-defined on occupied bins ($\max_\sigma |z_\sigma| > 0$); the per-step phase ratios are $u_{31} = z_{31}/z_0$ and $u_{32} = z_{63}/z_{31}$; and the B\'ezout label phasor is $w = u_{32}\,u_{31}^{-1}$. The subscripts on $z$ name shift values and those on $u$ name the increments $\Delta_1, \Delta_2$. This definition is notation only and assigns no screening status: the formal screen is \Cref{def:exact-screen}, and the implementation surrogates are \Cref{def:cv-singleton-test} and \Cref{rem:impl-surrogates}.
\end{definition}

\begin{definition}[Phase Unwrapping]
\label{def:phase-unwrapping}
(Implementation surrogate.) Phase unwrapping and the associated linear-regression $R^2$ test are a finite-precision implementation surrogate; the formal phase test is the modular phase-ratio condition of \Cref{lem:global-label-recovery} (the per-step ratios $u_{31} = z_{31}/z_0$ and $u_{32} = z_{63}/z_{31}$ are exact unit-modulus rotations), not cumulative unwrapping. For wrapped phases $\theta_s \in [-\pi, \pi]$, the unwrapped phase $\phi_s$ removes $2\pi$ discontinuities via cumulative adjustment~\cite{tribolet1977new}: $\phi_0 = \theta_0$, and for $s > 0$:
\begin{equation}
\phi_s = \phi_{s-1} + \text{wrap}(\theta_s - \theta_{s-1})
\end{equation}
where $\text{wrap}(\delta) = (\delta + \pi) \bmod 2\pi - \pi$ maps differences to $[-\pi, \pi]$.
\end{definition}

\subsection{Key Notation}

\Cref{tab:notation} summarizes the notation used throughout the paper.

\begin{table}[!t]
\centering
\caption{Key Notation and Symbols}
\label{tab:notation}
\begin{tabular}{l p{0.72\linewidth}}
\toprule
\textbf{Symbol} & \textbf{Definition} \\
\midrule
$N$ & Signal length (DFT size) \\
$k$ & Sparsity level (number of non-zero frequencies) \\
$x[n]$ & Time-domain samples, $n = 0, \ldots, N-1$ \\
$X[f]$ & Frequency-domain DFT coefficients, $f = 0, \ldots, N-1$ \\
$L$ & Number of CRT stages; in the exact-product model $N = \prod_{\ell=1}^{L} p_\ell$, so $L = O(\log_k N)$ \\
$p_i$ & Prime modulus for stage $i$, $p_i \in [k, ck]$ for $c \ge 2$ (special case $c=2$) \\
$m_i$ & Implementation batching length for stage $i$, $m_i = p_i \times 2^6$ (cache-aligned mixed-radix batching); the formal stage transform is the $p_i$-point DFT and the residue modulus stays $p_i$ \\
$M$ & CRT modulus product; exact-product model $M = \prod_{i=1}^L p_i = N$ (the generalized condition $M \ge N$ applies to the dense-fallback / extension case) \\
$r_i$ & Residue of frequency $f$ modulo $p_i$, $r_i = f \bmod p_i$ \\
$\mathcal{R}_\ell$ & Candidate-singleton screened bin set of stage $\ell$ (\Cref{alg:multistage_oklogn}); the operative engine reads it into per-bin global labels $\mathcal{G}_\ell$, and the residue-tuple view over the $\mathcal{R}_\ell$ is superseded (\Cref{app:tuple-negative}) \\
$\mathrm{nMAD}(r)$ & Normalized median absolute deviation of bin-$r$ magnitudes across shifts (implementation surrogate for exact magnitude equality) \\
$\text{MAD}(A)$ & Median absolute deviation of set $A$ \\
$S$ & Number of time shifts per stage (typically $S=3$) \\
\addlinespace
$\mathcal{S}$ & True support set $\mathcal{S} = \{f : X[f] \neq 0\}$, $|\mathcal{S}| \le k$ \\
$\mathcal{G}_\ell$ & Per-stage global-label set: the global frequency labels emitted by the candidate-singletons of stage $\ell$ (\Cref{def:total-recovery}), with $|\mathcal{G}_\ell| \le p_\ell$ \\
$\mathcal{G}$ & Accumulated global-label (candidate) set $\mathcal{G} = \bigcap_{\ell} \mathcal{G}_\ell$; the operative candidate set, of proved size $|\mathcal{G}| = O(k)$ (\Cref{thm:global-label-count}) \\
$\hat{X}[f]$ & Recovered (reconstructed) spectrum returned by the algorithm \\
$\hat{S}$ & Recovered support $\hat{S} = \{f : \hat{X}[f] \neq 0\}$ \\
\addlinespace
$X_M[r;s]$ & Shifted decimated spectrum, $X_M[r;s] = \tfrac{M}{N}\sum_{f \equiv r (\bmod M)} X[f]\, e^{+j2\pi s f/N}$; equals $X_{p_\ell}^{(\sigma)}[r]$ for $M = p_\ell$, $s = \sigma$ (\Cref{subsec:sampling-operator}) \\
$z_\sigma$ & Shifted bin observation $z_\sigma = X_{p_\ell}^{(\sigma)}[r]$ at shift $\sigma \in \{0, 31, 63\}$ (\Cref{lem:global-label-recovery}) \\
$u_{31}, u_{32}$ & Per-step phase ratios $u_{31} = z_{31}/z_0$, $u_{32} = z_{63}/z_{31}$ (subscripts name the increments $\Delta_1 = 31$, $\Delta_2 = 32$) \\
$w$ & B\'ezout label phasor $w = u_{32}\,u_{31}^{-1} = \omega^{g}$ (\Cref{def:total-recovery}) \\
$g$ & Global frequency label $g = \log_\omega w \in [0, N)$ emitted by a candidate-singleton (\Cref{def:total-recovery}) \\
$D_{\max}(r)$ & Normalized magnitude range of bin $r$ across shifts; exact screen statistic, $D_{\max}(r) = 0$ for a true singleton (\Cref{def:exact-screen}) \\
$\tau_{cv}$ & $\mathrm{nMAD}$ acceptance tolerance, an implementer-chosen $\tau_{cv}\in(0,1)$ (\Cref{def:cv-singleton-test}) \\
\addlinespace
$Q_\ell$ & Accumulated modulus product $Q_\ell = \prod_{i \le \ell} p_i$, $Q_0 = 1$; exact-product model $Q_L = N$ \\
$V$ & Sound final verifier $V(x,\hat{X}) \in \{\mathrm{PASS}, \mathrm{FAIL}\}$ (\Cref{def:verifier}); distinct from the view count $V \in \{3,4\}$ used only in \Cref{app:keyed-multiview} \\
$W$ & Final-verifier window size, a scalar: $W = k + |\mathcal{G}|$ consecutive residual samples read by $V$ (\Cref{def:verifier}) \\
$\Psi$ & Node matrix of \Cref{thm:verifier-soundness}, $\Psi_{i\ell} = \omega_\ell^{\,n_0+i}$ ($W \times s$) \\
$\omega$ & Primitive $N$th root of unity, $\omega = e^{+j 2\pi / N}$; the subscripted $\omega_\ell = e^{+j 2\pi f_\ell/N}$ are the Vandermonde nodes of $\Psi$; angular frequencies are written $\nu$ \\
\bottomrule
\end{tabular}
\end{table}

\section{Algorithm Overview}
\label{sec:algorithm-overview}

\noindent\emph{Terminology (normative).} This paper uses three strictly disjoint concepts and never overloads them. (i) A \emph{stage} ($\ell = 1, \ldots, L$) is one CRT decimation level with prime modulus $p_\ell \in [k, ck]$ for an explicitly selected constant $c \ge 2$ (special case $c=2$); the primary algorithm runs $L$ stages sequentially. (ii) A \emph{shift} is one of the $S = 3$ time-domain offsets $\{0, \Delta_1, \Delta_1 + \Delta_2\}$ (e.g., $\{0, 31, 63\}$) used within a single stage to perturb the input phase for candidate-singleton screening; shifts are never called views. (iii) A \emph{view}, used exclusively in the adaptive-moduli framework (\Cref{app:keyed-multiview}), refers to one of $V \in \{3, 4\}$ independent decimated FFT computations within a single stage, each with a distinct modulus $M_i$; views are never called shifts.

\noindent The primary multi-stage architecture therefore has $L$ stages, each with $S = 3$ shifts and $V = 1$ view, yielding $L \times S = 3L$ FFT computations (\Cref{alg:multistage_oklogn}). Multi-view constructions with $V \in \{3, 4\}$ appear only in the appendix-only keyed multi-view extension (\Cref{app:keyed-multiview}), which is a heuristic extension outside the formal model; they contribute no FFTs to the primary path. All downstream uses of the words ``stage'', ``shift'', and ``view'' follow these definitions and never conflate them.

\noindent\emph{Further normative terms.} The per-bin screening vocabulary is fixed by \Cref{def:exact-screen} as three distinct predicates, never conflated: the \emph{occupied equal-magnitude pre-screen} (its condition~1, whose passes are the candidate-singletons), the total label-emission rule (root membership plus class consistency, executed as the total rule $\textsc{Recover}$ of \Cref{def:total-recovery}), and their conjunction, the full exact emission predicate, whose one-sidedness is \Cref{thm:phase-consistency-certificate}; the historical name \emph{co-prime shift phase-consistency screen} refers to that conjunction, and \emph{full exact screen} is short for it, never condition~1 alone. An \emph{arithmetic-shift Hankel/Prony rank test} (\Cref{lem:arith-shift-hankel-rank}) is an optional deterministic escalator that requires a separate sample acquisition on arithmetic-progression shifts. The operative per-stage object is the \emph{global-label set} $\mathcal{G}_\ell$ (\Cref{subsec:global-label-engine}, \Cref{def:total-recovery}): each candidate-singleton emits at most one full frequency label and the stages are combined by intersection, $\mathcal{G} = \bigcap_\ell \mathcal{G}_\ell$. The earlier \emph{accumulated candidate set} is a superseded residue-set view, retained only as the negative result of \Cref{app:tuple-negative}.

\begin{table}[!t]
  \centering
  \caption{Achievable Engineered Transform Lengths by Sparsity and Stage Count\label{tab:prime-availability}}
  \resizebox{\linewidth}{!}{%
  \begin{tabular}{r l r r r}
    \toprule
    $k$ & $[k,ck]$ primes & $N$ at $L{=}4$ & $N$ at $L{=}5$ & $N$ at $L{=}6$ \\
    \midrule
      8 & $\{11,13,17,19,23,29\}$ & $46{,}189$ & $1{,}062{,}347$ & $30{,}808{,}063$ \\
     10 & $\{11,13,17,19,23,29\}$ & $46{,}189$ & $1{,}062{,}347$ & $30{,}808{,}063$ \\
     15 & $\{17,19,23,29,31,37\}$ & $215{,}441$ & $6{,}678{,}671$ & $\approx 2.5{\cdot}10^{8}$ \\
     20 & $\{23,29,31,37,41,43\}$ & $765{,}049$ & $\approx 3.1{\cdot}10^{7}$ & $\approx 1.3{\cdot}10^{9}$ \\
     50 & $\{53,59,61,67,71,73\}$ & $12{,}780{,}049$ & $\approx 9.1{\cdot}10^{8}$ & $\approx 6.6{\cdot}10^{10}$ \\
    \bottomrule
  \end{tabular}%
  }
\end{table}

\subsection{Algorithm Structure}

\noindent\textbf{Input:} Time samples $x[n]$, $n = 0,\ldots,N-1$, of an engineered length $N = \prod_{\ell=1}^{L} p_\ell$ given with its stage primes (\Cref{def:transform-family}); sparsity $k$\\
\textbf{Requirements:} $2 \le L \le \pi([k,ck])$ stages with the exact-product model $N = \prod_{\ell=1}^{L} p_\ell$ (so $L = O(\log_k N)$) and primes $p_i \in [k, ck]$ for $c \ge 2$ (special case $c=2$), implementation batching lengths $m_i = p_i \times 2^t$ with $t=6$ (mixed-radix)\\
Key Property: The exact-product relation $N = \prod_{i=1}^L p_i$ (exact-product model, with $k^L \leq N \leq (ck)^L$) gives divisor-compatible stage decimations and makes global frequency labels unique over $[0,N)$; the operative path emits full labels per bin and intersects the per-stage label sets, rather than reconstructing residue tuples\\
Prime availability. Intuitively, a constant-factor window widens with the sparsity: by the Prime Number Theorem the range $[k, ck]$ contains approximately $(c-1)k/\log k$ primes, so the stage budget grows without bound with $k$. For the algorithm to succeed with $L = O(\log_k N)$ stages, we require $L \leq \pi([k, ck])$. \Cref{tab:prime-availability} shows achievable engineered transform lengths $N = Q_L = \prod_{\ell=1}^{L} p_\ell$ for representative sparsity levels $k$ at stage counts $L \in \{4, 5, 6\}$. Each entry is the exact product of $L$ pairwise-coprime stage primes drawn from $[k, ck]$; the engineered-acquisition pipeline selects $N$ at design time to match these values, satisfying the divisor-compatibility envelope of \Cref{thm:sampling-bluestein}, and non-product $N$ is handled by the dense-fallback option of \Cref{subsec:sampling-operator}. Rows $k = 20$ and $k = 50$ of the table are admissible instances of \Cref{def:transform-family}; they lie at and beyond the boundary of the intended engineering regime $3 \le k < 20$ (\Cref{rem:engineering-regime}) and are shown to illustrate that boundary.

\begin{lemma}[Prime Availability via Dusart Bounds]
\label{lem:prime-existence}
Fix a constant $c \ge 2$. For every fixed $c$, the interval $[k, ck]$ contains
\[
\pi([k,ck]) \;=\; \bigl((c-1) + o(1)\bigr)\,\frac{k}{\ln k}
\]
primes as $k \to \infty$ (Prime Number Theorem). Explicitly, the special case $c=2$ gives
\[
\pi([k,2k]) \;\geq\; \frac{0.44\,k}{\ln k} \qquad \text{for all } k \ge 3,
\]
proved below from Dusart's explicit bounds for $k \ge 300$ and verified numerically for each $3 \le k < 300$ (e.g., $\pi([20,40]) = 4$ primes; widening to $[k, 3k]$ yields additional primes whenever the $[k, 2k]$ count is insufficient for the required $L$ stages, cf.\ \Cref{tab:prime-availability}).
\end{lemma}

\begin{proof}
We use two of Dusart's explicit prime-counting bounds~\cite{dusart2010estimates}:
\begin{align}
\pi(x) &\ge \frac{x}{\ln x}\left(1 + \frac{1}{\ln x}\right) \quad \text{for } x \ge 599, \\
\pi(x) &\le \frac{x}{\ln x}\left(1 + \frac{1.2762}{\ln x}\right) \quad \text{for } x > 1.
\end{align}
For $k \ge 300$ we have $2k \ge 600 \ge 599$, so applying the lower bound at $2k$ and the upper bound at $k$,
\begin{multline}
\pi([k,2k]) \ge \pi(2k) - \pi(k) \\
\ge \frac{2k}{\ln 2k}\left(1 + \frac{1}{\ln 2k}\right) - \frac{k}{\ln k}\left(1 + \frac{1.2762}{\ln k}\right) = \rho(k)\,\frac{k}{\ln k},
\end{multline}
where $\rho(k) = \frac{2\ln k}{\ln 2k}\bigl(1 + \frac{1}{\ln 2k}\bigr) - 1 - \frac{1.2762}{\ln k}$. For $k \ge 300$, the increasing factor $\ln k/\ln 2k = (1 + \ln 2/\ln k)^{-1} \ge (1 + \ln 2/\ln 300)^{-1} > 0.89$ and the decreasing term $1.2762/\ln k \le 1.2762/\ln 300 < 0.23$ give $\rho(k) \ge 2(0.89) - 1 - 0.23 = 0.55 > 0.44$, so $\pi([k,2k]) \ge 0.44\,k/\ln k$ for all $k \ge 300$. For each $3 \le k < 300$ the bound is confirmed by direct enumeration ($k=20$: $\pi([20,40]) = 4$ primes $\{23, 29, 31, 37\}$ against $0.44 \cdot 20/\ln 20 \approx 2.94$; $k=50$: $\pi([50,100]) = 10$ primes against $\approx 5.62$). Thus for every $k \ge 3$ the interval $[k,2k]$ alone supports up to $L \leq 0.44k/\ln k$ stages; when the required $L$ exceeds the $[k,2k]$ count at small $k$, widened intervals $[k, ck]$ with $c > 2$ supply the remaining primes (\Cref{tab:prime-availability}).
\end{proof}

For smaller $k$ values where this constraint is violated, the interval is widened: the stage primes are drawn from $[k, ck]$ with $c$ chosen so that $(c-1)k/\ln k \geq L$.
The tighter range $[k, 2k]$ minimizes the per-stage transform sizes; its availability is exactly what \Cref{lem:prime-existence} quantifies ($\pi([k,2k]) \ge 0.44\,k/\ln k$ for all $k \ge 3$, by Dusart bounds for $k \ge 300$ and finite enumeration for $3 \le k < 300$), and whenever the required $L$ exceeds that count the widened interval $[k,ck]$ supplies the remaining primes.\\
\emph{Mixed-Radix Optimization:} implementation batching length $m_i = p_i \times 2^6 \in [64k, 64ck]$ (the special case $c = 2$ gives $[64k, 128k]$; for widened intervals the upper end scales with $c$, e.g.\ $m = 64 \cdot 31 = 1984 > 128k = 1920$ at $k = 15$, $c = 3$) balances prime factor discrimination ($p_i$) with power-of-2 batching ($2^6 = 64$, a cache-alignment design rationale; no measurements are reported).

\emph{Fundamental Constraint on N:} With the exact-product model $N = \prod_{\ell=1}^{L} p_\ell$ and stage primes $k \le p_\ell \le ck$, the transform length and stage count obey
\begin{equation}
\begin{aligned}
k^{L} &\;\le\; N \;\le\; (ck)^{L}, \\
\text{equivalently}\quad \log_{ck} N &\;\le\; L \;\le\; \log_{k} N .
\end{aligned}
\end{equation}
For fixed $c$ this gives $L = \Theta(\log_k N)$. The binding availability constraint is that $L$ must not exceed the number of pairwise-coprime primes in $[k, ck]$, i.e. $L \le \pi([k, ck]) \approx (c-1)k/\ln k$; this caps the realizable engineered lengths for a given $k$ rather than imposing a closed-form bound on $N$.

\begin{algorithm*}[!t]
\caption{Multi-Stage Screened Sparse FFT (Global-Label Form)}\label{alg:multistage_oklogn}
\textbf{Input:} signal $x[n]$ of an engineered length $N = \prod_{\ell=1}^{L} p_\ell$, given together with its stage primes $p_1, \ldots, p_L \in [k,ck]$, $2 \le L \le \pi([k,ck])$ (\Cref{def:transform-family}); sparsity $k$; CV threshold $\tau_{cv}$ (the $\mathrm{nMAD}$ acceptance tolerance of \Cref{def:cv-singleton-test}; an implementer-chosen $\tau_{cv}\in(0,1)$, not a value this paper calibrates). A length that is not such an exact product lies outside the operating model and is dispatched directly to DENSE-FALLBACK (\Cref{thm:hybrid-correctness}).\\
\textbf{Output:} the exact spectrum $X$, returned either by verified sparse reconstruction or by the dense FFT fallback (dispatched when a stage emits no class-consistent global label, a budget guard fires, or final verification fails).
\begin{enumerate}[topsep=0pt,itemsep=1pt,parsep=0pt]
  \item Take the $L$ stage primes $p_1, \ldots, p_L \in [k, ck]$ of the engineered length ($N = \prod_{\ell=1}^{L} p_\ell$, so $L = O(\log_k N)$; \Cref{lem:prime-existence}, \Cref{def:transform-family}); fix co-prime shifts $\{\sigma_1, \sigma_2, \sigma_3\} = \{0, 31, 63\}$ with $\gcd(\Delta_1, \Delta_2) = \gcd(31, 32) = 1$. Initialize the global-label candidate set $\mathcal{G} \leftarrow \text{(undefined; set to }\mathcal{G}_1\text{ at the first stage)}$. No accumulated modulus is carried: the engine works with full frequency labels, not residue tuples.
  \item For each stage $\ell = 1, \ldots, L$: compute the three shift-FFTs $X_{p_\ell}[\cdot; \sigma_s]$ associated with modulus $p_\ell$ (length $m_\ell = \Theta(p_\ell)$) via the sampling operator (\Cref{subsec:sampling-operator}); for each bin $r \in \{0, \ldots, p_\ell - 1\}$, mark $r$ as a candidate-singleton via the occupied equal-magnitude pre-screen (\Cref{def:exact-screen}(1)), with $\mathrm{nMAD}(r) < \tau_{cv}$ as the fast computational surrogate, subject to its all-observations-nonzero guard before the Step~3 ratios are formed (\Cref{def:cv-singleton-test}). Let $\mathcal{R}_\ell$ be the candidate-singleton screened bin set.
  \item Total label-emission rule and intersection. For each candidate-singleton $r \in \mathcal{R}_\ell$, apply the total label-emission rule (\Cref{def:total-recovery}) to the three shifted observations $z_\sigma = X_{p_\ell}^{(\sigma)}[r]$, emitting at most one global label $g = \textsc{Recover}(z_0, z_{31}, z_{63})$ and possibly none; collect $\mathcal{G}_\ell = \{g : r \in \mathcal{R}_\ell,\ r\text{ emits a label}\}$ and accumulate by intersection, $\mathcal{G} \leftarrow \mathcal{G} \cap \mathcal{G}_\ell$ (initialized $\mathcal{G} \leftarrow \mathcal{G}_1$ at the first stage). If a stage emits no class-consistent global label, or the deterministic budget is exceeded ($|\mathcal{G}_\ell| > ck$ at that stage, or $|\mathcal{G}| > k$ once the stage chain is complete, i.e.\ after stage $\ell = L$), dispatch DENSE-FALLBACK and return the standard $O(N \log N)$ FFT spectrum of $x$ before verification.
  \item For each candidate frequency $g \in \mathcal{G}$, recover its coefficient from the zero-shift decimated observation at any stage $\ell$: $\hat{X}[g] = (N/p_\ell)\, X_{p_\ell}^{(0)}[g \bmod p_\ell]$, the $(N/p_\ell)$ factor being the reciprocal of the $(p_\ell/N)$ aliasing prefactor (\Cref{subsec:sampling-operator}, \Cref{lem:global-label-recovery}(4)); at $\sigma=0$ the analysis phase is unity, so no de-rotation is required. (A nonzero-shift observation may equivalently be used after de-rotation by the conjugate phase $e^{-j 2\pi g \sigma_v / N}$.)
  \item Final verification. Apply the sound final verifier $V$ of \Cref{def:verifier} to the sparse candidate $\hat{X}$ supported on $\mathcal{G}$, reading $k + |\mathcal{G}|$ consecutive residual samples, the residual-support bound being $|\mathrm{supp}(x-\hat{x})| \le k + |\mathcal{G}|$ (proved sound in \Cref{thm:verifier-soundness}). If $V$ returns PASS, return $\hat{X}$ as the sparse output; otherwise compute and return the dense $O(N \log N)$ FFT spectrum of $x$.
\end{enumerate}
Note (the budget guards are inactive in the exact model). In the exact model the budget guards of Step~3 are inactive: every bin passing the occupied equal-magnitude pre-screen is occupied, and a decimated at-most-$k$-sparse spectrum has at most $k$ occupied bins per stage (\Cref{lem:decimation-preserves-sparsity}), so $|\mathcal{G}_\ell| \le k < ck$ and $|\mathcal{G}| \le \min_\ell |\mathcal{G}_\ell| \le k$. Thus neither guard fires. They are included only to describe the fielded surrogate implementation, whose finite-precision surrogates (\Cref{rem:impl-surrogates}) can admit bins the exact pre-screen rejects; no cost, count, or correctness claim of this paper depends on their firing.
\end{algorithm*}

\noindent\emph{Algorithm 1: semantics of the steps.} The float states the procedure; the semantics its steps rest on are collected here, each a fact proved elsewhere rather than a step of the procedure. Step~1 carries no accumulated modulus and no residue state: the global-label engine never forms residue tuples, so there is no residue to complete and no CRT reassembly to trigger. Chain completion is the stage index reaching $L$, which the engineered exact-product length fixes at design time (\Cref{def:transform-family}); it certifies nothing about genuine-singleton survival, which is not detectable from the data (\Cref{cond:all-stage-survival}). In Step~2 the formal per-stage transform is the $p_\ell$-point DFT; an implementation may batch or cache-align it at a mixed-radix length $p_\ell \cdot 2^6$ without changing the residue modulus $p_\ell$ (\Cref{rem:64k-implementation}).

\noindent In Step~3 the label of a screened bin is read by the B\'ezout phase rule: $u_{31} = z_{31}/z_0$, $u_{32} = z_{63}/z_{31}$, $w = u_{32}\,u_{31}^{-1}$, and $g$ is the unique index in $\{0,\dots,N-1\}$ with $w = \omega^{g}$, $\omega = e^{+j2\pi/N}$ (equivalently $g = \log_\omega w$), emitted exactly when $w$ is an $N$th root of unity and $g \equiv r \pmod{p_\ell}$; the read is executed by the costed rounded-guess-then-exact-verify procedure of \Cref{def:label-decoding-model}, with no precomputed root table. Because each screened bin emits at most one label without assuming the bin is a true singleton, the per-stage bound $|\mathcal{G}_\ell| \le |\mathcal{R}_\ell| \le p_\ell \le ck$ holds deterministically (\Cref{thm:global-label-count}), and no residue tuples are formed, so the multiplicative Cartesian blow-up of the screen-only rule, whose tuple count can reach $2^L$, cannot occur (\Cref{rem:no-product-blowup}); \Cref{lem:global-label-recovery} supplies completeness for genuine singletons. A screen-inconsistent bin emits no label, a normal per-bin event, and any support tone thereby dropped is recovered by the dense fallback (\Cref{thm:hybrid-correctness}). The budget tests of that step are detected symptoms of a survival violation, never a certificate of one.

\noindent A Step~5 FAIL says only that the candidate is not the exact spectrum, either because a surviving phantom label entered it (\Cref{rem:phantom-labels}) or because a true tone was omitted from $\mathcal{G}$ (\Cref{thm:global-label-completeness}). The $O(k)$ candidate count is proved (\Cref{thm:global-label-count}), but the recovered labels are not all guaranteed true, so every sparse output is provisional until verification passes.

Concretely, the prime sets fix the maximum number of stages: $k=10$ admits the four primes $\{11,13,17,19\}$ in $[10,20]$, so $L \le 4$ and the largest engineered length is $N = 11\cdot 13\cdot 17\cdot 19 = 46{,}189$; the boundary case $k=20$ admits $\{23,29,31,37\}$, giving $N = 765{,}049$ at $L=4$, and $k=50$ admits ten primes in $[50,100]$, so $L \le 10$ and much larger engineered lengths are available (both admissible instances at and beyond the regime boundary of \Cref{rem:engineering-regime}; see \Cref{tab:prime-availability}). The analytic estimates indicate the intended operating point $k \in [8, 20)$ with signal lengths $N \le 10^6$; no implementation measurements are reported in this paper. The tighter range $[k, 2k]$ minimizes FFT overhead while maintaining sufficient stage coverage.

The algorithm chains $L \approx \log_k N$ CRT stages with $\Theta(k)$-sized DFTs per stage. \Cref{fig:algorithm_flowchart} visualizes the procedure; the five steps below specify the procedure from signal input through exact spectrum output, with each step's contribution to the overall $O(k \log N)$ bound (stated in the cost model of \Cref{def:label-decoding-model}) identified in the complexity analysis of \S\ref{sec:main-theorem}.

\paragraph{Modulus and transform length.} The CRT modulus at stage $\ell$ is the prime $p_\ell$, and the formal stage transform of \Cref{subsec:sampling-operator} is the $p_\ell$-point DFT of the decimated operator; the recovered class index lies in $\mathbb{Z}/p_\ell\mathbb{Z}$. The factor $2^6 = 64$ in an implementation length $64 p_\ell$ is an optional mixed-radix batching convenience (cache-friendly butterflies); it does not change the residue modulus, which remains $p_\ell$, and any $m_\ell = \Theta(p_\ell)$ choice yields the same complexity (\Cref{rem:64k-implementation}).

\noindent\emph{Operative global-label flow.} The operative candidate path forms no residue tuples: each candidate-singleton of stage $\ell$ emits at most one global frequency label $g \in [0,N)$ (possibly none) by the total rule of \Cref{def:total-recovery}, read in closed form from the three shifts by the B\'ezout phase rule of \Cref{lem:global-label-recovery}. The stage label set $\mathcal{G}_\ell$ therefore holds full frequencies, not residues, and the stages are combined by intersection $\mathcal{G} = \bigcap_\ell \mathcal{G}_\ell$, which can only shrink the set; $\mathcal{G}$ is never materialized as an $N$-element list, and $|\mathcal{G}| \le \min_\ell |\mathcal{G}_\ell| \le p_1 \le ck = O(k)$ (\Cref{thm:global-label-count}), with any budget overflow ($|\mathcal{G}_\ell| > ck$ at some stage, or $|\mathcal{G}| > k$ on the accumulated set once the stage chain is complete) detected and routed to dense fallback. The $O(k \log N)$ figure bounds only the cost of candidate construction (label recovery and intersection over $L = O(\log_k N)$ stages); exactness is verifier-gated, a constructed candidate being returned only on a PASS of the sound final verifier (\Cref{thm:verifier-soundness}, reading $k + |\mathcal{G}|$ consecutive residual samples, \Cref{def:verifier}) and otherwise routed to the dense $O(N \log N)$ FFT fallback (\Cref{thm:hybrid-correctness}). The multiplicative tuple-rectangle accounting $\prod_j |\mathcal{R}_j|$ is not the operative count; it is retained only as the non-operative negative result of \Cref{app:tuple-negative}.

\Needspace*{5\baselineskip}
\begin{example}[Single-Frequency Global-Label Walkthrough: $N = 46{,}189$, $k = 8$]\label{ex:n1024_k8}
We illustrate the screened sparse path on the engineered transform length $N = 46{,}189 = 11 \cdot 13 \cdot 17 \cdot 19$ at sparsity $k = 8$ by tracing one representative support frequency, $f = 67$, through the operative global-label engine (\Cref{subsec:global-label-engine}). \emph{Scope of this walkthrough.} This is an illustrative single-frequency trace intended to expose the per-stage B\'ezout phase read and the per-stage label intersection in concrete form; it is not a full $k = 8$ support example. A full $k = 8$ trace would additionally exhibit the cross-frequency collision combinatorics and is not required for the correctness or complexity arguments of \Cref{thm:sparse-path-complexity,thm:hybrid-correctness}. This $N$ satisfies the divisor-compatibility envelope of \Cref{thm:sampling-bluestein} since each stage prime in $\{11, 13, 17, 19\}$ divides $N$ and $\prod_\ell p_\ell = N$ exactly.

\noindent\emph{Stage 1: Parameter Selection.}
$L = 4$ stages, with primes $\{p_1, p_2, p_3, p_4\} = \{11, 13, 17, 19\}$ (range widened from $[k,2k]=[8,16]$ to $[k,3k]$ to obtain 4 primes, per \Cref{lem:prime-existence}); co-prime shifts $\{\sigma_1, \sigma_2, \sigma_3\} = \{0, \Delta_1, \Delta_1 + \Delta_2\} = \{0, 31, 63\}$ with $\Delta_1 = 31$, $\Delta_2 = 32$, $\gcd(\Delta_1, \Delta_2) = \gcd(31, 32) = 1$ (the relevant B\'ezout co-primality is between the increments $\Delta_1, \Delta_2$, not between the shift values). The stage chain requires exactly $L = 4$ stages, determined by the product equality $Q_4 = \prod_{\ell=1}^{4} p_\ell = N$ rather than by the generic logarithmic stage-count bound $\lceil \log_8 N \rceil = 6$; under engineered divisor compatibility, the exact-product model $N = \prod_{\ell=1}^{L} p_\ell$ pins $L$ exactly to the number of selected primes.

\noindent\emph{Global-Label Recovery for $f = 67$.}
\Cref{tab:n1024_k8_residues_f67} shows the per-stage screened bin $r_\ell = 67 \bmod p_\ell$ and the single global label $g$ that each stage recovers and emits.

\begin{table}[!t]
\centering
\caption{Per-Stage Screened Bins and Recovered Global Label}
\label{tab:n1024_k8_residues_f67}
\small
\begin{tabular}{cccccc}
\toprule
$f$ & $\bmod\,11$ & $\bmod\,13$ & $\bmod\,17$ & $\bmod\,19$ & label $g$ \\
\midrule
$67$ & $1$ & $2$ & $16$ & $10$ & $67$ \\
\bottomrule
\end{tabular}
\end{table}

At each screened bin the three shifted observations are read by the B\'ezout phase rule (\Cref{lem:global-label-recovery}): with $u_{31} = z_{31}/z_0$ and $u_{32} = z_{63}/z_{31}$, the global label is the exact root-of-unity index $g = \log_\omega w$ of $w = u_{32}\,u_{31}^{-1} = \omega^{g}$ ($\omega = e^{+j2\pi/N}$), accepted into $\mathcal{G}_\ell$ if and only if $g \equiv r_\ell \pmod{p_\ell}$. Every stage emits $g = 67$, so the intersection is $\mathcal{G} = \bigcap_\ell \mathcal{G}_\ell = \{67\}$ with $|\mathcal{G}| = 1 = O(k)$ (\Cref{thm:global-label-count}); no residue tuple is formed. The coefficient is the single zero-shift estimate $\hat{X}[67] = (N/p_\ell)\, X_{p_\ell}^{(0)}[r_\ell] = 1.0 + 0.0j$ (the $\sigma=0$ analysis phase is unity), and the sound final verifier (\Cref{thm:verifier-soundness}) certifies $\hat{x} = x$ exactly. The full per-stage numerical trace, including the phase reads and the verifier PASS, is in \Cref{app:detailed-example}.

\noindent\textbf{\emph{Candidate-construction cost for this example:}} on the $k$-based proxy with the $S = 3$ shift-FFTs counted, roughly $S \cdot k L \log_2 k = 3 \cdot 8 \cdot 4 \cdot 3 = 288$ operations; this is a coarse oracle-model candidate-construction proxy only (model (a) of \Cref{def:label-decoding-model}), not a total verified sparse-path cost, since the actual per-stage transform lengths are $p_\ell \in \{11, 13, 17, 19\}$ rather than $k$ (or $m_\ell = 64 p_\ell$ in the engineered batching, which raises the constant). The verification of \Cref{prop:verifier-cost} is additional: $O(k^2)$ in model (a) of \Cref{def:label-decoding-model}, $O(k^2 + k \log N)$ in model (b).
\end{example}

\let\LMPsavedneedspace\needspace
\renewcommand{\needspace}[1]{}%
\subsection{Modular Decimation}
\let\needspace\LMPsavedneedspace

\begin{definition}[M-Decimated DFT]
\label{def:m-decimated-dft}
For any integer $M \geq 2$, the $M$-decimated (aliased) spectrum is:
\begin{equation}
X_M[r] = \sum_{\substack{\ell \ge 0 \\ 0 \le r + \ell M \le N-1}} X[r + \ell M], \quad r = 0, \ldots, M-1.
\end{equation}
This definition coincides with the usual form when $M \mid N$, and otherwise sums over all indices in $[0,N)$ congruent to $r \bmod M$.
\end{definition}

\textbf{\emph{Key property.}} Intuitively, decimation is a sorting step: frequency $f$ maps to residue $r = f \bmod M$, one bin per frequency and never more, so a sparse spectrum stays sparse in every decimated view. Bin $r$ is a \emph{singleton} if exactly one frequency $f$ with $f \equiv r \pmod{M}$ is non-zero.

\begin{lemma}[Decimation Preserves Sparsity]
\label{lem:decimation-preserves-sparsity}
If $X$ is $k$-sparse, each decimated spectrum has at most $k$ non-zero bins.
\end{lemma}

\begin{proof}
Each frequency maps to one residue. At most $k$ frequencies implies at most $k$ non-zero residues.
\end{proof}

\subsection{Decimated Sampling Operator under Divisor Compatibility}\label{subsec:sampling-operator}

The screened sparse path's $O(k \log N)$ complexity (its label-read arithmetic charged in the cost model of \Cref{def:label-decoding-model}) is sample-and-arithmetic complexity simultaneously, achieved under a \emph{divisor-compatibility envelope}: the input length $N$ is selected so that every stage prime $p_\ell$ divides $N$. Equivalently, $N = Q_L = \prod_{\ell=1}^{L} p_\ell$ exactly: the engineered transform length is the product of the stage primes. Equivalently, an engineered acquisition pipeline selects the acquisition or processing length so that the divisor-compatibility envelope is satisfied as an equality. (Multiples of $\prod_\ell p_\ell$ are also covered if $N \ge Q_L$, but the prime-availability envelope of \Cref{def:transform-family} forces $Q_L = N$ when both $Q_L \mid N$ and $Q_L \ge N$ hold simultaneously.) Within this envelope, the per-stage measurement reduces to standard modular decimation and the textbook aliasing identity holds exactly.

\paragraph{Decimated sampling operator.} For stage $\ell$ with prime $p_\ell \mid N$ and decimation step $D_\ell = N / p_\ell$, the time-shift-$\sigma$ sampling operator $S_{p_\ell}^{(\sigma)} : \mathbb{C}^N \to \mathbb{C}^{p_\ell}$ reads
\[
  y_\ell^{(\sigma)}[m] \;:=\; x\bigl[(\sigma + m D_\ell) \bmod N\bigr], \qquad m = 0, 1, \ldots, p_\ell - 1,
\]
that is, $p_\ell$ samples of $x$ at uniformly-spaced positions $\sigma, \sigma + D_\ell, \ldots, \sigma + (p_\ell - 1) D_\ell$ modulo $N$. The size-$p_\ell$ DFT of $y_\ell^{(\sigma)}$ is the $p_\ell$-decimated frequency view of the shifted spectrum:
\begin{equation*}
\begin{split}
  X_{p_\ell}^{(\sigma)}[r] \;&=\; \frac{p_\ell}{N} \sum_{\substack{f \\ f \,\equiv\, r \,(\bmod\, p_\ell)}} X[f] \, e^{+j 2\pi \sigma f / N}, \\
  &\qquad\qquad r = 0, 1, \ldots, p_\ell - 1.
\end{split}
\end{equation*}
This is the standard modular-decimation aliasing identity, valid exactly when $p_\ell \mid N$. Notationally, for $M = p_\ell$ and $s = \sigma$ the three surface forms denote the same object: $X_M[r;s] = X_{p_\ell}[r;\sigma] = X_{p_\ell}^{(\sigma)}[r]$ (\Cref{def:multidim-isolation}, \Cref{lem:global-label-recovery}). The $e^{+j 2\pi \sigma f/N}$ phase sign follows directly from the forward DFT of \Cref{def:dft}: substituting $y_\ell^{(\sigma)}[m] = x[(\sigma + mD_\ell)\bmod N]$ into the size-$p_\ell$ analysis sum and expanding $x$ by the inverse DFT ($+$ exponent) leaves the surviving residue-class coefficients carrying the time-shift phase $e^{+j 2\pi \sigma f/N}$. The coefficient is recovered cleanly at $\sigma=0$, where this phase is unity: $X[f] = (N/p_\ell)\, X_{p_\ell}^{(0)}[f \bmod p_\ell]$ on a singleton bin; for $\sigma \ne 0$ the de-rotation that undoes the phase is the conjugate $e^{-j 2\pi \sigma f/N}$.

\paragraph{Cost.} Per stage and per shift, $S_{p_\ell}^{(\sigma)}$ reads $p_\ell = O(k)$ samples and computes a size-$p_\ell$ DFT in $O(p_\ell \log p_\ell) = O(k \log k)$ arithmetic operations. With three co-prime shifts per stage and $L = O(\log_k N)$ stages, total sample reads and arithmetic operations are
\[
  3 L \cdot O(p_\ell \log p_\ell) \;=\; O\bigl(\log_k N \cdot k \log k\bigr) \;=\; O(k \log N).
\]
This is the screened sparse path's sample-and-arithmetic cost under divisor compatibility, valid for every admissible instance of \Cref{def:transform-family}; the per-bin label-read arithmetic is charged separately in \Cref{def:label-decoding-model}. Because the operating model fixes $L \ge 2$ (\Cref{def:transform-family}), the count $k \log_2 N$ lies strictly below $N$: every admissible instance has $N \ge k(k+1)$, and since $N - k \log_2 N$ increases in $N$ for $N > k/\ln 2$ (and $k(k+1) > k/\ln 2$), it suffices to check $N = k(k+1)$, where the claim reduces to $k + 1 > \log_2(k^2 + k)$, true for every $k \ge 3$ ($4 > \log_2 12 \approx 3.59$ at $k = 3$, with increasing gap). The excluded single-stage instance is exactly where the comparison can reverse (at $L = 1$, $k = 4$, $N = p_1 = 5$ one has $4 \log_2 5 > 5$), which is why the model excludes it.

\noindent\emph{Informally:} because each stage prime divides $N$, a stage view costs one $\Theta(k)$-point FFT rather than anything of size $N$, and there are only $O(\log_k N)$ stages, so the whole candidate construction is charged at $k$-scale and not at $N$-scale. The theorem states that accounting, and states equally that it buys a candidate set, not an answer.

\begin{theorem}[Sample-and-arithmetic cost of candidate construction under divisor compatibility]\label{thm:sampling-bluestein}
Within the divisor-compatibility envelope (every stage prime $p_\ell$ divides $N$), and in the exact label-decoding (root-index oracle) model of \Cref{def:label-decoding-model}(a), the screened sparse path of \Cref{alg:multistage_oklogn} either constructs a candidate set $\mathcal{G}$ of deterministically bounded size $|\mathcal{G}| = O(k)$ (\Cref{thm:global-label-count}) within $O(k \log N)$ time-domain samples and $O(k \log N)$ arithmetic operations, or triggers the dense fallback at Step~3 after spending no more than that sparse-attempt cost. In the comparison real-RAM model of \Cref{def:label-decoding-model}(b) the same construction performs $O(k \log^2 N / \log k)$ arithmetic operations on the same $O(k \log N)$ samples, with the same either/or reading (\Cref{thm:cost-comparison-model}). Containment of the support is conditional, not automatic: under all-stage genuine-singleton survival of the support (each true tone is the unique true tone in its bin at every stage), \Cref{thm:global-label-completeness} gives $\mathcal{S} \subseteq \mathcal{G}$; a true tone that collides at some stage is not guaranteed in $\mathcal{G}$. This bounds the cost of candidate construction only; it does not by itself recover the spectrum exactly, because the per-stage screen is one-sided and a measure-zero screen-passing phantom (\Cref{prop:two-tone-phantom}) may be admitted. Exactness of the returned spectrum is verifier-gated: an accepted candidate is screened by the sound final verifier of \Cref{thm:verifier-soundness}, which adds $O(k^2)$ arithmetic in model~(a) and $O(k^2 + k \log N)$ in model~(b) (\Cref{prop:verifier-cost}), and an input that fails verification is routed to the dense $O(N \log N)$ fallback (\Cref{thm:hybrid-correctness}).
\end{theorem}

\begin{proof}
By construction the sampling operator $S_{p_\ell}^{(\sigma)}$ reads exactly $p_\ell$ time-domain entries; for $L$ stages and $3$ shifts per stage (the shifts $\{0,31,63\}$), total sample reads are $3 \sum_{\ell=1}^{L} p_\ell \le 3 L c k = O(k \log N / \log k) = O(k \log N)$, using $L = O(\log N / \log k)$ and constant $c$. The aliasing identity above gives $X_{p_\ell}^{(\sigma)}[r]$ as a size-$p_\ell$ DFT of $y_\ell^{(\sigma)}$ in $O(p_\ell \log p_\ell)$ time. Candidate-singleton screening (\Cref{thm:phase-consistency-certificate}; the optional escalation of \Cref{lem:arith-shift-hankel-rank} is not a step of \Cref{alg:multistage_oklogn} and carries no charge here, \Cref{rem:collision-escalation}), total per-bin global-label recovery (\Cref{def:total-recovery}, the B\'ezout phase read of \Cref{lem:global-label-recovery}, costed per bin by the decoding procedure of \Cref{def:label-decoding-model}(a)), and per-stage intersection into the global-label set $\mathcal{G} = \bigcap_\ell \mathcal{G}_\ell$ (\Cref{thm:global-label-count}) each contribute $O(k \log k)$ per executed stage, whether or not a label is emitted; the per-label coefficient is read in closed form at zero shift as $\hat{X}[g] = \tfrac{N}{p_\ell} z_0$ (\Cref{lem:global-label-recovery}(4)). Summing across $L = O(\log_k N)$ stages yields $O(k \log N)$ both for sample reads and for the arithmetic of candidate construction. Every charge is an input-independent per-executed-stage bound, so if Step~3 dispatches the dense fallback at some stage $\ell_0 \le L$, only stages $1, \dots, \ell_0$ have executed and the spend is a partial sum of the same per-stage charges, bounded by the full-chain totals above; the exactness of an accepted candidate is established separately by the final verifier of \Cref{thm:verifier-soundness}, not by this cost bound.
\end{proof}

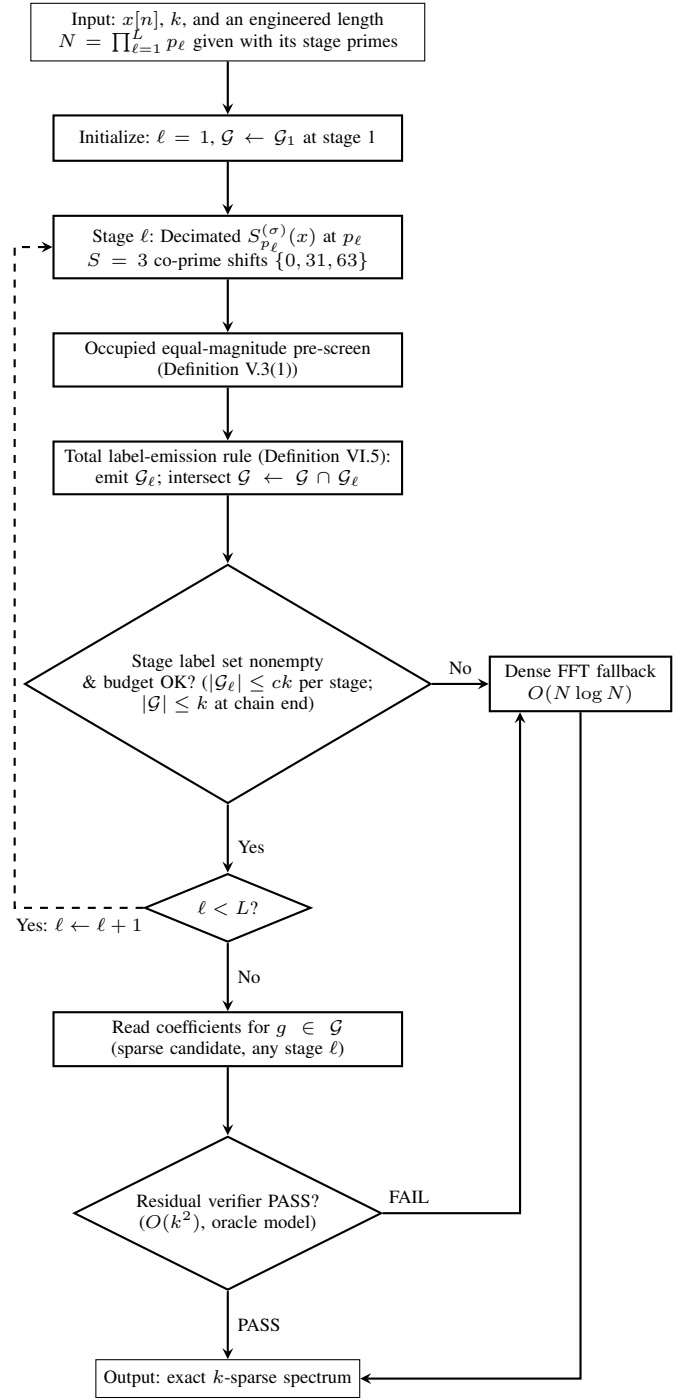
\begin{figure}[!t]
\centering
\begin{tikzpicture}[
  node distance=0.7cm,
  process/.style={rectangle, draw=black, thick, text width=4.4cm, minimum height=0.6cm, inner sep=3pt, align=center, font=\scriptsize},
  decision/.style={diamond, draw=black, thick, minimum width=2.2cm, minimum height=0.8cm, inner sep=3pt, align=center, aspect=2, font=\scriptsize},
  data/.style={rectangle, draw=black, minimum width=2.6cm, minimum height=0.5cm, inner sep=3pt, align=center, font=\scriptsize},
  fallback/.style={rectangle, draw=black, thick, minimum width=2.4cm, minimum height=0.6cm, inner sep=3pt, align=center, font=\scriptsize},
  arrow/.style={->, thick, >=stealth},
  loop/.style={->, thick, >=stealth, dashed}
]
\node[data, text width=5.0cm] (input) {Input: $x[n]$, $k$, and an engineered length\\ $N = \prod_{\ell=1}^{L} p_\ell$ given with its stage primes};
\node[process, below=of input] (init) {Initialize: $\ell = 1$, $\mathcal{G} \leftarrow \mathcal{G}_1$ at stage~1};
\node[process, below=of init] (stage) {Stage $\ell$: Decimated $S_{p_\ell}^{(\sigma)}(x)$ at $p_\ell$\\$S=3$ co-prime shifts $\{0, 31, 63\}$};
\node[process, below=of stage] (cert) {Occupied equal-magnitude pre-screen\\(\Cref{def:exact-screen}(1))};
\node[process, below=of cert] (update) {Total label-emission rule (\Cref{def:total-recovery}):\\emit $\mathcal{G}_\ell$; intersect $\mathcal{G} \leftarrow \mathcal{G} \cap \mathcal{G}_\ell$};
\node[decision, aspect=1.7, inner sep=1pt, below=0.9cm of update] (passq) {Stage label set nonempty\\\&\ budget OK? ($|\mathcal{G}_\ell| \le ck$ per stage;\\$|\mathcal{G}| \le k$ at chain end)};
\node[fallback, right=0.75cm of passq] (dense) {Dense FFT fallback\\$O(N \log N)$};
\node[decision, below=0.9cm of passq] (donelq) {$\ell < L$?};
\node[process, below=0.9cm of donelq] (recon) {Read coefficients for $g \in \mathcal{G}$\\(sparse candidate, any stage $\ell$)};
\node[decision, below=0.9cm of recon] (verify) {Residual verifier PASS?\\($O(k^2)$, oracle model)};
\node[data, below=0.9cm of verify] (output) {Output: exact $k$-sparse spectrum};
\coordinate (loopchan) at ([xshift=-5mm]stage.west);
\draw[arrow] (input) -- (init);
\draw[arrow] (init) -- (stage);
\draw[arrow] (stage) -- (cert);
\draw[arrow] (cert) -- (update);
\draw[arrow] (update) -- (passq);
\draw[arrow] (passq) -- node[above, font=\scriptsize] {No} (dense);
\draw[arrow] (passq) -- node[right, pos=0.62, font=\scriptsize] {Yes} (donelq);
\draw[arrow] (donelq) -- node[right, font=\scriptsize] {No} (recon);
\draw[loop] (donelq.west) -- node[midway, below, font=\scriptsize] {Yes: $\ell \leftarrow \ell+1$} (donelq.west -| loopchan) -- (loopchan) -- (stage.west);
\draw[arrow] (recon) -- (verify);
\draw[arrow] (verify) -- node[right, font=\scriptsize] {PASS} (output);
\draw[arrow] (verify.east) -- node[above, font=\scriptsize] {FAIL} ++(0.7,0) -| ([xshift=-8mm]dense.south);
\draw[arrow] (dense) |- (output);
\end{tikzpicture}
\caption{Screened Multi-Stage Sparse FFT Flow with Global-Label Engine and Fallback}
\label{fig:algorithm_flowchart}
\end{figure}

\paragraph{Outside the divisor-compatibility envelope.} For arbitrary $N$ where one or more stage primes do not divide $N$, the screened framework prescribes three handling options, listed in order of preference for sparse-FFT applications:
\begin{enumerate}
  \item Alternative coprime divisor-compatible moduli: select a different prime tuple $p'_1, \ldots, p'_{L'}$ from a widened interval $[k, c'k]$ with each $p'_\ell \mid N$. Feasibility depends on the prime factorization of $N$; engineered system designs typically choose $N$ so that this option succeeds.
  \item Engineered acquisition or processing length: in pipelined applications where the input length is selected at design time, choose the acquisition or processing length $N' = Q_L = \prod_\ell p_\ell$ from the start. The screened sparse path then runs on the engineered sequence with the formal theorem applying directly. Note: this option is engineered acquisition, not post-hoc zero-padding of a fixed $N_0$-DFT problem; zero-padding a sparse signal from arbitrary $N_0$ to $Q_L$ generally spreads each on-grid tone into a Dirichlet kernel and is incompatible with the on-grid sparsity assumption.
  \item Dense fallback: dispatch directly to the $O(N \log N)$ dense FFT, preserving exact recovery via \Cref{thm:hybrid-correctness}.
\end{enumerate}
The screened sparse-path cost theorem (\Cref{thm:sampling-bluestein}) applies only inside the divisor-compatibility envelope; the handling options above ensure exact recovery on every input.

\paragraph{Implementation note (chirp-$z$ probe, non-canonical).} For applications that attempt a sparse-path probe on arbitrary $N$ without padding, the chirp-$z$ transform~\cite{bluestein1970linear, rabiner1969chirpz} produces samples of the underlying DTFT at $p_\ell$ uniformly-spaced frequency points. These samples do not in general equal residue sums of $X$ modulo $p_\ell$ when $p_\ell \nmid N$, and the screened theorem above does not apply. The chirp-$z$ probe is therefore an implementation extension, not part of the formal sparse-path cost theorem; it can be used as a heuristic precheck before dense fallback in latency-sensitive deployments where the cost of padding (option~2) exceeds the cost of an exploratory probe.

\section{Singleton Screens and Core Lemmas}
\label{sec:screens}

This section states the five operative facts about a single decimated bin on which the multi-stage engine of \S\ref{sec:main-theorem} is built, written directly in the exact ratio/root-membership language of the recovery rule: (i) the true-singleton shifted-observation identity (\Cref{thm:singleton-observation-identity}); (ii) exact phase-ratio recovery of the global label from the three shifted observations (\Cref{lem:global-label-recovery}); (iii) the occupied equal-magnitude pre-screen, the total label-emission rule, and their conjunction, the full exact emission predicate, with its one-sidedness (\Cref{def:exact-screen}, \Cref{thm:phase-consistency-certificate}); (iv) the existence of screen-passing phantoms, so that a screen pass certifies nothing (\Cref{rem:one-sided-normative}); and (v) the verifier backstop for accepted false labels (\Cref{rem:verifier-backstop}).

A bin failing the occupied equal-magnitude pre-screen is never submitted to the total label-emission rule, and a candidate-singleton failing that total rule emits no label; both are normal per-bin events, with the dense-FFT fallback triggered only at stage or chain level per the routing rule of \Cref{thm:hybrid-correctness} (a no-consistent-label stage, a budget overflow, a verifier FAIL, or an envelope violation).

The supporting material behind these statements is collected in \Cref{app:screen-analysis}: the B\'ezout co-primality theorem and its scope, the generic (measure-zero-exceptional) collision-rejection theorem, the shift-triple design notes with the arithmetic-shift Hankel/Prony escalation lemma, the magnitude and phase surrogate lemmas, and the cross-stage collision-multiplicity estimates. The implementation surrogates (normalized MAD, unwrapped affine phase regression) are specified in \Cref{def:cv-singleton-test}, \Cref{def:phase-unwrapping}, and \Cref{rem:impl-surrogates}; they appear in no formal statement of this section.

\subsection{True-Singleton Shifted-Observation Identity}

\noindent\emph{Informally:} a lone tone in a bin does exactly one thing when the acquisition window slides, which is rotate. Its magnitude never moves, and the rotation per unit shift is its frequency; the identity below writes that out, and every condition the screen tests is a test of it.

\begin{theorem}[True-Singleton Shifted-Observation Identity]
\label{thm:singleton-observation-identity}
Let $f$ be an on-grid singleton frequency in decimated bin $r = f \bmod p_\ell$ of stage $\ell$ (that is, $f$ is the only support frequency with $f \equiv r \pmod{p_\ell}$), and write the shifted observations $z_\sigma \coloneqq X_{p_\ell}^{(\sigma)}[r]$ for $\sigma \in \{0, 31, 63\}$. Then, under the sampling convention of \S\ref{subsec:sampling-operator}, exactly
\begin{equation}
z_\sigma = \frac{p_\ell}{N}\, X[f]\, e^{+j 2\pi \sigma f / N},
\label{eq:singleton-observation-identity}
\end{equation}
and consequently:
\begin{enumerate}
  \item (Magnitude invariance and occupancy.) $|z_0| = |z_{31}| = |z_{63}| = \tfrac{p_\ell}{N}\,|X[f]| > 0$: the shifted magnitude is independent of the shift and strictly positive.
  \item (Exact increment ratios.) The increment ratios are exact unimodular rotations, $u_{31} = z_{31}/z_0 = e^{+j 2\pi\,31 f/N}$ and $u_{32} = z_{63}/z_{31} = e^{+j 2\pi\,32 f/N}$, and the B\'ezout phasor is the exact root of unity $w = u_{32}\,u_{31}^{-1} = e^{+j 2\pi f/N} = \omega^{f}$.
\end{enumerate}
\end{theorem}

\begin{proof}
For a singleton bin $r$, the frequency $f$ is the only support frequency with $f \equiv r \pmod{p_\ell}$, so the aliasing identity of \S\ref{subsec:sampling-operator} (with $M = p_\ell$ and $s = \sigma$) reduces to the single term \eqref{eq:singleton-observation-identity}. The shift $\sigma$ enters only through the unimodular phase factor $e^{+j 2\pi \sigma f/N}$, and $X[f]$ is the exact on-grid DFT coefficient with $X[f] \ne 0$ for a true tone; hence $|z_\sigma| = \tfrac{p_\ell}{N}|X[f]| > 0$ is constant across shifts, giving item~1. Forming the ratios cancels the common factor $\tfrac{p_\ell}{N} X[f]$: $u_{31} = e^{+j 2\pi (31-0) f/N}$, $u_{32} = e^{+j 2\pi (63-31) f/N} = e^{+j 2\pi\,32 f/N}$, and $w = u_{32}\,u_{31}^{-1} = e^{+j 2\pi (32-31) f/N} = \omega^{f}$, giving item~2.
\end{proof}

\subsection{Exact Phase-Ratio Recovery of the Global Label}

For a genuine singleton, the identity of \Cref{thm:singleton-observation-identity} makes the global label exactly recoverable from the three shifted observations. Intuitively, the two increments $31$ and $32$ rotate the tone at two rates that come back into step only after a full turn, because they are coprime, so the pair of measured rotations pins the frequency to a single integer rather than to a family of them. The following lemma, on which the global-label engine of \S\ref{sec:main-theorem} rests, records the recovery.

\begin{lemma}[Per-Bin Global-Label Recovery]
\label{lem:global-label-recovery}
Fix stage $\ell$ with prime $p_\ell \mid N$ and the shift triple $\{0,\Delta_1,\Delta_1+\Delta_2\} = \{0,31,63\}$, $\gcd(\Delta_1,\Delta_2) = \gcd(31,32) = 1$. Suppose bin $r$ of the $p_\ell$-decimated view holds a single on-grid tone at global frequency $g$ (so $g \equiv r \pmod{p_\ell}$ and no other support frequency lies in class $r$). Write the three shifted observations $z_\sigma \coloneqq X_{p_\ell}^{(\sigma)}[r]$ for $\sigma \in \{0,31,63\}$. Then:
\begin{enumerate}
  \item (Nonvanishing.) If bin $r$ passes the occupied equal-magnitude pre-screen (condition~1 of \Cref{def:exact-screen}; automatic for a genuine singleton by \Cref{thm:singleton-observation-identity}), then $|z_0| = |z_{31}| = |z_{63}| = \tfrac{p_\ell}{N}|X[g]| > 0$; in particular $z_0 \ne 0$, so the ratios below are well-defined.
  \item (Exact phase ratios.) The per-step ratios $u_{31} \coloneqq z_{31}/z_0$ and $u_{32} \coloneqq z_{63}/z_{31}$ satisfy, exactly,
        \[
          \begin{aligned}
            u_{31} &= e^{+j 2\pi\,\Delta_1 g/N} = e^{+j 2\pi\,31 g/N}, \\
            u_{32} &= e^{+j 2\pi\,\Delta_2 g/N} = e^{+j 2\pi\,32 g/N},
          \end{aligned}
        \]
        the prefactor $p_\ell/N$ and amplitude $X[g]$ canceling.
  \item (Exact global recovery.) With the B\'ezout coefficients $a\Delta_1 + b\Delta_2 = 1$, $(a,b) = (-1,1)$,
        \[
          \begin{aligned}
            w &\coloneqq u_{32}^{\,b}\,u_{31}^{\,a} = u_{32}\,u_{31}^{-1} \\
              &= e^{+j 2\pi\,(\Delta_2 - \Delta_1) g/N} = e^{+j 2\pi\,g/N},
          \end{aligned}
        \]
        and the global label is recovered as the unique index $g \in \{0,1,\dots,N-1\}$ with
        \[
          w = \omega^{g}, \qquad \omega = e^{+j2\pi/N}
        \]
        (equivalently $g = \log_\omega w$, the discrete root-index of \Cref{def:total-recovery}, whose semantics are defined without any transcendental $\arg(\cdot)$ and which is computed by the costed procedure of \Cref{def:label-decoding-model}), and in integer form $g = (a_{63} - a_{31})\,\Delta_2^{-1} \bmod N$, where $a_{\Delta} \coloneqq \log_\omega(z_\Delta/z_0) = \Delta g \bmod N$ is the discrete index of the unit-modulus ratio $z_\Delta/z_0 = \omega^{\Delta g}$ and $\Delta_2^{-1} = 32^{-1} \bmod N$ is well-defined because $N$ is a product of odd primes $p_\ell \ge k \ge 3$, so $N$ is odd and $\gcd(32,N) = \gcd(2^5,N) = 1$. Exactly one value $g \in [0,N)$ is produced.
  \item (Coefficient.) $X[g] = \tfrac{N}{p_\ell}\,z_0$ (no de-rotation, as $\sigma = 0$).
\end{enumerate}
\end{lemma}

\begin{proof}
By the decimation aliasing identity (\Cref{subsec:sampling-operator}) and the singleton hypothesis, the only surviving term in class $r$ is $g$: $z_\sigma = \tfrac{p_\ell}{N} X[g]\,e^{+j 2\pi\,\sigma g/N}$. (1) The magnitude is $\sigma$-independent: $|z_\sigma| = \tfrac{p_\ell}{N}|X[g]|$, equal across shifts, and positive since $X[g] \ne 0$ for a true tone; hence $z_0 \ne 0$. (2) Forming ratios cancels the common $\tfrac{p_\ell}{N} X[g]$: $u_{31} = e^{+j 2\pi\,31 g/N}$, $u_{32} = e^{+j 2\pi\,(63-31) g/N} = e^{+j 2\pi\,32 g/N}$. (3) $w = u_{32} u_{31}^{-1} = e^{+j 2\pi\,(32-31) g/N} = e^{+j 2\pi\,g/N} = \omega^{g}$; since the map $g \mapsto \omega^{g}$ is a bijection from $\mathbb{Z}_N$ onto the $N$th roots of unity, the integer $g \in \{0,\dots,N-1\}$ with $w = \omega^{g}$ (i.e.\ $g = \log_\omega w$) is recovered uniquely. The integer form is the same statement read through $a_{31} = 31 g \bmod N$, $a_{63} = 63 g \bmod N$ and $(a_{63} - a_{31}) = 32 g \bmod N$, inverted by $32^{-1} \bmod N$. Exactly one $g$ results. (4) Set $\sigma = 0$ in $z_\sigma = \tfrac{p_\ell}{N} X[g] e^{+j 2\pi \sigma g/N}$: $z_0 = \tfrac{p_\ell}{N} X[g]$, so $X[g] = \tfrac{N}{p_\ell} z_0$.
\end{proof}

\subsection{The One-Sided Exact Screen}

\begin{definition}[Occupied Equal-Magnitude Pre-Screen, Total Label-Emission Rule, and Full Exact Emission Predicate (Ratio/Root Form)]
\label{def:exact-screen}
Fix stage $\ell$ (prime $p_\ell \mid N$) and the shift triple $\{0, 31, 63\}$ ($\Delta_1 = 31$, $\Delta_2 = 32$, $\gcd(31, 32) = 1$). For bin $r$ with shifted observations $z_\sigma = X_{p_\ell}^{(\sigma)}[r]$, three named predicates are distinguished and never conflated:
\begin{enumerate}
  \item (Occupied equal-magnitude pre-screen.) $|z_0| = |z_{31}| = |z_{63}| > 0$; equivalently, the bin is \emph{occupied}, $\max_\sigma |z_\sigma| > 0$, and attains the vanishing normalized range $D_{\max}(r) = (\max_\sigma |z_\sigma| - \min_\sigma |z_\sigma|)/\max_\sigma |z_\sigma| = 0$, which the occupancy condition makes well-defined. An all-zero bin ($z_0 = z_{31} = z_{63} = 0$, e.g.\ a bin whose class contains no support frequency) fails the occupancy condition, is not a candidate-singleton, and emits no label; consequently every bin passing this pre-screen satisfies $|z_0| = |z_{31}| = |z_{63}| > 0$ without any singleton assumption, so the increment ratios $u_{31} = z_{31}/z_0$ and $u_{32} = z_{63}/z_{31}$ and the B\'ezout phasor $w = u_{32}\,u_{31}^{-1}$ are defined exactly, the well-definedness required by the total rule of \Cref{def:total-recovery}.
  \item (Total label-emission rule: root membership.) $w$ is exactly an $N$th root of unity: $w = \omega^{g}$ ($\omega = e^{+j 2\pi/N}$) for a necessarily unique integer $g \in \{0, \dots, N-1\}$, the bin's \emph{global label}.
  \item (Total label-emission rule: class consistency.) $g \equiv r \pmod{p_\ell}$.
\end{enumerate}
A bin passing the pre-screen of condition~1 is a \emph{candidate-singleton}, the canonical term throughout: the status names the outcome of the pre-screen, not a ground-truth property, and every candidate-singleton is submitted to the \emph{total label-emission rule}, the conjunction of conditions~2 and~3, which are exactly the emit conditions of the total per-bin rule $\textsc{Recover}$ of \Cref{def:total-recovery}, executed and costed by the decoding procedure of \Cref{def:label-decoding-model}. The conjunction of all three conditions is the \emph{full exact emission predicate}: a bin satisfying it emits the label $g$; a candidate-singleton failing the emission rule emits no label; a bin failing the pre-screen is never submitted to $\textsc{Recover}$. No $\arg(\cdot)$-based phase-linearity predicate appears in any of the three: the unwrapped affine phase regression of \Cref{def:phase-unwrapping} is an implementation surrogate only (\Cref{rem:impl-surrogates}), as is the normalized-MAD magnitude test of \Cref{def:cv-singleton-test}.
\end{definition}

\noindent\emph{Informally:} the pre-screen and the emission rule can only ever say ``not a singleton''. A genuine singleton passes every condition of the full exact emission predicate, so a failure is decisive, while a pass is not, because a collision can imitate the signature on a thin set of coefficients. The theorem states those two halves and the gap between them.

\begin{theorem}[One-Sided Exact Screen (Pre-Screen and Emission Behavior on the Co-Prime Shift Triple)]
\label{thm:phase-consistency-certificate}
In the noiseless on-grid model, on the shift triple $\{0, 31, 63\}$:
\begin{enumerate}
  \item (Genuine singletons pass the pre-screen and emit their true label.) If bin $r$ of stage $\ell$ holds a genuine singleton at frequency $f$, the bin passes the occupied equal-magnitude pre-screen, satisfies the total label-emission rule, and therefore satisfies the full exact emission predicate of \Cref{def:exact-screen}, emitting exactly the label $g = f$; the predicate has zero false negatives on singletons.
  \item (Emission-rule failures reject singleton status.) Contrapositively, a bin that fails the pre-screen or the total label-emission rule, i.e.\ any conjunct of the full exact emission predicate, does not hold a genuine singleton; it emits no label, a normal per-bin event, and the input reaches the dense-FFT fallback only through the stage- and chain-level routing rule of \Cref{thm:hybrid-correctness} (the routing and the optional higher-order Prony escalation are detailed in \Cref{rem:collision-escalation}).
  \item (Passes are provisional.) The predicate is one-sided: at every collision multiplicity $m \ge 2$ there are admissible instances and bins on which a measure-zero coefficient slice of $m$-tone collisions satisfies the full exact emission predicate and emits a phantom label (\Cref{prop:two-tone-phantom,prop:higher-order-phantom}), while for generic (Lebesgue-almost-every) coefficients on each fixed support an $m \ge 2$ collision fails it (\Cref{thm:generic-collision-rejection}). An emitted label is therefore provisional, and exactness of an accepted output is secured by the sound final verifier of \Cref{thm:verifier-soundness}, not by the pre-screen or the emission rule.
\end{enumerate}
\end{theorem}

\begin{proof}
Item 1. By \Cref{thm:singleton-observation-identity}, a genuine singleton at $f$ has $|z_0| = |z_{31}| = |z_{63}| = \tfrac{p_\ell}{N}|X[f]| > 0$ (the pre-screen, condition~1) and $w = \omega^{f}$ exactly, so $w$ is an $N$th root of unity whose unique index is $g = f$ (root membership, condition~2; the map $g \mapsto \omega^{g}$ is a bijection from $\mathbb{Z}_N$ onto the $N$th roots of unity), and $g = f \equiv r \pmod{p_\ell}$ because $f$ lies in class $r$ (class consistency, condition~3). The emitted label is $g = f$, as recorded in \Cref{lem:global-label-recovery}(3). Item 2 is the contrapositive of item~1: a genuine singleton satisfies every conjunct of the full exact emission predicate, so a bin failing one holds no genuine singleton. The proof of item~2 uses nothing beyond item~1; the routing consequence it mentions is stated, not used, and is named in the theorem statement. Item 3. The explicit two-tone witness of \Cref{prop:two-tone-phantom} satisfies the full exact emission predicate and emits a class-consistent phantom label, and the construction extends to every $m \ge 3$ (\Cref{prop:higher-order-phantom}); genericity of rejection is \Cref{thm:generic-collision-rejection}, proved in \Cref{app:screen-analysis}. One-sidedness and verifier-gated exactness are \Cref{rem:one-sided-normative} and \Cref{thm:verifier-soundness}.
\end{proof}

\subsection{Phantoms and the Verifier Backstop}

\begin{remark}[Normative One-Sidedness of the Screen]
\label{rem:one-sided-normative}
The three-shift screen is \emph{one-sided}: every true on-grid singleton passes (zero false negatives), and generic $m$-tone collisions fail (\Cref{thm:generic-collision-rejection}), but at every collision multiplicity $m \ge 2$ there are admissible instances on which a measure-zero coefficient slice satisfies the full exact emission predicate of \Cref{def:exact-screen} while holding a non-singleton (\Cref{prop:two-tone-phantom,prop:higher-order-phantom}). A bin passing the pre-screen is only a candidate-singleton, an emitted label is only provisional, and a full-emission-predicate-passing phantom emits a provisional label and can be admitted to the candidate set, never caught by the screen. A stage-local screen pass does not by itself invoke the verifier: if a phantom label survives the global-label intersection and is used in the reconstructed candidate, the residual verifier rejects the candidate (\Cref{thm:verifier-soundness}) and the input is routed to the dense fallback (\Cref{thm:hybrid-correctness}). Exactness of an accepted output rests on that verifier and never on the screen. All later one-sidedness statements defer to this remark.
\end{remark}

\begin{remark}[Verifier Backstop for Accepted False Labels]
\label{rem:verifier-backstop}
If a phantom label survives the global-label intersection and is used in the reconstructed candidate, the residual verifier rejects the candidate: the reconstruction then differs from $X$, the residual $x - \hat{x}$ is nonzero, and the sound final verifier of \Cref{thm:verifier-soundness} returns FAIL, routing the input to the exact dense fallback (\Cref{rem:phantom-labels}, \Cref{thm:hybrid-correctness}). Worked instances are given in \S\ref{sec:adversarial}: a surrogate-admitted collision rejected at root membership, and an exact full-emission-predicate-passing two-tone phantom that no screen can exclude, rejected only by final verification if it survives the intersection.
\end{remark}

\section{Conditional Fast-Path Complexity and Unconditional Correctness}
\label{sec:main-theorem}

The screen results of the previous section deliver per-bin candidate-singleton screening and exact global-label emission within a single CRT stage (one-sided screen, \Cref{rem:one-sided-normative}; generic collisions are rejected outside a measure-zero coefficient slice, \Cref{thm:generic-collision-rejection}). This section assembles them into a multi-stage algorithm, in dependency order: the two cost models are fixed first (\Cref{def:label-decoding-model}), then the sound final verifier that carries correctness (\Cref{def:verifier}, \Cref{thm:verifier-soundness}), then the global-label engine and its deterministic candidate count and containment property (\Cref{thm:global-label-count}, \Cref{thm:global-label-completeness}), then the cost of candidate construction in each model (\Cref{thm:sparse-path-complexity}, \Cref{thm:cost-comparison-model}), and finally the unconditional hybrid correctness statement (\Cref{thm:hybrid-correctness}). Each of these objects is proved once and cited thereafter; no proof in this section re-derives another Section~VI object's content. The accounting is that each of the $L = O(\log_k N)$ stages does only $O(k \log k)$ FFT work while the global-label engine produces a candidate set of proved size $O(k)$, so the total work along the screened sparse path is $L \cdot O(k \log k) = O(k \log N)$. This bounds the cost of candidate construction; exactness of the returned spectrum is verifier-gated, not established by the stage chain alone. The attempted sparse path is allowed to abort: if a stage emits no class-consistent global label, or a budget guard fires, the input reaches the dense-FFT fallback through the stage- and chain-level routing rule of \Cref{thm:hybrid-correctness}, after a sparse-attempt spend bounded by the cost theorems (\Cref{thm:sparse-path-complexity,thm:cost-comparison-model}).

\subsection{Standing Envelope and the Two Cost Models}

\begin{table*}[!t]
\centering
\caption{Operating Model, Fast-Path Condition, and Proved Safeguards}
\label{tab:oklogn-assumptions}
\small
\begin{tabular}{p{2.8cm}p{6.5cm}p{6cm}}
\toprule
\emph{Item} & \emph{Description} & \emph{Role / status} \\
\midrule
\multicolumn{3}{l}{\emph{(1) Operating model (exact, design-selected).}} \\
\addlinespace
Signal model & Exact, at-most-$k$-sparse on-grid spectrum $x \in \mathbb{C}^N$; noiseless; exact arithmetic; frequencies $f \in \{0, 1, \ldots, N-1\}$ & Standing model of \Cref{def:signal-model} for \Cref{thm:verifier-soundness,thm:hybrid-correctness} \\
\addlinespace
Transform length & Exact product $N = \prod_{\ell=1}^{L} p_\ell$ of the $L$ stage primes ($L = O(\log_k N)$); divisor-compatible decimation (\Cref{thm:sampling-bluestein}) & Engineered divisor compatibility (\Cref{def:transform-family}); $k^L \leq N \leq (ck)^L$ when $p_\ell \in [k,ck]$ \\
\addlinespace
Prime selection & Pairwise distinct $p_\ell \in [k, ck]$, $c \ge 2$ (special case $c=2$); $k \le \min_\ell p_\ell$ & Prime availability with CRT coverage; at least $k$ bins per stage \\
\addlinespace
Co-prime shifts & $S = 3$ shifts $\{0, 31, 63\}$, increments $\Delta_1 = 31$, $\Delta_2 = 32$ with $\gcd(31,32) = 1$ (B\'ezout phase diversity; not Hankel/Prony rank-one) & Exact global-label phase read for genuine singletons (\Cref{lem:global-label-recovery}) \\
\midrule
\multicolumn{3}{l}{\emph{(2) Fast-path condition (premise, not assumed for correctness).}} \\
\addlinespace
All-stage genuine-singleton survival & $\mathcal{S} \ne \emptyset$ and every $f \in \mathcal{S}$ is the unique support frequency in its class $f \bmod p_\ell$ at every stage $\ell$ (\Cref{cond:all-stage-survival}) & Premise of containment $\mathcal{S} \subseteq \mathcal{G}$ (\Cref{thm:global-label-completeness}), hence of the sparse path completing at $O(k \log N)$ rather than falling back; violators routed to dense fallback \\
\midrule
\multicolumn{3}{l}{\emph{(3) Proved safeguards (theorems, not assumptions).}} \\
\addlinespace
$O(k)$ candidate count & Each candidate-singleton emits at most one global label; per-stage label sets are intersected, never multiplied & Proved by \Cref{thm:global-label-count} (\Cref{rem:no-product-blowup}); holds for every input \\
\addlinespace
Sound final verifier & PASS on $k + |\mathcal{G}|$ consecutive residual samples $\Rightarrow \hat{X} = X$ & Proved by \Cref{thm:verifier-soundness} (Vandermonde); secures unconditional correctness (\Cref{thm:hybrid-correctness}) \\
\bottomrule
\end{tabular}
\end{table*}

\noindent The formal model is given in \S\ref{sec:problem} and is not restated here: the signal model is \Cref{def:signal-model}, the engineered exact-product transform family is \Cref{def:transform-family}, and the fast-path premise is \Cref{cond:all-stage-survival}. \Cref{tab:oklogn-assumptions} summarizes them alongside the two proved safeguards, so that what is assumed and what is proved can be read off in one place. Implementation surrogates (normalized MAD, $R^2$, the $64 p_\ell$ batching factor) are not part of the formal model and are stated separately in the implementation note \Cref{rem:impl-surrogates}. Only the per-bin label-read cost model remains to be fixed, and it is fixed once, here.

\begin{definition}[Exact Label-Decoding (Root-Index Oracle) Model, and the Comparison Real-RAM Model]
\label{def:label-decoding-model}
All label-read cost claims of this paper are stated in one of the two computational models fixed once here and referenced by \Cref{lem:global-label-recovery}, \Cref{def:total-recovery}, and the cost accounting of \Cref{thm:sampling-bluestein,thm:sparse-path-complexity,thm:cost-comparison-model}.

(a) Exact label-decoding (root-index oracle) model. A real-RAM with exact complex arithmetic and exact equality comparison, extended with three unit-cost oracles: exact evaluation of $\arg(\cdot)$, exact evaluation of the root power $\omega^{m} = e^{+j 2\pi m/N}$ at any integer $m$, and exact nearest-integer rounding $\operatorname{round}(\cdot)$ on exact real quantities. The per-bin label read $g = \log_\omega w$ of \Cref{def:total-recovery} (bin of stage $\ell$, class $r$) then costs $O(1)$: a rounded-$\arg$ guess $g_0 = \operatorname{round}\!\bigl(N \arg(w)/2\pi\bigr) \bmod N$, decided by the exact root-membership comparison $\omega^{g_0} = w$ and the exact class check $g_0 \equiv r \pmod{p_\ell}$, so a label is emitted only on exact passes and exactness never rests on the $\arg$-quantization. Preprocessing: none; storage: $O(1)$ beyond the samples; no precomputed $\Theta(N)$ root table is used or assumed.

(b) Comparison real-RAM model (no root-index oracle). The same real-RAM without the unit-cost rounding and root-evaluation oracles, retaining exact complex arithmetic, exact equality and order comparison, and exact $\arg(\cdot)$. The per-bin label read then costs $O(\log N)$: the nearest-integer conversion is located among the $N$ admissible indices by $O(\log N)$ exact comparisons (binary search), $\omega^{g_0}$ is computed from $\omega$ by repeated squaring in $O(\log N)$ exact multiplications, and the class check remains $O(1)$.

The unit-cost primitives of model (a) are a declared oracle assumption about root-index decoding, and neither model is a bit-complexity model: no bit-level implementability claim is made for either. The costed decoding procedure in step-by-step form, and the full disclosure of what model (b) is and is not, are collected in \Cref{app:cost-models}. Indyk, Kapralov \& Price~\cite{ikp2014soda} make a comparable unit-cost assumption and discharge it; no such discharge is offered here for the root-index read, which is why it is declared as a model rather than removed (the anatomy of that comparison, with its source provenance, is also in \Cref{app:cost-models}).

\textbf{Global cost-model convention:} unless stated otherwise, every unqualified $O(k \log N)$ sparse-path arithmetic figure in this paper is stated in the root-index-oracle model of part~(a) of this definition; comparison-model figures are explicitly marked. Costs quoted as $O(k^2)$ for the final verifier are likewise model (a) figures; in model (b) the verifier's root powers are formed by repeated squaring, adding $O(k \log N)$ (\Cref{prop:verifier-cost}).
\end{definition}

\subsection{Sound Final Verifier}

\noindent The residual $r = x - \hat{x}$ between the signal and the inverse transform of a candidate spectrum $\hat{X}$ supported on a known candidate set $\mathcal{G}$ is a sum of at most $k + |\mathcal{G}|$ on-grid complex exponentials. Such a signal is identically zero if and only if it vanishes on any $k + |\mathcal{G}|$ consecutive time samples. This is a deterministic Vandermonde fact requiring no modulus family, no divisor compatibility, and no separation argument, and it establishes final-verifier soundness as a theorem (\Cref{thm:verifier-soundness}).

\begin{definition}[Consecutive-Sample Final Verifier]
\label{def:verifier}
Let $\mathcal{G} \subseteq \{0,\ldots,N-1\}$ be a known finite candidate set and let $\hat{X}$ be a candidate spectrum supported on $\mathcal{G}$, i.e.\ $\hat{S} \coloneqq \{f : \hat{X}[f] \ne 0\} \subseteq \mathcal{G}$, with duplicate frequencies consolidated and zero coefficients dropped. Both the sparsity bound $k$ of \Cref{def:signal-model} and the candidate-set size $|\mathcal{G}|$ are known to the algorithm, so the \emph{window size}
\[
  W \;\coloneqq\; k + |\mathcal{G}|
\]
is computable before the verifier runs. Fix any base index $n_0 \in \{0,\ldots,N-1\}$ (for example $n_0 = 0$). Given $\hat{X}$ and time-domain access to $x$, compute the residual at $W$ consecutive indices
\[
  \begin{aligned}
    r[n] &= x[n] - \hat{x}[n], \qquad
    \hat{x}[n] = \tfrac{1}{N}\!\sum_{f \in \hat{S}} \hat{X}[f]\, e^{+j 2\pi f n/N}, \\
    &\quad\text{for } n = n_0, n_0+1, \ldots, n_0 + W - 1,
  \end{aligned}
\]
with all indices taken modulo $N$ (the window may wrap cyclically past $N-1$; since $\omega_\ell^{\,n}$ is $N$-periodic in $n$, cyclic wrapping leaves the Vandermonde structure of \Cref{thm:verifier-soundness} unchanged). The verifier returns $V(x,\hat{X}) = \mathrm{PASS}$ if and only if $r[n] = 0$ for all these $W$ indices (exactly, in the noiseless on-grid model). The $W$ values $x[n_0],\ldots,x[n_0+W-1]$ are read directly from $x$, which is available in the time domain by the recovery model; for $|\mathcal{G}| = O(k)$ this is an additional $O(k)$ reads, so the total sample budget of the verified sparse path remains $O(k \log N)$.

\emph{Uniform special case $|\mathcal{G}| \le k$.} Whenever the candidate set satisfies $|\mathcal{G}| \le k$, which the global-label budget of \Cref{alg:multistage_oklogn} enforces on the accumulated set before $V$ is invoked, the window obeys $W = k + |\mathcal{G}| \le 2k$, so the instance-adapted window is contained in a uniform $2k$-sample window and a $2k$-sample verifier is also sound. The uniform $2k$ form is the worst-case instantiation of this definition and is the sizing quoted when no candidate-set size is at hand; the general $k + |\mathcal{G}|$ form is the operative one and is what \Cref{alg:multistage_oklogn} reads. Sizing the window to $2|\mathcal{G}|$ instead would be unsound when $|\mathcal{G}| < k$ (\Cref{thm:verifier-soundness}).

\emph{Precondition (the window is budgeted).} If $|\mathcal{G}_\ell| > ck$ at any stage, or if the final accumulated set satisfies $|\mathcal{G}| > k$ once the stage chain completes, fallback occurs before verification: the input is routed to the dense $O(N \log N)$ path and $V$ is never invoked. This precondition is total in the algorithm, so whenever $V$ runs the window fits and the hypotheses of \Cref{thm:verifier-soundness} hold.
\end{definition}

\noindent\emph{Informally:} a nonzero signal built from at most $k + |\mathcal{G}|$ tones cannot hide on $k + |\mathcal{G}|$ consecutive samples, so a residual that vanishes there is identically zero.

\begin{theorem}[Final-Verifier Soundness]
\label{thm:verifier-soundness}
Let $x$ satisfy \Cref{def:signal-model} ($|\mathrm{supp}(X)| \le k$, noiseless, on-grid, exact arithmetic) and let $\hat{X}$ be a candidate spectrum supported on a known candidate set $\mathcal{G}$. Then the verifier of \Cref{def:verifier}, reading $W = k + |\mathcal{G}|$ consecutive residual samples, satisfies
\[
  V(x,\hat{X}) = \mathrm{PASS} \ \Longrightarrow\ \hat{X} = X .
\]
The bound is unconditional in $|\mathcal{G}|$ versus $k$: it does not assume $|\mathrm{supp}(X)| \le |\mathcal{G}|$, which can fail when $|\mathcal{G}| < k$ and a true tone is omitted from $\mathcal{G}$. In the uniform special case $|\mathcal{G}| \le k$ the window is $W \le 2k$, and a $2k$-sample uniform window is likewise sound; the rank argument behind it does not extend below a uniform window of $2k$ consecutive samples.
\end{theorem}

\begin{proof}
The candidate $\hat{X}$ is supported on $\mathcal{G}$ and the true spectrum $X$ on $\mathrm{supp}(X)$ with $|\mathrm{supp}(X)| \le k$, so the residual spectrum $\hat{r} = X - \hat{X}$ is supported on $D = \mathrm{supp}(X) \cup \mathcal{G}$, of size
\[
  s \ \coloneqq\ |D| \ \le\ |\mathrm{supp}(X)| + |\mathcal{G}| \ \le\ k + |\mathcal{G}| \ =\ W .
\]
This residual-support bound holds even when a true tone is absent from $\mathcal{G}$ (in which case $\mathrm{supp}(X)$ and $\mathcal{G}$ are not nested and $s$ can reach $k + |\mathcal{G}| > 2|\mathcal{G}|$). Writing $D = \{f_1,\ldots,f_s\}$, $a_\ell = \hat{r}[f_\ell]$, and $\omega_\ell = e^{+j 2\pi f_\ell/N}$, the inverse DFT gives $r[n] = \tfrac{1}{N}\sum_{\ell=1}^{s} a_\ell\,\omega_\ell^{\,n}$. The nodes $\omega_1,\ldots,\omega_s$ are distinct (distinct on-grid frequencies in $[0,N)$). Write $\Psi$ for the $W \times s$ node matrix $\Psi_{i\ell} = \omega_\ell^{\,n_0+i}$, $i = 0,\ldots,W-1$ (indices modulo $N$; $\omega_\ell^{\,n}$ is $N$-periodic in $n$, so a cyclically wrapped window gives the same matrix); the window size $W$ is a scalar throughout and $\Psi$ is the only matrix in this proof. Its first $s$ rows equal a Vandermonde matrix on the distinct nodes times $\mathrm{diag}(\omega_\ell^{n_0})$ (the invertible diagonal factor scales the columns, so it multiplies on the right), hence $\Psi$ has full column rank $s$ (the determinant $\prod_{a<b}(\omega_b - \omega_a) \neq 0$ for distinct nodes). A PASS verdict is exactly $\Psi a = 0$ on those $W \ge s$ consecutive rows, which forces $a = 0$; therefore $\hat{r} = 0$ and $\hat{X} = X$.

Consecutiveness is essential: it is what makes the rows a Vandermonde system. Two sizing remarks complete the statement. First, sizing the window to $2|\mathcal{G}|$ would be unsound when $|\mathcal{G}| < k$: a nonzero residual with $k + |\mathcal{G}| > 2|\mathcal{G}|$ distinct frequencies can vanish on $2|\mathcal{G}|$ consecutive samples (an underdetermined system), producing a false PASS, whereas the $W = k + |\mathcal{G}|$ window makes the system full rank and forecloses this. Second, in the uniform special case $|\mathcal{G}| \le k$ one has $W \le 2k$, so a $2k$-sample uniform window contains the required one and is sound; and the rank argument does not extend below $2k$ for a uniform window: whenever $N \ge 2k$ (so that $2k$ distinct nodes exist) a window of only $2k-1$ consecutive samples yields a $(2k-1)\times 2k$ system with a nontrivial null space, so vanishing on such a window no longer forces $a = 0$.
\end{proof}

\noindent\emph{Why checking only claimed frequencies is insufficient.} A Goertzel evaluation restricted to $\hat{S}$ cannot certify that omitted support is absent (a claimed-frequency-only check leaves this gap). \Cref{def:verifier} avoids this gap: the $k + |\mathcal{G}|$ consecutive samples constrain the entire residual, including any unclaimed support frequency, through the full-rank Vandermonde system of \Cref{thm:verifier-soundness}.

\begin{proposition}[Verifier Complexity]
\label{prop:verifier-cost}
On a candidate set of size $|\mathcal{G}| = O(k)$, $V$ reads $k + |\mathcal{G}| = O(k)$ time-domain samples of $x$ and performs $O\!\bigl((k+|\mathcal{G}|)^2\bigr) = O(k^2)$ arithmetic operations in the root-index oracle model of \Cref{def:label-decoding-model}(a); the standalone $O(k^2)$ figure is a model-(a) figure and is not stated model-free. In the comparison real-RAM model of \Cref{def:label-decoding-model}(b) the verifier's root powers must be built by repeated squaring, which adds $O(k \log N)$, so $V$ costs $O(k^2 + k \log N)$ there. Consequently the sparse path, when taken and verified, uses $O(k \log N)$ samples and, in model (a), $O(k \log N) + O(k^2)$ arithmetic operations; in model (b) the verified total is $O(k \log^2 N / \log k) + O(k^2)$ (\Cref{thm:cost-comparison-model}), the verifier's extra $O(k \log N)$ being absorbed by the candidate-construction term.
\end{proposition}

\begin{proof}
Reading $x[n_0],\ldots,x[n_0+W-1]$ with $W = k + |\mathcal{G}| = O(k)$ is $O(k)$ sample reads, an additive term over the $O(k\log N)$ samples of the recovery path. Each $\hat{x}[n]$ is a sum of $|\hat{S}| \le |\mathcal{G}| = O(k)$ exponential terms, costing $O(k)$ per index and $O(k^2)$ over the $W = O(k)$ indices; the $W$ zero comparisons cost $O(k)$. The root powers $e^{+j2\pi f n/N} = \omega^{fn \bmod N}$ in these sums are unit-cost under the root-evaluation primitive of \Cref{def:label-decoding-model}(a); in the comparison real-RAM model, $\omega^{f n_0}$ and the per-index update factor $\omega^{f}$ are formed once per frequency by repeated squaring ($O(k \log N)$ total) and each subsequent index costs one multiplication per term, leaving the verifier at $O(k^2 + k \log N)$, absorbed by the stated comparison-model total. Hence $O(k)$ samples and $O(k^2)$ arithmetic for $V$ in model (a), additive to the $O(k\log N)$ arithmetic of \Cref{thm:sparse-path-complexity}.
\end{proof}

\noindent\emph{The verified sparse-path cost: state the $O(k^2)$ term explicitly.} The verified sparse-path arithmetic is $O(k \log N) + O(k^2)$, and the $O(k^2)$ verifier term must not be folded into $O(k \log N)$. The two coincide as $O(k \log N)$ only when $k = O(\log N)$; this is already violated at the boundary illustration $N = 765{,}049$, $k = 20$, $L = 4$ (an admissible regime-boundary instance, \Cref{rem:engineering-regime}), where $k^2 = 400$ exceeds $k \log_2 N \approx 391$. In every case $O(k^2)$ is far below the $O(N \log N)$ dense fallback (the model's standing $L \ge 2$ gives $N > k^2$ at every admissible instance, \Cref{def:transform-family}; for that corner, $400 \ll 765{,}049 \cdot 19.5 \approx 1.5 \times 10^7$), so the verifier never changes the worst-case bound; but the headline sparse-path cost is the explicit sum $O(k \log N) + O(k^2)$, not $O(k \log N)$ alone.

\subsection{Global-Label Engine: Deterministic \texorpdfstring{$O(k)$}{O(k)} Candidate Count}
\label{subsec:global-label-engine}

The per-bin recovery of \Cref{lem:global-label-recovery} is the operative candidate path. Instead of stitching screened residues into cross-stage tuples (a rule retained only as the negative result of \Cref{app:tuple-negative}), each candidate-singleton emits at most one global frequency label (possibly none) by the total rule of \Cref{def:total-recovery}, without assuming the bin is a true singleton; the per-stage label sets are then intersected. This replaces the multiplicative tuple count by a per-stage linear count, establishing the proved $O(k)$ bound (\Cref{thm:global-label-count}). Correctness remains with the verifier of \Cref{thm:verifier-soundness}: the global-label engine constructs candidates, and the engine does not certify which of them are true.

\begin{definition}[Total Label-Emission Rule ($\textsc{Recover}$, Total Per-Bin Form)]
\label{def:total-recovery}
\Cref{lem:global-label-recovery} assumes the bin holds an actual singleton; a screen-passing multi-tone (phantom) bin is not a singleton, so that lemma does not by itself bound a phantom bin's output. We therefore fix a \emph{total} rule $\textsc{Recover}(z_0,z_{31},z_{63})$, the operational form of the total label-emission rule of \Cref{def:exact-screen}, defined on every candidate-singleton, true singleton or not, and emitting at most one label. Let bin $r$ of stage $\ell$ (prime $p_\ell \mid N$, class $r$) pass the occupied equal-magnitude pre-screen (condition~1 of \Cref{def:exact-screen}) with $z_0 \ne 0$ (the pre-screen forces $|z_0| = |z_{31}| = |z_{63}| > 0$ on every passing bin, singleton or not, \Cref{def:exact-screen}; \Cref{lem:global-label-recovery}(1) exhibits the genuine-singleton instance). Form $u_{31} = z_{31}/z_0$, $u_{32} = z_{63}/z_{31}$, and $w = u_{32}\,u_{31}^{-1}$. The bin emits the label $g \in [0,N)$ if and only if
\begin{enumerate}
  \item $w$ is exactly an $N$th root of unity, i.e.\ $w = e^{+j 2\pi g/N}$ for an integer $g \in \{0,\dots,N-1\}$; the label is then that unique integer $g$, equivalently the discrete index $g = \log_\omega w$ to base $\omega = e^{+j 2\pi/N}$ (the unique $g \in \{0,\dots,N-1\}$ with $w = \omega^{g}$), whose existence and uniqueness are guaranteed precisely by the $N$th-root-of-unity condition just stated; this is an exact integer index, defined without any transcendental $\arg(\cdot)$, and
  \item that $g$ is consistent with the bin's class, $g \equiv r \pmod{p_\ell}$;
\end{enumerate}
otherwise the bin is declared \emph{screen-inconsistent} (it passes the pre-screen but fails the total label-emission rule) and emits no label (its frequency mass is not screened at stage $\ell$; any support tone thereby dropped is recovered by the dense fallback, triggered per the routing rule of \Cref{thm:hybrid-correctness}, namely a no-consistent-label stage, a budget overflow, or a final-verifier FAIL; cf.\ \Cref{rem:alg1-semantics}). Thus each bin emits at most one in-range, on-grid label, totally and deterministically, without assuming the bin is a true singleton. The label read is executed, and costed, by the procedure of \Cref{def:label-decoding-model}: a rounded-$\arg$ guess followed by exact root-membership and exact class verification, with no precomputed root table; both emit conditions above are decided by exact comparisons, so the rule's semantics are unchanged by the costed procedure. For a genuine singleton both conditions hold automatically and $\textsc{Recover}$ returns the lemma's $g$ (\Cref{lem:global-label-recovery}(3), completeness direction unchanged); for a screen-passing phantom whose triple matches a single exponential at some $h$ (\Cref{prop:two-tone-phantom,prop:higher-order-phantom}), the rule emits that single $h$ if it is on-grid and class-consistent, and otherwise emits nothing, so a phantom contributes at most one label and never more.
\end{definition}

\noindent\emph{Informally:} every screened bin casts at most one vote for a full frequency, and stages can only veto votes, never multiply them, so the candidate list cannot outgrow one stage's bin count.

\begin{theorem}[Deterministic $O(k)$ Global-Label Count]
\label{thm:global-label-count}
Let $\textsc{Recover}$ be the total label-emission map of \Cref{def:total-recovery} (defined on every candidate-singleton, true singleton or not) and set
\[
  \begin{aligned}
    \mathcal{G}_\ell \coloneqq \{\,g ={} &\textsc{Recover}(z_0,z_{31},z_{63}) : \\[-2pt]
      &\text{bin }r\text{ of stage }\ell\text{ passes the}\\[-2pt]
      &\text{pre-screen \Cref{def:exact-screen}(1)}\,\},
  \end{aligned}
\]
and let $b_\ell$ denote the number of bins of stage $\ell$ passing that pre-screen. Then, in the noiseless on-grid model of \Cref{def:signal-model} and with no association assumption:
\[
  |\mathcal{G}_\ell| \,\le\, b_\ell \,\le\, p_\ell \,\le\, ck = O(k),
\]
deterministically for every input spectrum. Consequently the accumulated candidate set $\mathcal{G} \coloneqq \bigcap_{\ell=1}^{L} \mathcal{G}_\ell$ satisfies
\[
  |\mathcal{G}| \,\le\, \min_{\ell}|\mathcal{G}_\ell| \,\le\, |\mathcal{G}_1| \,\le\, p_1 \,\le\, ck = O(k),
\]
deterministically. Here $c$ is the constant of \Cref{def:transform-family}, fixed independently of $k$ and $N$ (an absolute constant, $c \le 4$ sufficing across \Cref{tab:prime-availability}), so $ck = O(k)$ with an absolute implied constant; the same convention applies to every $O(k)$ claim in this paper that rests on $p_\ell \le ck$. In particular the per-stage and accumulated candidate counts are bounded by a constant multiple of $k$ without any association assumption: the count is proved, not assumed.
\end{theorem}

\begin{proof}
Stage $\ell$ has exactly $p_\ell$ decimated bins $r \in \{0,\dots,p_\ell-1\}$. The count uses the total label-emission rule of \Cref{def:total-recovery}, which is defined on every candidate-singleton (every bin passing the pre-screen) and emits at most one in-range on-grid label, without assuming any bin is a true singleton: a candidate-singleton either emits one label (when its $w$ is exactly an $N$th root of unity yielding an integer $g \equiv r \pmod{p_\ell}$) or is screen-inconsistent and emits none. Hence the map $\{\text{passing bins}\} \to [0,N) \cup \{\varnothing\}$ is well-defined and total, and its image $\mathcal{G}_\ell$ (dropping the no-label outcomes) has cardinality at most the number of passing bins, which is at most the total number of bins $p_\ell$. Crucially this holds whether the passing bins are true singletons or screen-passing phantoms (\Cref{prop:two-tone-phantom,prop:higher-order-phantom}), since a phantom bin also emits at most one label (\Cref{def:total-recovery}); the count therefore does not rely on singleton truth. The transform family (\Cref{def:transform-family}) gives $p_\ell \le ck$, so $|\mathcal{G}_\ell| \le ck = O(k)$. The intersection of sets is contained in each, so $|\mathcal{G}| = |\bigcap_\ell \mathcal{G}_\ell| \le |\mathcal{G}_1| \le p_1 \le ck$. No step uses randomness, amplitude separation, a cross-stage association rule, or an assumption that any bin is a genuine singleton; the bound holds for every input. The count is linear per stage because the stages are combined by intersection and never by a Cartesian product, the architectural point recorded in \Cref{rem:no-product-blowup}.
\end{proof}

\noindent\emph{Informally:} a true tone that never shares a bin with another true tone is read correctly at every stage and survives every intersection.

\begin{theorem}[Support Containment under All-Stage Survival]
\label{thm:global-label-completeness}
Let
\[
  \begin{aligned}
    \mathcal{S}^{\mathrm{surv}} \coloneqq \{f \in \mathcal{S} : \ {}&f\text{ is the unique support frequency} \\[-2pt]
      &\text{in class }f \bmod p_\ell\text{ for every stage }\ell\}
  \end{aligned}
\]
be the all-stage-surviving support (true tones that are genuine singletons at every stage). Then every $f \in \mathcal{S}^{\mathrm{surv}}$ satisfies $f \in \mathcal{G}_\ell$ for all $\ell$, hence $f \in \mathcal{G} = \bigcap_\ell \mathcal{G}_\ell$. In particular, under \Cref{cond:all-stage-survival} the whole support is contained in the candidate set, $\mathcal{S} \subseteq \mathcal{G}$. A tone that collides at some stage $\ell_0$ (sharing class $f \bmod p_{\ell_0}$ with another support frequency) is not guaranteed in $\mathcal{G}$ and is recovered by the dense-FFT fallback (\Cref{thm:hybrid-correctness}), as intended by \Cref{rem:alg1-semantics}. The theorem asserts containment only; the reverse inclusion is never invoked, and nothing downstream uses equality.
\end{theorem}

\begin{proof}
Fix $f \in \mathcal{S}^{\mathrm{surv}}$ and a stage $\ell$. By definition $f$ is the unique support frequency in class $f \bmod p_\ell$, so bin $r = f \bmod p_\ell$ holds the single tone $f$ and passes the occupied equal-magnitude pre-screen (equal magnitudes at a strictly positive common value; its increment ratios are exact singleton rotations, \Cref{thm:singleton-observation-identity}, as it is a single exponential). By \Cref{thm:phase-consistency-certificate}(1) the bin then passes every condition of the exact screen and the total rule of \Cref{def:total-recovery} emits exactly the label $g = f$ (the phase read itself is \Cref{lem:global-label-recovery}), which is what membership in $\mathcal{G}_\ell$ requires. Hence $f \in \mathcal{G}_\ell$. As $\ell$ was arbitrary, $f \in \mathcal{G}_\ell$ for all $\ell$, so $f \in \bigcap_\ell \mathcal{G}_\ell = \mathcal{G}$; under \Cref{cond:all-stage-survival} one has $\mathcal{S}^{\mathrm{surv}} = \mathcal{S}$, giving $\mathcal{S} \subseteq \mathcal{G}$. The collision routing is the content of \Cref{rem:alg1-semantics}: a tone that is not a singleton at some stage is, by design, handled by the dense path, not the sparse global-label path.

\emph{Two scope notes, recorded here so that no other proof re-derives them.} First, containment rests only on the survival premise, never on a counting estimate: the CRT counting bound (\Cref{rem:crt-occupancy-estimate}, \eqref{eq:crt-counting-bound}) controls colliders modulo the accumulated product $\prod_{i \le t} p_i$, not modulo an individual stage prime, so two frequencies can still collide modulo $p_L$ at the final stage even though $\prod_{i \le L} p_i = N$ forces the accumulated count to zero. No per-stage collision-freeness follows from that estimate, and none is used here. Second, under \Cref{cond:all-stage-survival} no phantom can arise at all: a screen-passing phantom requires an $m \ge 2$ collision of support tones in one class (\Cref{prop:two-tone-phantom,prop:higher-order-phantom}), which survival excludes, and an all-zero bin is not a candidate-singleton (\Cref{def:exact-screen}). Surplus labels and measure-zero screen-passing phantoms are therefore a phenomenon of inputs outside the survival condition.
\end{proof}

\begin{remark}[Phantom labels leave the count intact; correctness stays with the verifier]
\label{rem:phantom-labels}
A screen-passing phantom bin (a multi-tone collision whose measured triple equals that of a single exponential at some $h \in [0,N)$, e.g.\ the $m \ge 2$ witnesses of \Cref{prop:two-tone-phantom,prop:higher-order-phantom}) emits at most one label $h$ by the total rule of \Cref{def:total-recovery}, which does not assume the bin is a singleton. Two consequences follow immediately from results already proved, and neither needs a separate argument. (Count unaffected.) $h$ occupies exactly one of the $\le p_\ell$ label slots of $\mathcal{G}_\ell$, so the bound of \Cref{thm:global-label-count}, whose proof is indifferent to whether a label is true or phantom, is preserved; the phantom may even survive the intersection, and still $|\mathcal{G}| \le p_1 = O(k)$. (Correctness preserved.) If a phantom label survives the global-label intersection and is used in the reconstructed candidate, the reconstruction differs from $X$, so by the contrapositive of \Cref{thm:verifier-soundness} the residual verifier rejects the candidate and the input is routed to the exact dense fallback (\Cref{thm:hybrid-correctness}).

Hence the engine supplies only the $O(k)$ count; which of the labels are true tones is decided by the sound verifier, not by the engine. In particular $\mathcal{G}$ is not phantom-free, and the phantom phenomenon is orthogonal to the count: a phantom is one label, not many, so it cannot inflate $|\mathcal{G}_\ell|$ beyond $p_\ell$. The phantom is a soundness object, and soundness is delivered by \Cref{thm:verifier-soundness}. The framework is therefore a deterministic $O(k)$-candidate guarantee with verifier-gated exactness.
\end{remark}

\begin{remark}[Algorithm 1 candidate-set semantics]
\label{rem:alg1-semantics}
\Cref{alg:multistage_oklogn} retains a global label $g$ in the stage set $\mathcal{G}_\ell$ only when bin $g \bmod p_\ell$ is screened as a candidate-singleton and the total per-bin rule of \Cref{def:total-recovery} emits $g$ at stage $\ell$; the per-stage sets are intersected, $\mathcal{G} = \bigcap_\ell \mathcal{G}_\ell$. A true frequency that collides in some stage $\ell$ is not guaranteed to emit its true label there: for generic coefficients the collided bin fails the screen, but on the measure-zero screen-passing slice the bin passes and emits a phantom label $\ne f$ (\Cref{prop:two-tone-phantom}). In either branch $f$ need not enter $\mathcal{G}_\ell$, so exclusion of a colliding tone from $\mathcal{G}$ is not certified stage-locally; correctness is recovered because an incomplete or incorrect reconstruction fails the final verifier of \Cref{thm:verifier-soundness} (or a sparse-attempt abort fires first), routing the input to the dense FFT fallback. The screen is one-sided, so the complementary case must be stated explicitly: a measure-zero screen-passing phantom (\Cref{prop:two-tone-phantom}) is admitted into $\mathcal{G}_\ell$ and is not excluded by the engine. Such a stage-local phantom label does not by itself invoke the verifier; if it survives the global-label intersection and is used in the reconstructed candidate, the residual verifier rejects the candidate (\Cref{rem:phantom-labels}). \Cref{thm:hybrid-correctness} guarantees exact recovery in all cases via verifier-gated acceptance or the dense fallback; the superseded accumulated-cascade analysis of \Cref{app:tuple-negative} is a contrasting negative result, not the operative path.
\end{remark}

\subsection{Conditional Sparse-Path Candidate-Construction Cost}

\noindent\emph{Informally:} $L = O(\log_k N)$ stages of $\Theta(k)$-sized FFTs cost $O(k \log N)$ in total and always yield at most $O(k)$ candidates; whether those candidates include the whole support is the separate containment question already settled by \Cref{thm:global-label-completeness}.

\begin{theorem}[Sparse-Path Candidate-Construction Cost in the Root-Index Oracle Model]\label{thm:sparse-path-complexity}
Under the engineered exact-product transform family of \Cref{def:transform-family} ($N = Q_L = \prod_{\ell=1}^{L} p_\ell$) and in the exact label-decoding (root-index oracle) model of \Cref{def:label-decoding-model}(a), \Cref{alg:multistage_oklogn} either constructs the global-label candidate set $\mathcal{G}$, of deterministically bounded size $|\mathcal{G}| = O(k)$ (\Cref{thm:global-label-count}), within $O(k \log N)$ time-domain samples and $O(k \log N)$ arithmetic operations, or triggers the dense fallback at Step~3 after spending no more than that sparse-attempt cost. Support containment under \Cref{cond:all-stage-survival} remains the separate \Cref{thm:global-label-completeness}; neither arm asserts it.
\end{theorem}

\begin{proof}
This is a cost accounting only; the candidate-count bound it quotes is \Cref{thm:global-label-count} and is not re-derived here, and containment of the support is the separate \Cref{thm:global-label-completeness}. The per-stage sample and transform costs are those already established for the decimated sampling operator in \Cref{thm:sampling-bluestein}; they are restated below only in the per-step form the multi-stage accounting needs.

Stage Count: the exact-product model $N = \prod_{\ell=1}^{L} p_\ell$ gives $L = O(\log N / \log k) = O(\log_k N)$.

Transform Size Per Stage: Each stage uses prime $p_\ell \in [k, ck]$ for an explicitly selected $c \ge 2$, so $m_\ell = \Theta(p_\ell) = \Theta(k)$. (The specific mixed-radix instantiation $m_\ell = p_\ell \times 64$ of the engineered batching convention is an implementation batching length outside the formal theorem; see \Cref{rem:64k-implementation}.)

Step-by-Step Complexity:
\begin{enumerate}
\item Per-Stage FFT: Each stage computes an FFT of size $\Theta(k)$, performed $S=3$ times, at cost $S \cdot O(m_\ell \log m_\ell) = O(k \log k)$ per stage.

\item All Stages (Sequential Work): Total sequential work across $L$ stages is
   \begin{align*}
   L \times O(k \log k) &= \frac{\log N}{\log k} \times O(k \log k) = O(k \log N).
   \end{align*}
   Parallel Wall-Clock Time: All $L$ stages can be computed in parallel, so with sufficient parallelism ($L$ cores), wall-clock time is $O(k \log k)$ per stage.

\item Singleton Detection: Per stage, test all $p_\ell \le ck = O(k)$ bins across $S=3$ shifts. Per bin: $O(S) = O(1)$ operations, giving $L \times O(k) = O(k \log N / \log k)$ in total.

\item Global-label assembly: Each candidate-singleton emits at most one full frequency label directly, by the total rule of \Cref{def:total-recovery} (for a genuine singleton it is the true label, \Cref{lem:global-label-recovery}; the co-prime increments pin $g$ modulo $N$, so no residue-tuple CRT reassembly is performed), decoded at $O(1)$ per bin in model (a) of \Cref{def:label-decoding-model} (rounded-$\arg$ guess, exact root-membership comparison, exact class-residue check; no root table), i.e.\ $O(k)$ per stage over the at most $\min(k, p_\ell)$ candidate-singletons, every one an occupied bin (\Cref{lem:decimation-preserves-sparsity}). The per-stage label sets $\mathcal{G}_\ell$, each of size $O(k)$ by \Cref{thm:global-label-count}, are combined by intersection $\mathcal{G} = \bigcap_\ell \mathcal{G}_\ell$ over the $L$ stages: $O(L) = O(\log N / \log k)$ per candidate, hence $O(k \log N / \log k)$ for $O(k)$ candidates.

\item Coefficient Estimation: For each candidate label, read the closed-form zero-shift estimate $\hat{X}[g] = (N/p_\ell)\, X_{p_\ell}^{(0)}[g \bmod p_\ell]$ (\Cref{subsec:sampling-operator}, \Cref{lem:global-label-recovery}(4)) from one designated screened stage; a single estimate per candidate is all the cost bound requires, giving $O(kL)$ in total. For a genuine singleton the $S = 3$ shifted observations agree exactly after conjugate de-rotation $e^{-j2\pi f\sigma_v/N}$ (\Cref{thm:singleton-observation-identity}) and the per-stage estimates agree across the $L$ stages after the $N/p_\ell$ rescale, so for such bins the estimates may be averaged for numerical robustness; a screen-passing phantom carries no such cross-stage agreement guarantee (the screen is one-sided, \Cref{rem:one-sided-normative}), which is a soundness matter, settled downstream by the verifier rather than by this estimate, and not a cost matter.
\end{enumerate}

Total (completed arm): $O(k \log N) + O(k \log N / \log k) + O(k \log N) = O(k \log N)$ arithmetic operations, on $3\sum_\ell p_\ell \le 3Lck = O(k \log N)$ sample reads. The size of the constructed set is $|\mathcal{G}| = O(k)$ by \Cref{thm:global-label-count}, which holds for every input; the accounting above is valid whether or not the labels are true.

Fallback arm (the spend bound). Every charge above is an input-independent per-stage bound: the three shift-FFTs, the $p_\ell$-bin pre-screen, the at most $\min(k, p_\ell)$ label decodes (a nonzero-bin count bounded by \Cref{lem:decimation-preserves-sparsity} for every input), and the $O(k)$-operand intersection are each incurred at most once per executed stage, and the budget-guard comparisons cost $O(1)$ per stage given the maintained set sizes. If Step~3 dispatches the dense fallback at some stage $\ell_0 \le L$ (a stage emitting no class-consistent global label, or a budget guard firing), the sparse attempt has executed only stages $1, \dots, \ell_0$ and never reaches Steps~4 and~5, so its spend is a partial sum of the same per-stage charges, bounded by the full-chain totals: at most $O(k \log N)$ samples and $O(k \log N)$ arithmetic before the dense path begins. The dense fallback's own $O(N \log N)$ cost is charged in the hybrid correctness accounting, not here.
\end{proof}

\begin{theorem}[Sparse-Path Candidate-Construction Cost in the Comparison Real-RAM Model]\label{thm:cost-comparison-model}
Under the same transform family (\Cref{def:transform-family}) and in the comparison real-RAM model of \Cref{def:label-decoding-model}(b), with no unit-cost rounding or root-evaluation oracle, \Cref{alg:multistage_oklogn} either constructs the same candidate set $\mathcal{G}$, of deterministically bounded size $|\mathcal{G}| = O(k)$ (\Cref{thm:global-label-count}), within $O(k \log N)$ time-domain samples and
\[
O\!\left(\frac{k \log^2 N}{\log k}\right)
\]
arithmetic operations, or triggers the dense fallback at Step~3 after spending no more than that sparse-attempt cost. This bound assumes nothing about root-index decoding beyond exact real-RAM arithmetic, exact comparisons, and exact $\arg(\cdot)$.
\end{theorem}

\begin{proof}
The stage count obeys $L \le \log N / \log k$, from $k^L \le N$ (\Cref{def:transform-family}). Every component of the accounting in \Cref{thm:sparse-path-complexity} other than the per-bin label decode uses neither removed primitive, so its cost is unchanged and is quoted rather than re-derived: per-stage FFT work $O(k \log k)$ over $L$ stages totals $O(k \log N)$; screening all $p_\ell \le ck$ bins costs $O(k)$ per stage, totaling $O(k \log N / \log k)$; the per-candidate intersection costs $O(L)$ over the $O(k)$ candidates, totaling $O(k \log N / \log k)$; coefficient reads cost $O(kL)$. Sample reads involve no decoding arithmetic, so the sample count is that of \Cref{thm:sparse-path-complexity}: $3 \sum_{\ell} p_\ell \le 3 L c k = O(k \log N)$.

For the decode, \Cref{def:label-decoding-model}(b) charges each candidate-singleton $O(\log N)$: the nearest-index binary search costs $O(\log N)$ exact comparisons, the repeated-squaring evaluation of $\omega^{g_0}$ costs $O(\log N)$ exact multiplications, and the residue check $g_0 \equiv r \pmod{p_\ell}$ remains $O(1)$ exact integer arithmetic. A candidate-singleton is occupied, hence nonzero, and each stage has at most $\min(k, p_\ell) \le k$ nonzero bins (\Cref{lem:decimation-preserves-sparsity}), so decoding costs $O(k \log N)$ per stage and
\[
L \cdot O(k \log N) \;\le\; \frac{\log N}{\log k} \cdot O(k \log N) \;=\; O\!\left(\frac{k \log^2 N}{\log k}\right)
\]
in total. Since $\log N \ge \log k$ for every admissible instance ($N \ge k$), every other term satisfies $O(k \log N) \subseteq O(k \log^2 N / \log k)$, so the decode term dominates and the arithmetic total is $O(k \log^2 N / \log k)$. The fallback arm is bounded exactly as in \Cref{thm:sparse-path-complexity}: every charge, including the $O(\log N)$-per-bin decode, is an input-independent per-executed-stage bound, so a Step~3 dispatch at stage $\ell_0 \le L$ spends a partial sum of the same charges, at most the full-chain totals above.
\end{proof}

\begin{remark}[Relation between the two cost theorems]
\label{rem:cost-model-relation}
\Cref{thm:cost-comparison-model} states the cost once the unit-cost decoding primitives are withdrawn (its transform-family hypothesis is exactly that of \Cref{thm:sparse-path-complexity}); adopting the three unit-cost primitives of \Cref{def:label-decoding-model}(a) strengthens its $O(k \log^2 N / \log k)$ to the $O(k \log N)$ of \Cref{thm:sparse-path-complexity}, which is why the standing convention of \Cref{def:label-decoding-model} ties every $O(k \log N)$ figure to model (a). Neither theorem is a bit-complexity statement: both models retain exact real arithmetic and unit-cost exact $\arg(\cdot)$, and the gap between them measures the root-index read alone (\Cref{def:label-decoding-model}(b)).
\end{remark}

\noindent\emph{Scope of the conditional qualifier.} Both cost theorems bound the cost of candidate construction only, and their relationship, together with the fact that neither is a bit-complexity statement, is \Cref{rem:cost-model-relation}. The $O(k)$ candidate count they quote is the proved \Cref{thm:global-label-count} and holds for every input. What is conditional is the usefulness of the constructed set, not its size: containment $\mathcal{S} \subseteq \mathcal{G}$ requires \Cref{cond:all-stage-survival} (\Cref{thm:global-label-completeness}), and the screen is one-sided (\Cref{rem:one-sided-normative}), so a measure-zero screen-passing phantom (\Cref{prop:two-tone-phantom}) may be admitted. Exactness is therefore verifier-gated: a candidate is returned only on a PASS of the sound final verifier (\Cref{thm:verifier-soundness}, \Cref{def:verifier}) and is otherwise routed to the dense $O(N \log N)$ fallback (\Cref{thm:hybrid-correctness}). The conditional-sparse-path qualifier in the abstract refers to this operating envelope for sparse completion (engineered divisor compatibility, \Cref{def:transform-family}, together with all-stage genuine-singleton survival, \Cref{cond:all-stage-survival}), and never to randomized sampling, randomized stage selection, or any algorithmic coin flip.

\subsection{Hybrid Unconditional Correctness}

\noindent\emph{Informally:} the algorithm never trusts the fast path; an accepted output has passed a sound exactness proof, and everything else is recomputed densely.

\begin{theorem}[Hybrid Unconditional Correctness]
\label{thm:hybrid-correctness}
In the noiseless, on-grid, exact-arithmetic model of \Cref{def:signal-model}, for every input $x \in \mathbb{C}^N$ that is at most $k$-sparse on-grid, the hybrid algorithm of \Cref{alg:multistage_oklogn} returns the exact spectrum $X$, by exactly one of two deterministic branches:
\begin{enumerate}
\item Verified sparse path (PASS). If the sound final verifier of \Cref{def:verifier}, reading $k + |\mathcal{G}|$ consecutive residual samples on the global-label candidate set $\mathcal{G}$ of \Cref{thm:global-label-count}, returns $V(x,\hat{X}) = \mathrm{PASS}$ on the reconstructed sparse candidate $\hat{X}$, then $\hat{X} = X$ exactly (\Cref{thm:verifier-soundness}), and $\hat{X}$ is returned.
\item Dense fallback (FAIL or out-of-envelope). Otherwise, on a verifier FAIL, on any sparse-attempt abort (a stage emits no class-consistent global label, or the global-label budget is exceeded, namely $|\mathcal{G}_\ell| > ck$ at some stage or, once the stage chain is complete, $|\mathcal{G}| > k$ on the accumulated set), or on any violation of the operating envelope for sparse completion, the algorithm dispatches the standard $O(N \log N)$ dense FFT, which returns $X$ exactly on every input of the model.
\end{enumerate}
Exactly one output is accepted (the sparse output on verifier PASS, otherwise the dense-fallback output), the routing test (the verifier of \Cref{def:verifier} together with the existing screen and budget checks) is deterministic, and both branches output $X$; hence correctness holds with probability one, unconditionally within the noiseless, on-grid, at-most-$k$, exact-arithmetic model of \Cref{def:signal-model} and independent of the fast-path survival condition and of candidate-construction success. In particular no probabilistic, association, or verifier-existence assumption is needed for correctness, and the total runtime is bounded above by $O(N \log N)$: a rejected sparse attempt incurs at most its candidate-construction and, if reached, verification cost, $O(k \log N) + O(k^2)$ in model (a) (\Cref{def:label-decoding-model}), after which the dense fallback runs in $O(N \log N)$, and the sum remains $O(N \log N)$, because the model's standing $L \ge 2$ gives $N \ge k(k+1) > k^2$ (\Cref{def:transform-family}), so both $k^2$ and $k \log N$ are $O(N)$. The displayed sparse-attempt figure is the model-(a) bound; model (b) is handled analogously, a rejected attempt costing $O(k \log^2 N / \log k) + O(k^2)$ (\Cref{thm:cost-comparison-model}, \Cref{prop:verifier-cost}), and since $k < \sqrt{N}$ both terms are again $O(N \log N)$, leaving the total unchanged.

\noindent\emph{Cost and scope of the guarantee.} The verifier discharges correctness only, not the runtime conditions: the sparse path, when taken and verified, uses $O(k \log N)$ samples and $O(k \log N) + O(k^2)$ arithmetic in model (a) of \Cref{def:label-decoding-model} ($O(k \log^2 N / \log k) + O(k^2)$ in the comparison real-RAM model, \Cref{thm:cost-comparison-model}), with the $O(k^2)$ verifier term stated separately and not folded into $O(k \log N)$ (\Cref{prop:verifier-cost}). A sparse attempt spends at most $O(k \log N)$ under the transform family alone, whether it completes candidate construction or dispatches the fallback (\Cref{thm:sparse-path-complexity}); what additionally requires all-stage genuine-singleton survival (\Cref{cond:all-stage-survival}) is containment of the support (\Cref{thm:global-label-completeness}), and hence the sparse path completing rather than routing to the dense fallback, so the end-to-end $O(k \log N)$ sparse-path runtime remains conditional on the full operating envelope for sparse completion; the $O(k)$ candidate count it requires is the proved \Cref{thm:global-label-count}, so the runtime is not conditional on any bounded-association hypothesis. Only correctness is unconditional. The PASS-branch exactness is supplied by the proved \Cref{thm:verifier-soundness}.

\noindent\emph{Applicability range.} The availability constraint $L \le \pi([k,ck])$ in \Cref{def:transform-family} is non-restrictive at every sparsity: the prime sets cap the number of stages, so the largest engineered length is $N = \prod_{\ell=1}^{L} p_\ell$ over the available primes. For $k = 10$ the four primes $\{11,13,17,19\}$ give $N = 46{,}189$ at $L=4$; at the intended-regime boundary $k = 20$, $\{23,29,31,37\}$ give $N = 765{,}049$, and at $k = 50$ ten primes in $[50,100]$ allow much larger $N$ (admissible instances at and beyond the regime boundary; \Cref{rem:engineering-regime}, \Cref{tab:prime-availability}). Widening the prime interval (larger $c$) admits further stages.
\end{theorem}

\begin{proof}
PASS branch (exact). When the sparse path completes and the final verifier of \Cref{def:verifier} returns PASS, the candidate set is the global-label set $\mathcal{G} = \bigcap_\ell \mathcal{G}_\ell$ emitted by the total per-bin rule of \Cref{def:total-recovery} and bounded in size by the proved \Cref{thm:global-label-count}; its budget precondition $|\mathcal{G}| \le k$ is enforced before $V$ is invoked (\Cref{def:verifier}). Each candidate-singleton emits at most one in-range, on-grid global frequency label; for a genuine singleton that label is its exact full frequency $g = f \in [0, N)$, recovered directly from the de-rotated three-shift phasors by \Cref{lem:global-label-recovery} (the co-prime increments $\{31,32\}$ pin $g$ modulo $N$, so no residue-tuple CRT reassembly is invoked), and the singleton coefficient is read in closed form as $\hat{X}[g] = \tfrac{N}{p_\ell} z_0$ (\Cref{lem:global-label-recovery}(4)), exact in the noiseless setting. By \Cref{thm:verifier-soundness} a PASS verdict implies $\hat{X} = X$; if instead a screen-passing phantom survived the intersection and was used in the reconstruction, the reconstruction differs from $X$, the residual is nonzero, $V$ returns FAIL, and the input is routed to the dense fallback rather than returned (\Cref{rem:phantom-labels}).

FAIL / fallback branch (exact). An individual bin failing the pre-screen, or a candidate-singleton emitting no class-consistent global label under the total per-bin rule of \Cref{def:total-recovery} (a screen-inconsistent bin), is a normal per-bin event and does not by itself trigger fallback.

The operative routing rule is stage- and chain-level: if some stage emits no class-consistent global label, or the global-label budget is exceeded ($|\mathcal{G}_\ell| > ck$ at some stage or, once the stage chain is complete, $|\mathcal{G}| > k$ on the accumulated set), or the final verifier of \Cref{def:verifier} does not return PASS, or the input falls outside the operating envelope for sparse completion, the algorithm routes to the dense-FFT fallback (the screen and budget conditions are detected deterministically at Step 3 of \Cref{alg:multistage_oklogn}; the verifier condition at Step 5). The dense FFT computes all $N$ coefficients in $O(N \log N)$ deterministic time, whose correctness is classical, producing the exact spectrum unconditionally.

\needspace{4\baselineskip}\emph{Unconditional, deterministic correctness.} Exactly one of the two branches completes successfully; both produce the exact spectrum. The screen (candidate-singleton screening plus the global-label class-consistency check of \Cref{def:total-recovery}), the budget guard, and the verifier are each deterministic, so the routing decision is deterministic. The PASS branch's exactness rests on the proved \Cref{thm:verifier-soundness} rather than on an assumption, so correctness holds with probability one with no probabilistic or association hypothesis, independent of the sparse-path completion profile, which is characterized only conditionally by \Cref{cond:all-stage-survival} and \Cref{thm:global-label-completeness}. The cost figures are \Cref{prop:verifier-cost} added to \Cref{thm:sparse-path-complexity}.
\end{proof}

\paragraph{Result hierarchy.} The screened framework's correctness claims are organized in three tiers:
\begin{itemize}
  \item All inputs of the model (\Cref{def:signal-model}: noiseless, on-grid, at most $k$ true tones, exact arithmetic): exact recovery via the sparse path or dense fallback (\Cref{thm:hybrid-correctness}). No support-structure assumption beyond the at-most-$k$ model; for spectra with more than $k$ tones the verifier's residual-support bound does not apply and no claim is made.
  \item Completed sparse attempts (every stage emits a class-consistent global label and both budgets hold): candidate construction completes in $O(k \log N)$ arithmetic operations under the engineered exact-product transform family and in model (a), and in $O(k \log^2 N / \log k)$ in model (b) (\Cref{thm:sparse-path-complexity}, \Cref{thm:cost-comparison-model}, \Cref{def:label-decoding-model}); verification adds $O(k^2)$ in model (a) and $O(k^2 + k \log N)$ in model (b) (\Cref{prop:verifier-cost}).
  \item Aborted sparse attempts (collision-heavy or adversarial inputs): dense fallback returns the exact spectrum in $O(N \log N)$ time, after a sparse-attempt spend bounded as in \Cref{thm:sparse-path-complexity} (\Cref{thm:hybrid-correctness}).
\end{itemize}
The superseded accumulated cascade (\Cref{app:tuple-negative}) is well-defined for all inputs, but the operative sparse-path complexity bound, carried by the global-label engine, completes its construction arm only on completed sparse attempts.

\needspace{12\baselineskip}\section{Implementation and Computational Considerations}
\label{sec:computational}

\subsection{Parameter Selection for Multi-Stage Algorithm}

Reference Stage Configuration:\\
For signal length $N$ and sparsity $k$:
\begin{enumerate}
\item Stage Count: $L = O(\log_k N)$ with the exact-product model $N = \prod_{i=1}^L p_i$ by construction (with $k^L \leq N \leq (ck)^L$ when $p_i \in [k,ck]$)
\item Prime Selection: Choose primes $p_i \in [k, ck]$ for $c \ge 2$ (special case $c=2$) for each stage, ensuring pairwise coprimality
\item Implementation Batching Length: Use mixed-radix FFTs of length $m_i = p_i \times 2^6 \in [64k, 64ck]$ ($[64k, 128k]$ in the special case $c = 2$); the formal per-stage transform remains the $p_i$-point DFT (\Cref{rem:64k-implementation})
\item Shift Count: $S = 3$ shifts per stage for the magnitude-equality (normalized-MAD) test (constant overhead)
\end{enumerate}

\begin{definition}[Normalized MAD Singleton Test (Implementation Surrogate)]
\label{def:cv-singleton-test}
(Implementation surrogate.) The fielded implementation of the occupied equal-magnitude pre-screen of \Cref{def:exact-screen}(1) computes, for bin $r$ with amplitudes $A = \{|X_M[r; 0]|, |X_M[r; 31]|, |X_M[r; 63]|\}$ across the $S = 3$ co-prime-increment shifts $\{0, 31, 63\}$ (where $\gcd(31, 32) = 1$), the normalized median absolute deviation
\begin{equation}
\mathrm{nMAD}(r) = \frac{\text{MAD}(A)}{\text{median}(A)}
\end{equation}
where $\text{MAD}(A) = \text{median}(\{|a - \text{median}(A)| : a \in A\})$, and treats bin $r$ as a candidate-singleton if $\mathrm{nMAD}(r) < \tau_{cv}$ (the tolerance $\tau_{cv}\in(0,1)$ is chosen by the implementer; this paper calibrates no value for it, reports no measurements bearing on it, and no formal screen statement uses it). $\mathrm{nMAD}(r) = 0$ is a necessary condition for exact magnitude equality, not an equivalence: $\mathrm{MAD} = 0$ does not force constancy (e.g.\ the multiset $A = \{a, a, b\}$ with $a \ne b$ has $\mathrm{MAD}(A) = 0$), so the surrogate is not itself an exact-equality test (\S\ref{sec:adversarial} exhibits a three-tone collision admitted by the surrogate and rejected by the exact screen). The occupancy convention of \Cref{def:exact-screen} applies to the surrogate as well: a bin with $\mathrm{median}(A) = 0$ (in particular the all-zero bin) is declared non-candidate outright rather than evaluated through a zero denominator, and a surrogate-admitted candidate-singleton is handed to the ratio formation of \Cref{def:total-recovery} only when all three observations are nonzero, $z_0, z_{31}, z_{63} \ne 0$ (the \emph{all-observations-nonzero guard}); a bin failing the guard emits no label. Normalized MAD alone is insufficient: the operative formal screen is the exact ratio/root screen of \Cref{def:exact-screen}, whose remaining conditions are executed by the total rule of \Cref{def:total-recovery}.
\end{definition}

\noindent\emph{Why the guard is needed on the surrogate path.} $\mathrm{MAD}$-based acceptance does not force nonzero observations: the magnitude triple $(1, 0, 1)$ has median $1$ and $\mathrm{MAD} = 0$, hence $\mathrm{nMAD} = 0 < \tau_{cv}$, while the ratio $u_{32} = z_{63}/z_{31}$ of \Cref{def:total-recovery} would be undefined. A guard failure routes as any other no-label bin (\Cref{thm:hybrid-correctness}). The exact pre-screen needs no such guard, since it already forces $|z_0| = |z_{31}| = |z_{63}| > 0$ on every passing bin (\Cref{def:exact-screen}).

\begin{remark}[Implementation surrogates for the formal screen and transform length]
\label{rem:impl-surrogates}
The formal singleton screen of \Cref{tab:oklogn-assumptions} is stated with exact predicates (exact magnitude equality and exact ratio/root-membership conditions, \Cref{def:exact-screen}); a fielded implementation replaces them with finite-precision surrogates that are not part of the formal model: the normalized-MAD magnitude test of \Cref{def:cv-singleton-test} with threshold $\mathrm{nMAD} < \tau_{cv}$ and an unwrapped-affine-phase regression test with $R^2 > 1 - \varepsilon$ (\Cref{def:phase-unwrapping}), with an optional frequency-consistency cross-check, and with the all-observations-nonzero guard of \Cref{def:cv-singleton-test} preceding any ratio or phase formation. Likewise the formal per-stage transform is the $p_\ell$-point DFT (\Cref{subsec:sampling-operator}); a fielded implementation may batch it as a mixed-radix FFT of length $m_\ell = p_\ell \times 2^6 = 64 p_\ell$ (the design-stage batching convention of \Cref{rem:64k-implementation}). The tolerances $\tau_{cv}$ and $\varepsilon$ are implementer-chosen and are deliberately left unvalued here: this paper reports no calibration data or measurements that would fix them. They and the $64$ batching factor are engineering choices outside the formal cost claims of \Cref{thm:sparse-path-complexity,thm:hybrid-correctness}; they do not change the residue modulus $p_\ell$, the proved $O(k)$ candidate count (\Cref{thm:global-label-count}), or verifier soundness (\Cref{thm:verifier-soundness}).
\end{remark}

\begin{remark}[Implementation note: the $64 p_\ell$ mixed-radix factor is not part of the theorem]\label{rem:64k-implementation}
The formal sparse-path bound of \Cref{thm:sparse-path-complexity} requires only that per-stage transform sizes satisfy $m_\ell = \Theta(p_\ell) = \Theta(k)$. The engineered batching convention adopted in this paper (a design-stage choice; no implementation or measurements are reported) sets $m_\ell = p_\ell \times 2^6 = 64 p_\ell$ for mixed-radix butterflies (radix-$p_\ell$ outer stage, radix-$2$ inner stages); the constant $64$ is a nominal cache-alignment design rationale and is not part of the formal cost bound. Any other implementation choice with $m_\ell = \Theta(p_\ell)$ (for instance, $m_\ell = p_\ell \times 2^t$ with $t \in \{4, 5, 6, 7\}$, or pure radix-$p_\ell$ Bluestein) yields the same $O(k \log N)$ conditional sparse-path complexity. In particular, the formal per-stage transform of \Cref{subsec:sampling-operator} is the $p_\ell$-point DFT of the decimated operator; the $64 p_\ell$ mixed-radix length is an implementation batching of that $p_\ell$-point transform, never the formal DFT length.
\end{remark}

The adaptive-moduli scaling extension (an $N$-dependent choice of the moduli scaling factor $\theta$ and view count $V$ for the keyed multi-view reconstruction, together with its illustrative occupancy/yield design targets) sits outside the formal divisor-compatibility envelope and is collected, as a non-theorem implementation extension, in \Cref{app:keyed-multiview}.

\needspace{6\baselineskip}\emph{Example 1 (beyond the intended engineering regime, \Cref{rem:engineering-regime}: $k = 50$ is an admissible instance of \Cref{def:transform-family}; shown for parameter-selection illustration):} $N_{\mathrm{target}} = 10^6$, $k = 50$
\begin{itemize}
\item Stage count: $L = \lceil \log_2(10^6) / \log_2(50) \rceil = \lceil 19.93 / 5.64 \rceil = 4$ stages
\item Primes in $[50, 100]$: $p_1 = 53$, $p_2 = 59$, $p_3 = 61$, $p_4 = 67$
\item Product: $M = 53 \times 59 \times 61 \times 67 = 12{,}780{,}049 \ge N_{\mathrm{target}}$. The engineered transform length is the exact product $N_{\mathrm{eng}} = 12{,}780{,}049$ (matching the $k{=}50$, $L{=}4$ entry of \Cref{tab:prime-availability}), not the nominal target $N_{\mathrm{target}} = 10^6$.
\item Implementation batching lengths: $m_1 = 3{,}392$, $m_2 = 3{,}776$, $m_3 = 3{,}904$, $m_4 = 4{,}288$ (avg: 3,840)
\item Total candidate-construction arithmetic complexity: $4 \times O(k \log k)$ stages $= O(k \log N_{\mathrm{eng}})$ in the oracle model and $O(k \log^2 N_{\mathrm{eng}} / \log k)$ in the comparison model (\Cref{def:label-decoding-model}); the sample count is $O(k \log N_{\mathrm{eng}})$ in both models
\end{itemize}

\emph{Example 2:} $N_{\mathrm{target}} = 10^6$, $k = 15$
\begin{itemize}
\item Stage count: choose $L$ stage primes so the engineered length $N_{\mathrm{eng}} = \prod_{\ell=1}^{L} p_\ell$ first exceeds the target $N_{\mathrm{target}} = 10^6$; here $L = 5$, since the partial products reach $215{,}441 < 10^6$ at $L{=}4$ and $6{,}678{,}671 \ge 10^6$ at $L{=}5$. The transform length is then the exact product $N_{\mathrm{eng}} = 6{,}678{,}671$.
\item Primes in $[15, 45]$ (widened interval to ensure availability): $p_1 = 17$, $p_2 = 19$, $p_3 = 23$, $p_4 = 29$, $p_5 = 31$. We widen to $[k,3k]$ here to satisfy $\pi([k,3k]) \geq L$; this is consistent with the ``widen interval'' option noted earlier.
\item Product: $M = 17 \times 19 \times 23 \times 29 \times 31 = 6{,}678{,}671 \ge N_{\mathrm{target}}$. The engineered transform length is $N_{\mathrm{eng}} = 6{,}678{,}671$ (matching the $k{=}15$, $L{=}5$ entry of \Cref{tab:prime-availability}), a valid choice under the stated count-triggered widening rule (the widening to $[k,3k]$ is licensed here because $\pi([15,30]) = 4 < L = 5$); no minimality among admissible prime sets is claimed.
\item Implementation batching lengths: $m_i \in [1{,}088, 1{,}984]$ (exactly $64 \times [17, 31]$; note $1{,}984 > 128k = 1{,}920$, an instance of the general $[64k, 64ck]$ range at $c = 3$)
\item Total candidate-construction arithmetic complexity: $O(k \log N_{\mathrm{eng}})$ in the oracle model and $O(k \log^2 N_{\mathrm{eng}} / \log k)$ in the comparison model (\Cref{def:label-decoding-model}); the sample count is $O(k \log N_{\mathrm{eng}})$ in both models, with very small per-stage FFTs at $k=15$
\end{itemize}

\subsection{Shift Count Selection}

\emph{Design note.} The formal screen uses three observations to obtain two independent modular phase ratios. The core noiseless screen therefore uses $S = 3$ co-prime-increment shifts $\{0, \Delta_1, \Delta_1 + \Delta_2\}$ with $\gcd(\Delta_1, \Delta_2) = 1$. A larger shift count $S \in [4, 8]$ is an implementation-time design parameter; selecting a specific $S$ for noisy or off-grid operation is outside the noiseless model of this paper and is not characterized here.

\subsection{Deterministic Set Intersection}

\noindent The accumulation $\mathcal{G} \leftarrow \mathcal{G} \cap \mathcal{G}_\ell$ of Step~3 of \Cref{alg:multistage_oklogn} is a deterministic intersection of integer sets and needs no hashing. Each stage emits its labels into a sorted integer array, sorted once at $O(|\mathcal{G}_\ell| \log |\mathcal{G}_\ell|)$ or maintained sorted by insertion, and the accumulated set $\mathcal{G}$ is kept sorted by the same convention. Two sorted arrays are then intersected by a single merge pass with two indices, at cost $O(|\mathcal{G}| + |\mathcal{G}_\ell|)$ comparisons and no auxiliary storage beyond the output.

\noindent By \Cref{thm:global-label-count} both operands are $O(k)$ at every stage ($|\mathcal{G}_\ell| \le p_\ell \le ck$ and $|\mathcal{G}| \le p_1 \le ck$), so each stage's intersection costs $O(k)$ and the $L$ intersections cost $O(kL) = O(k \log N / \log k)$ in total, exactly the intersection charge already carried in the proofs of \Cref{thm:sparse-path-complexity,thm:cost-comparison-model}; the $L$ sorts add $O(k \log k)$ per stage, absorbed by the per-stage $O(k \log k)$ transform charge. This is an implementation fact consistent with the count theorem, not an additional formal claim: the $O(k)$ operand sizes are supplied by \Cref{thm:global-label-count}, and nothing here re-derives or strengthens them.

\subsection{Root-Index Decoding}

\noindent The one step of the sparse path whose cost depends on the machine model is the per-bin root-index read $g = \log_\omega w$ of \Cref{def:total-recovery}: it costs $O(1)$ under the unit-cost primitives of model (a) of \Cref{def:label-decoding-model} and $O(\log N)$ in the comparison model (b), which is exactly the gap between \Cref{thm:sparse-path-complexity} and \Cref{thm:cost-comparison-model}. A fielded implementation sits in model (b) unless it can supply the unit-cost primitives; the practical discussion, including the two implementation postures excluded by design (a $\Theta(N)$ precomputed root table, and a floating-point $\arg$ read in place of the exact emission comparisons), is collected in \Cref{app:cost-models}.

\subsection{Memory Complexity}

\begin{theorem}[Space Complexity]
\label{thm:space-complexity}
The multi-stage global-label algorithm uses $O(k + L) = O(k + \log N / \log k)$ working memory in sequential operation, and $O(L k) = O(k \log N / \log k)$ when all $L$ stages are retained in parallel.
\end{theorem}

\begin{proof}
Sequential Processing (Reusing Buffers): The global-label engine (\Cref{subsec:global-label-engine}) never materializes residue tuples or a CRT-reconstruction workspace; its working set is
\begin{itemize}
\item FFT computation buffer: $O(k)$ for the largest stage transform $m_\ell = \Theta(p_\ell) = \Theta(k)$, reused across stages;
\item the accumulated global-label set $\mathcal{G}$ with $|\mathcal{G}| \le p_1 \le ck = O(k)$ (\Cref{thm:global-label-count}), plus, while a single stage is processed, the per-stage label set $\mathcal{G}_\ell$ with $|\mathcal{G}_\ell| \le p_\ell = O(k)$, which is intersected into $\mathcal{G}$ and then discarded;
\item the $L$ stage indices / primes: $O(L)$.
\end{itemize}
The total sequential working memory is therefore $O(k) + O(L) = O(k + L) = O(k + \log N / \log k)$; no per-stage residue lists and no incremental-CRT workspace are stored.

Parallel Processing (No Buffer Reuse): If all $L$ stages are computed concurrently, each retains its own $O(k)$ FFT buffer and $O(k)$ label set $\mathcal{G}_\ell$, giving $O(L k) = O(k \log N / \log k)$ total.

For small $k$ (the regime of \Cref{rem:engineering-regime}), both figures remain modest even at large engineered lengths (e.g., at the regime boundary $k=20$, the $L=4$ engineered instance $N = 765{,}049$ of \Cref{tab:prime-availability} gives a sequential working set on the order of $k + L \approx 24$ entries); for larger $k$ both grow only as the stated $O(k + \log N/\log k)$ and $O(k \log N/\log k)$ bounds.
\end{proof}

\section{Adversarial Examples and the Dense Fallback}
\label{sec:adversarial}

\subsection{Adversarial-Support Example: A Collision that Passes the nMAD Surrogate and is Rejected at Root-Membership}

To exercise the sparse/dense hybrid, take $N = 46{,}189 = 11 \cdot 13 \cdot 17 \cdot 19$ (engineered transform length), $k = 4$, primes $\{11, 13, 17, 19\}$ ($L = 4$), and shifts $\{0, 31, 63\}$. Screen convention for this example: shifts $\{0,31,63\}$; formal screen is the exact ratio/root screen of \Cref{def:exact-screen} (occupied equal-magnitude bin, $D_{\max} = 0$, root membership, class consistency), with $\mathrm{nMAD} < \tau_{cv}$ as the finite-precision implementation surrogate; label emission follows the total rule of \Cref{def:total-recovery} executed by the costed procedure of \Cref{def:label-decoding-model}. Under this shift set, measure-zero two-tone screen-passing phantoms are possible (\Cref{prop:two-tone-phantom}), unlike consecutive-shift conventions with different blind spots (\Cref{rem:consecutive-shifts-blind-spots}). Suppose the support is the residue-collision set
\begin{equation*}
\begin{aligned}
\mathcal{S} = \{\;& f_1 = 11 \cdot 13 = 143, \quad f_2 = 11 \cdot 17 = 187, \\
                 & f_3 = 13 \cdot 17 = 221, \quad f_4 = 11 \cdot 13 \cdot 17 = 2431 \,\},
\end{aligned}
\end{equation*}
which is constructed so that every pair shares at least one residue modulo a stage prime. Assign equal coefficients $X[f] = 1$ for all four support frequencies. At stage $\ell = 1$ ($p = 11$), the residues are $f_1 \bmod 11 = 0$, $f_2 \bmod 11 = 0$, $f_3 \bmod 11 = 1$, $f_4 \bmod 11 = 0$: bin $r = 0$ contains $\{f_1, f_2, f_4\}$ (a $3$-tone collision) while bin $r = 1$ contains the singleton $\{f_3\}$.

\emph{The three complex bin-$0$ observations (PLUS convention).} The bin-$0$ observation is the multi-tone phasor sum $z_\sigma \coloneqq X_{p_1}^{(\sigma)}[0] = \tfrac{p_1}{N}\sum_{f \in \{143,187,2431\}} X[f]\, e^{+j 2\pi f \sigma / N}$, with the common prefactor $p_1/N = 11/46{,}189$. Under the PLUS convention the three shifted observations are, numerically,
\begin{equation*}
\begin{aligned}
z_0    &= 7.1446\times 10^{-4}, \\
z_{31} &= 2.0271\times 10^{-4} + 1.2879\times 10^{-4}\,j, \\
z_{63} &= -2.263\times 10^{-5} + 6.8022\times 10^{-4}\,j,
\end{aligned}
\end{equation*}
whose inner phasor-sum magnitudes (after dividing out $p_1/N$) are $\bigl| \sum_f e^{+j 2\pi f \sigma / N} \bigr| = (3.000,\; 1.008,\; 2.858)$ at $\sigma \in \{0,31,63\}$.

\emph{The formal screen rejects it; only the implementation surrogate admits it.} Which test admits this bin must be stated precisely. The formal exact equal-magnitude condition of \Cref{def:exact-screen} is the vanishing normalized range $D_{\max}(0) = (\max_\sigma|z_\sigma| - \min_\sigma|z_\sigma|)/\max_\sigma|z_\sigma|$; here $D_{\max}(0) \approx 0.664 \ne 0$, so the bin fails the formal exact-equality screen (a true singleton would pass with $D_{\max} = 0$). Under the formal screen this bin is therefore rejected before any candidate admission; it is not admitted by the formal model. What admits it is only the implementation nMAD surrogate: the median-based surrogate evaluates to $\mathrm{nMAD}(0) \approx 0.050$, so the bin passes the surrogate, and is not rejected by the implementation magnitude test, at every tolerance $\tau_{cv} > 0.050$. This is the surrogate's weakness made explicit: a genuine three-tone collision survives the $\mathrm{nMAD} < \tau_{cv}$ surrogate at any such tolerance even though it fails the formal $D_{\max} = 0$ screen. The example therefore motivates the formal exact-equality screen over the $\mathrm{nMAD}$ surrogate, and, because the surrogate (not the formal screen) admits it, motivates the sound verifier as the actual guarantor in the fielded implementation.

\emph{The B\'ezout combination is not even on the unit circle: the total rule rejects at root-membership.} Suppose the surrogate-admitted bin $0$ (the formal screen having already rejected it) is handed to the global-label engine. Forming the per-step ratios $u_{31} = z_{31}/z_0$, $u_{32} = z_{63}/z_{31}$ and the B\'ezout combination $w = u_{32}\,u_{31}^{-1}$ (which does not assume the bin is a singleton) gives, at high precision, $w = 7.5083 + 3.8338\,j$ with $|w| \approx 8.4305$. Since $|w| \ne 1$, $w$ is not an $N$th root of unity and no integer $g$ with $w = \omega^{g}$ exists: condition~1 of \Cref{def:total-recovery} (exact root-membership) fails, so the bin emits no label, before the class-residue check is ever reached. The costed decoding procedure of \Cref{def:label-decoding-model} still produces the rounded phase index $g_0 = \operatorname{round}(N \arg(w)/2\pi) = 3471$, but this is a guess only, discarded by the exact root-membership comparison $\omega^{g_0} \ne w$; it never becomes a label. (Secondarily, even the discarded guess is class-inconsistent: $g_0 \bmod 11 = 6 \ne 0 = r$.)

\emph{Resolution: the total rule rejects the bin outright, with the verifier as backstop.} The formal total recovery rule (\Cref{def:total-recovery}) rejects this bin outright at condition~1: $w$ is not an $N$th root of unity ($|w| \approx 8.4305$), so no label is emitted and the bin contributes nothing to $\mathcal{G}_\ell$, before any verifier call. This is the screen architecture working as designed on this instance. The rejection is a computed fact about this example, not a consequence of the unequal magnitudes alone: since $w = z_{63}\,z_0/z_{31}^2$, its modulus is the ratio $|w| = |z_{63}||z_0|/|z_{31}|^2$, which can equal one even when the three magnitudes differ; here the computed value is $|w| \approx 8.4305 \ne 1$, and the exact decoding rule rejects this bin deterministically.

Only an implementation that carried the $\mathrm{nMAD}$-surrogate-admitted candidate forward (bypassing the exact root-membership and class checks of the total rule and forcing the rounded guess $g_0 = 3471$ into the candidate set) would rely on the verifier to reject it: in that case the forced label produces a candidate reconstruction $\hat{X} \ne X$, the residual $x - \hat{x}$ is nonzero, the sound final verifier of \Cref{thm:verifier-soundness} returns $V(x,\hat{X}) = \mathrm{FAIL}$, and the input is routed to the dense FFT, exact in $O(N \log N) = O(46{,}189 \cdot 15.5) \approx 7.2 \times 10^5$ operations. (All displayed values in this subsection were recomputed at high precision under the PLUS convention.) This worked example demonstrates that a genuine collision can pass the $\mathrm{nMAD}$ surrogate yet fail the formal exact-equality screen ($D_{\max} \approx 0.664 \ne 0$), be admitted by the surrogate as a candidate-singleton, and be rejected by the total rule's exact root-membership check (condition~1 of \Cref{def:total-recovery}) or, if the total rule is bypassed entirely, by the verifier without sacrificing exact recovery; it is not a proof that every adversarial collision is rejected by the three-shift screen, and it does not claim a threshold-triggered fallback at this bin.

The cost theorem (\Cref{thm:sparse-path-complexity}) needs only the engineered transform family; it is containment of the support (\Cref{thm:global-label-completeness}), and hence the sparse path completing rather than falling back, that is conditional on all-stage genuine-singleton survival. The $O(k)$ candidate count is proved by \Cref{thm:global-label-count} rather than assumed. The adversarial support above falls outside the all-stage genuine-singleton-survival regime of \Cref{cond:all-stage-survival} (one of its tones collides at a stage); the surrogate admits the bin as a candidate-singleton, the total rule of \Cref{def:total-recovery} emits no label for it (root-membership failure, $|w| \ne 1$), and exact recovery of the full spectrum is in any case secured by the dense fallback rather than the three-shift screen. The hybrid theorem (\Cref{thm:hybrid-correctness}) guarantees exact recovery in this worst case at the cost of $O(N \log N)$ runtime, matching what an unconditional dense FFT would have done from the start.

\subsection{Exact Emission-Predicate-Passing Phantom and nMAD Surrogate False Accept}

The previous example has unequal shifted magnitudes ($D_{\max} \approx 0.664 \ne 0$): it fails the formal exact-equality screen yet still slips through the median $\mathrm{nMAD} < \tau_{cv}$ surrogate. The complementary and stronger direction is a non-singleton bin that passes the full exact screen ($D_{\max} = 0$ and exact singleton-consistent increment ratios), which no magnitude test, formal or surrogate, can exclude: this is the screen's irreducible one-sidedness (\Cref{rem:one-sided-normative}). Screen convention for this subsection: shifts $\{0,31,63\}$; formal screen $D_{\max} = 0$ with exact root membership and class consistency (\Cref{def:exact-screen}), $\mathrm{nMAD} < \tau_{cv}$ as the surrogate; label emission by \Cref{def:total-recovery}.

\emph{Formal exact false accept (two-tone witness).} The formal counterexample to screen sufficiency is the exact two-tone witness of \Cref{prop:two-tone-phantom}: $N = 765{,}049 = 23 \cdot 29 \cdot 31 \cdot 37$, stage prime $p = 23$, colliding tones $\{0, d\}$ with $d = N/31 = 24{,}679$ (both $\equiv 0 \bmod 23$), and coefficients $A_0 = -\alpha$, $A_d = 1 + \alpha$ with $\alpha = e^{+j 2\pi/31}$. By exact algebra the three shifted bin observations are $(1,\, 1,\, \alpha^2)$, byte-for-byte the triple of a genuine singleton at the phantom $h = 2d = 49{,}358$ (class-consistent, $h \equiv 0 \bmod 23$): $D_{\max} = 0$ exactly, the increment ratios are exactly those of that singleton, and the total rule of \Cref{def:total-recovery} emits the phantom label $h$. This is an exact algebraic identity, not a finite-precision event: the bin defeats the formal exact screen itself. Emitting the phantom label does not by itself invoke the verifier: if $h$ survives the global-label intersection and is used in the reconstructed candidate, the residual verifier rejects the candidate (\Cref{thm:verifier-soundness}) and the input is routed to the dense fallback.

\emph{Surrogate-level illustration (three tones, finite precision).} A separate, weaker phenomenon is a configuration that defeats only the implementation surrogate. Keep $N = 46{,}189$ and take a $p = 11$ bin holding the three tones $f \in \{0, 11, 22\}$ (all $\equiv 0 \bmod 11$) with real coefficients
\begin{equation*}
A_{0} = 1.00000, \qquad A_{11} = 0.99778, \qquad A_{22} = -0.20000 .
\end{equation*}
With $Y(\sigma) = \sum_f A_f\, e^{+j 2\pi f \sigma / N}$ (aliasing prefactor divided out), the three shifted magnitudes evaluate to
\begin{equation*}
\begin{aligned}
|Y(0)|  &= 1.7977800, \\
|Y(31)| &\approx 1.7977808, \\
|Y(63)| &\approx 1.7977800,
\end{aligned}
\end{equation*}
agreeing to about seven significant digits. Their median is $|Y(0)| = 1.79778$, and the deviations from the median are $\bigl|\,|Y(\sigma)| - 1.79778\,\bigr| = 0$, $8.0520 \times 10^{-7}$, and $1.46149 \times 10^{-10}$ at $\sigma = 0$, $31$, $63$; the deviations are stated directly because they are not recoverable from the displayed magnitudes: those are rounded in the $10^{-7}$ place, while the deviation that determines the $\mathrm{MAD}$ lies at the $10^{-10}$ scale. The median absolute deviation is the middle deviation, $\mathrm{MAD} = 1.46149 \times 10^{-10}$, so the surrogate statistic is $\mathrm{nMAD} = \mathrm{MAD}/\mathrm{median} \approx 8.1294 \times 10^{-11}$ and the surrogate admits the bin at every tolerance $\tau_{cv} > 8.13 \times 10^{-11}$, which covers the usual implementer choices; a smaller fixed tolerance, say $\tau_{cv} = 10^{-12}$, rejects it. The exact statistic, however, is nonzero: $D_{\max} \approx 4.48 \times 10^{-7} \ne 0$, so this configuration fails the formal exact-equality screen and is a false accept only at the surrogate level, not against the formal $D_{\max} = 0$ screen. (All values recomputed at high precision.)

\begin{remark}[The magnitude screen alone is not sound]
\label{rem:false-accept}
The two-tone witness above is the operative false accept against the formal screen: a bin holding $m = 2$ tones passes the full exact screen and emits a phantom label, so the full exact emission predicate does not certify singleton status (\Cref{prop:two-tone-phantom,prop:higher-order-phantom}; false accepts arise already at $m = 2$, cf.\ \Cref{ex:n15-mixed-tuple}, \Cref{rem:one-sided-normative}). The three-tone configuration illustrates the additional, weaker finite-precision gap of the $\mathrm{nMAD}$ surrogate: exact predicates and surrogate decisions must be reported separately, and this paper does so throughout. What prevents a wrong output in either case is the sound final verifier of \Cref{thm:verifier-soundness}, invoked once on the reconstructed candidate rather than per bin: if a phantom label survives the global-label intersection and is used in that candidate, the residual verifier rejects the candidate, after which the input is routed to the $O(N \log N)$ dense fallback. This is the concrete motivation for the verifier of \Cref{def:verifier} rather than relying on the three-shift screen alone.
\end{remark}

\section{Limitations and Open Directions}
\label{sec:limitations}

\subsection{Current Limitations}

The construction is exact only inside the model it assumes, and that model is a strong one. This section states where the guarantees stop: the measure-zero phantom slice that the one-sided screen cannot exclude, the exact-arithmetic requirement behind verifier soundness, the all-stage-survival specialization of the fast path, and the noiseless on-grid restriction. Correctness within the model is unconditional (\Cref{thm:verifier-soundness}, \Cref{thm:hybrid-correctness}); what follows concerns the boundary of the model and the runtime regime, not that correctness statement.

\subsubsection{Phantom Exclusion: Correctness is Verifier-Secured, not Screen-Secured}

The proved $O(k)$ candidate set is not phantom-free. The global-label engine bounds the count (\Cref{thm:global-label-count}), but a screen-passing phantom, a structured $m \ge 2$ collision (\Cref{prop:two-tone-phantom,prop:higher-order-phantom}), recovers a single phantom label that can survive the intersection (\Cref{rem:phantom-labels}). Correctness is therefore verifier-secured, not screen-secured: the $O(k)$ recovered labels are not all guaranteed true, and the sound verifier of \Cref{thm:verifier-soundness} (reading $k + |\mathcal{G}|$ consecutive residual samples, \Cref{def:verifier}) is what rejects the candidate carrying a surviving phantom and routes the input to the dense fallback. The open direction is a deterministic per-bin soundness rule that would exclude the measure-zero phantom slice at every $m \ge 2$ before verification; correctness does not depend on closing it, but a screen-only guarantee would.

\subsubsection{Finite-Precision Verification: the Exact-Arithmetic Requirement}

Verifier soundness is an exact-arithmetic, noiseless property (\Cref{rem:verifier-noiseless-floor}); exact arithmetic is the theorem model, and a practical finite-precision verification theory is future work (listed below), not a defect in \Cref{thm:verifier-soundness}. There is no universal residual lower bound: with arbitrary complex coefficients a nonzero residual can be scaled below any fixed tolerance, and coherent multi-tone cancellation can shrink the checked $k + |\mathcal{G}|$ samples further, so a fixed absolute float zero-test is unsound. The formal model is exact arithmetic; a fielded implementation must carry out the $k + |\mathcal{G}|$-sample residual in exact or rational arithmetic, or use coefficient-normalized, condition-aware thresholds, not a single absolute floor. The $\Theta(k/N^2)$ figure of \Cref{rem:verifier-noiseless-floor} is an illustrative near-miss scale only, not a safe threshold.

\subsubsection{All-Stage-Survival Specialization of the Fast Path}

The $O(k \log N)$ sparse path is specialized to support that survives as a genuine singleton at every stage. A true tone that collides in some stage $\ell$ is not guaranteed to be in the candidate set: if its bin fails the screen it is excluded outright, whereas a measure-zero screen-passing collision (\Cref{prop:two-tone-phantom}, \Cref{rem:one-sided-normative}) may instead emit a phantom or a different label under the total rule of \Cref{def:total-recovery}; should that label survive the intersection and be used in the reconstructed candidate, the sound verifier rejects the candidate (\Cref{thm:verifier-soundness}) and the input is routed to the $O(N \log N)$ dense fallback (\Cref{thm:hybrid-correctness}, \Cref{rem:alg1-semantics}). Recovery from partial singleton appearances rather than all-stage survival is left open. The regime consequence follows from the per-stage per-bin collision estimate derived in \Cref{rem:engineering-regime}, which does not vanish at larger admissible $k$ ($\approx 0.39$ at $c = 2$) because the stage primes scale with the sparsity: under any input model that makes bin loads comparable to uniform random support, the fraction of inputs taking the dense fallback rises with $k$, and the $O(k \log N)$ advantage is concentrated in the very-low-sparsity regime. Correctness is unaffected in either case (\Cref{thm:hybrid-correctness}); only the runtime regime shifts.

\subsubsection{Noiseless On-Grid Model}

The construction assumes noiseless observations and exact on-grid support, $f \in \{0, 1, \ldots, N-1\}$. Off-grid frequencies $f = f_0 + \delta$ with fractional offset $\delta \in (-0.5, 0.5)$ break the model in three ways: spectral leakage spreads energy across adjacent bins and violates singleton isolation; the residue $r_\ell = f \bmod p_\ell$ becomes ambiguous under peak-detection rounding; and the linear phase relation $\phi_v = 2\pi f \sigma_v / N$ no longer holds for fractional $f$. Likewise, additive noise inflates the normalized-MAD and degrades phase coherence, so the noiseless screen statistics do not transfer. A rigorous noisy and off-grid analysis is outside the present formal model and is not characterized here.

\needspace{16\baselineskip}
\subsection{Future Research Directions}
\nopagebreak[4]

The open directions below are listed in order of priority for elevating the conditional architecture toward a complete deterministic sparse FFT. They are research directions, not specified algorithms, and carry no performance promises.
\begin{enumerate}
\item Deterministic per-bin phantom exclusion. The open problem is soundness, not the count: the count is bounded deterministically for every input by the global-label engine (\Cref{thm:global-label-count}), while a measure-zero screen-passing phantom exists at every $m \ge 2$ (\Cref{prop:two-tone-phantom,prop:higher-order-phantom}) and can survive the intersection (\Cref{rem:phantom-labels}). A deterministic per-bin rule excluding that slice would make the $O(k)$ recovered labels all true before verification. This is a soundness-refinement direction only; correctness is already secured by the verifier (\Cref{thm:verifier-soundness}), which rejects any candidate built on a surviving phantom. The screen-only residue-tuple rule, whose tuple count multiplies across stages (\Cref{prop:assoc-count}), is not a route to it and is retained only as a negative result.
\item Partial-stage recovery. Recovery from partial singleton appearances rather than survival through the full stage chain, which would shrink the fraction of inputs routed to the dense fallback.
\item Finite-precision and noise model. Noise-tolerant screen and verifier statistics replacing the exact-arithmetic, noiseless assumptions; a fielded residual test in particular needs exact or rational arithmetic, or coefficient-normalized condition-aware thresholds, rather than an absolute floor.
\item Off-grid model. Sub-bin refinement for fractional frequencies, where spectral leakage breaks singleton isolation and the residue $f \bmod p_\ell$ becomes ambiguous under peak-detection rounding.
\item Arbitrary-length and non-divisor-compatible operation. Handling lengths outside the engineered exact-product envelope without dense fallback.
\item Empirical and hardware validation. Monte-Carlo benchmarking and hardware measurement, together with streaming and multidimensional variants.
\item Reduced-sample verifier variants. A divisor-compatible $2k$-separating CRT family (\Cref{rem:verifier-crt-variant}) that would read fewer samples than the consecutive-sample verifier of \Cref{def:verifier} when such a family can be engineered into $N$.
\end{enumerate}

\section{Conclusion}

This paper builds a deterministic global-label CRT sparse Fourier architecture for engineered exact-product transform lengths. The transform length is selected at design time as the exact product of small pairwise-coprime stage primes, so each stage transform is sized by the sparsity rather than by the ambient length. Co-prime-increment time shifts supply a one-sided singleton-consistency screen on the exact true-singleton shift signature (constant magnitude and exact unimodular increment ratios) of the divisor-compatible shifted-aliasing identity; each candidate-singleton is submitted to a total emission rule that emits at most one global frequency label, read in closed form from the three shifts, and the per-stage label sets are intersected into the candidate set. A consecutive-sample residual verifier gates acceptance, and a dense FFT is the fallback.

Four results carry the construction, all in the noiseless, on-grid, at-most-$k$-sparse, exact-arithmetic model. The candidate count is deterministically $O(k)$ for every input, proved rather than assumed (\Cref{thm:global-label-count}), and under all-stage genuine-singleton survival the support is contained in that candidate set (\Cref{cond:all-stage-survival}, \Cref{thm:global-label-completeness}). The consecutive-sample verifier of \Cref{def:verifier} is proved sound, so a passing sparse output is exact (\Cref{thm:verifier-soundness}), and the dense fallback is exact on every input of the model, which together make correctness unconditional within it (\Cref{thm:hybrid-correctness}). The costs are stated in named models: in the comparison real-RAM model of \Cref{def:label-decoding-model}, candidate construction uses $O(k \log N)$ samples and $O(k \log^2 N / \log k)$ arithmetic (\Cref{thm:cost-comparison-model}), and in the stronger exact label-decoding (root-index oracle) model this tightens to $O(k \log N)$ arithmetic (\Cref{thm:sparse-path-complexity}), with the verifier term reported separately rather than folded in ($O(k^2)$ in the oracle model, $O(k^2 + k \log N)$ in the comparison model), and $O(N \log N)$ worst case through the fallback; working memory is $O(k + \log N/\log k)$ in sequential operation (\Cref{thm:space-complexity}). The formal bounds hold asymptotically for every $k \ge 3$ under the engineered exact-product model of \Cref{def:transform-family}, while the practical sparse-path advantage is concentrated in the very-low-sparsity regime of \Cref{rem:engineering-regime}.

What remains open is visible in those same statements. The oracle-model $O(k \log N)$ rests on the root-index decoding oracle; the default proved figure without it is the comparison-model $O(k \log^2 N / \log k)$, and closing that gap is a decoding question rather than an architectural one. The candidate set is also not phantom-free: a measure-zero screen-passing phantom can survive the intersection (\Cref{prop:two-tone-phantom}, \Cref{rem:phantom-labels}), so the $O(k)$ recovered labels are not all guaranteed true and exactness stays verifier-gated rather than screen-gated. Beyond these, the operating model itself is narrow, and a finite-precision and noise model, an off-grid model, recovery from partial rather than all-stage singleton survival, arbitrary-length operation, and empirical validation are the principal directions left.

\section*{Acknowledgments}
The authors gratefully acknowledge the collaborative environment at SparseTech that made this research possible.

The theoretical and computational developments presented in this paper are part of an ongoing SparseTech research initiative on deterministic Sparse Fast Fourier Transform algorithms.

\textbf{Patent pending.}


\appendices
\crefalias{section}{appendix}

\section{Why Unlabeled Residue-Tuple Intersection Fails: The Exact Cartesian Tuple Count}
\label{app:tuple-negative}

This appendix analyzes the screen-only residue-tuple-intersection rule, superseded by the global-label engine (\Cref{thm:global-label-count}); it is retained as a negative result showing why per-bin global labels are required. It carries candidates as Cartesian residue tuples whose accepted-tuple count is the full rectangle $\prod_j |\mathcal{R}_j|$, which multiplies across stages, reaches $2^L$ at two-per-stage occupancy, and exceeds the algorithm's concrete budgets at attainable instances, exactly the obstruction the global-label engine removes. The companion per-bin B\'ezout recovery is also recorded here as part of the same superseded path; its measure-zero screen-passing phantoms are screen-soundness results and are stated, with their exact witnesses, in \Cref{app:screen-analysis} (\Cref{prop:two-tone-phantom,prop:higher-order-phantom}). The separate keyed multi-view and adaptive-modulus extensions are collected, as a non-operative extension, in \Cref{app:keyed-multiview}.

\begin{remark}[Why the $2^L$ blow-up of the tuple cascade cannot occur]
\label{rem:no-product-blowup}
The candidate count of the superseded accumulated-residue cascade (analyzed as a negative result in this appendix) would be
\[
  \bigl|\{(r_1,\dots,r_t) \in \mathcal{R}_1 \times \cdots \times \mathcal{R}_t : \mathrm{CRT} \in [0,N)\}\bigr| = \prod_{j\le t}|\mathcal{R}_j|,
\]
which multiplies per-stage bin sets and already reaches $2^t$ when each stage retains two screened bins (\Cref{prop:assoc-count}). \Cref{thm:global-label-count} avoids this entirely: it never forms a product. Each stage independently emits at most $p_\ell$ global labels, and the stages are combined by intersection, which can only shrink the set. The label count is linear per stage and the intersection is bounded by the smallest stage, so the accumulated count is $O(k)$ rather than a multiplicative per-stage product. The architectural difference is: the cascade carries candidates as residue tuples (Cartesian, multiplicative); the global-label engine carries them as recovered full frequencies (per-bin, additive then intersected). No Cartesian product $\mathcal{R}_1 \times \cdots \times \mathcal{R}_t$ is ever formed.
\end{remark}

\subsection{The Cartesian Residue-Tuple Cascade (Superseded)}

\noindent\emph{Setup (minimal tuple notation of the superseded cascade).} Let $p_1, \ldots, p_L$ be the pairwise-coprime stage primes of \Cref{def:transform-family} with accumulated modulus $Q_t \coloneqq \prod_{j \le t} p_j$, and let $\mathcal{R}_j \subseteq \{0, \ldots, p_j - 1\}$ be the pre-screen-passing bins retained at stage $j$ by the co-prime shift phase-consistency screen (\Cref{thm:phase-consistency-certificate}), realized via the exact conditions of \Cref{def:exact-screen}. The screen-only rule carries candidates as residue tuples in the Cartesian rectangle $U_t \coloneqq \mathcal{R}_1 \times \cdots \times \mathcal{R}_t$.

By CRT each tuple $\tau = (r_1, \ldots, r_t) \in U_t$ has a unique residue $\bar f_\tau \in [0, Q_t)$ with $\bar f_\tau \equiv r_j \pmod{p_j}$ for all $j \le t$, and its \emph{lift set} is $\mathrm{Lift}_t(\tau) \coloneqq \{ f \in [0, N) : f \equiv \bar f_\tau \pmod{Q_t} \}$. The rule accepts $\tau$ into the accepted-tuple set $T_t \subseteq U_t$ if and only if $\mathrm{Lift}_t(\tau) \neq \emptyset$. The induced accepted-frequency set is
\begin{equation}
C_t \coloneqq \{ f \in [0, N) : f \bmod p_j \in \mathcal{R}_j \text{ for all } j \le t \},
\label{eq:candidate-set}
\end{equation}
the union of the accepted lifts, refining $C_0 = [0,N) \supseteq C_1 \supseteq \cdots \supseteq C_L$. For a $k$-sparse support $\mathcal{S}$, write $\mathcal{S}^{\mathrm{cert}} \coloneqq \{ f \in \mathcal{S} : f \bmod p_j \in \mathcal{R}_j \text{ for all } j \le t \}$ for the screen-surviving support.

Two facts carry over from the fuller development this Setup condenses. Completeness: every $f \in \mathcal{S}^{\mathrm{cert}}$ has its residue tuple in $U_t$ by construction, hence accepted, so the rule loses no all-stage screen-surviving tone; at the final exact-product stage $Q_L = N$ each lift is a single frequency and every such tone is recovered. One-sidedness: acceptance certifies nothing, since $\mathcal{R}_j$ may contain screen-passing collision bins (\Cref{rem:one-sided-normative}); the defect quantified next is that the rule accepts far more than the support even when every screened bin holds a genuine singleton.

\begin{proposition}[Accepted-Tuple Count: Exact Form and Adversarial Lower Bound]
\label{prop:assoc-count}
With $T_t$, $U_t$, $\mathrm{Lift}_t$, $\mathcal{R}_j$, and $\mathcal{S}^{\mathrm{cert}}$ as in the Setup paragraph above:
\begin{enumerate}
  \item (Exact count.) $|T_t| = \big|\{\,(r_1,\dots,r_t)\in\mathcal{R}_1\times\cdots\times\mathcal{R}_t :\ \mathrm{Lift}_t(r_1,\dots,r_t)\neq\emptyset\,\}\big|$, the number of screened residue tuples with at least one frequency in $[0,N)$. Because $Q_t\le N$ along the cascade, $\bar f_\tau\in[0,Q_t)\subseteq[0,N)$ for every $\tau\in U_t$, so every residue tuple has a non-empty lift and the tuple count is the full rectangle at every stage, $|T_t|=\prod_{j\le t}|\mathcal{R}_j|=|U_t|$. What varies with $Q_t$ is the lift multiplicity, not the tuple count: because $Q_t=\prod_{j\le t}p_j$ is a prefix product of $N=\prod_{\ell\le L}p_\ell$ (\Cref{def:transform-family}), $Q_t\mid N$, so every accepted tuple lifts to exactly $N/Q_t$ frequencies in $[0,N)$, and the frequency count is $|C_t|=\sum_{\tau\in T_t}|\mathrm{Lift}_t(\tau)|=(N/Q_t)\prod_{j\le t}|\mathcal{R}_j|$, exceeding the tuple count by exactly the factor $N/Q_t$; at the final exact-product stage $Q_L=N$ this multiplicity is $N/Q_L=1$, so each tuple has exactly one lift, making $\tau\mapsto\bar f_\tau$ a bijection $T_L\to C_L$ and $|T_L|=|C_L|=\prod_{j\le L}|\mathcal{R}_j|$.
  \item (Spurious survivors.) At the final stage, $|T_L|-|\mathcal{S}^{\mathrm{cert}}|$ counts the spurious accepted frequencies in $C_L$: lifts $\bar f_\tau\in[0,N)$ that pass all-stage screens but are not true tones. Each is a non-true accepted frequency; the one-sided screen admits into $\mathcal{R}_j$ both genuine singleton bins of other true tones and screen-passing collision bins (\Cref{rem:one-sided-normative}).
  \item (Adversarial lower bound: the tuple count multiplies across stages and exceeds the algorithm's budgets.) Suppose the support contains at least two tones whose stage-$j$ residues are distinct, in such a way that the pre-screen-passing bin set has $|\mathcal{R}_j|\ge 2$ at every stage $j$ (it suffices that two fixed support tones $f\neq f'$ have $f\not\equiv f'\pmod{p_j}$ and both bins are pre-screen-passing at every stage, which holds whenever neither is involved in a stage-$j$ collision). Then by part~(1), at the final exact-product stage $Q_L=N$ the accepted-tuple count is the full rectangle product
  \[
    |T_L|=\prod_{j\le L}|\mathcal{R}_j|\ \ge\ 2^{L}.
  \]
  The count multiplies across stages: each added stage $j$ multiplies it by a further factor $|\mathcal{R}_j|\ge 2$, up to the family's maximal admissible stage count.

  For each fixed $k$ and $c$ the prime supply caps the stage count at $L \le \pi([k,ck])$. Across the admissible family, as $k \to \infty$ this cap grows without bound (\Cref{lem:prime-existence}), so $2^L$ genuinely diverges along admissible instances while the label engine's count stays $O(k)$ (\Cref{thm:global-label-count}). The separation is multiplicative versus additive at every instance: the accepted-tuple count multiplies across stages, $|T_L| = \prod_{j\le L}|\mathcal{R}_j| \ge 2^L$, while the global-label engine keeps $|\mathcal{G}_\ell| \le p_\ell$ at every stage of every instance (\Cref{thm:global-label-count}). At attainable instances the tuple count exceeds the algorithm's concrete budgets: for $k = 8$ with $c = 4$ the interval $[8,32]$ contains $\pi([8,32]) = 7$ primes, so the maximal admissible instance has $L = 7$ and $|T_L| \ge 2^{7} = 128 = 16k$; the chain-end candidate budget $k = 8$ is exceeded from $L = 4$ onward ($2^{4} = 16 > 8$), the per-stage budget $ck = 32$ from $L = 6$ onward ($2^{6} = 64 > 32$), and the maximal instance exceeds the chain-end budget by a factor of $16$, whereas the label engine's accumulated count is at most $p_1 = 11$ there. The screened residue rectangle $\mathcal{R}_1\times\cdots\times\mathcal{R}_L$ therefore contains, besides the diagonal tuples $(\,f\bmod p_1,\dots,f\bmod p_L\,)$ of true tones, mixed tuples that stitch residues of distinct true tones, each a valid lattice point of the rectangle but not a true tone, in numbers that multiply with each added stage. Consequently unlabeled tuple intersection cannot operate within the budgets ($k$ at chain end, $ck$ per stage) under which the label engine runs; this budget exceedance at attainable instances, compounded by the unbounded growth of $2^L$ across the family as $k$ grows, is the design-history ground for requiring labels. The smallest instance is the $N=15$, $\mathcal{S}=\{1,2\}$ example of \Cref{ex:n15-mixed-tuple}, where $|\mathcal{R}_1|=|\mathcal{R}_2|=2$ gives $|T_2|=2^2=4=2k$.
\end{enumerate}
\end{proposition}

\begin{proof}
(1) Membership in $T_t$ is the single non-empty-lift condition of the Setup paragraph (screen consistency $r_j\in\mathcal{R}_j$ being automatic for $\tau\in U_t$), which is exactly the displayed set. Since $Q_t\le N$ along the cascade, $\bar f_\tau\in[0,Q_t)\subseteq[0,N)$ for every $\tau\in U_t$, so every tuple has a non-empty lift and $|T_t|=|U_t|=\prod_j|\mathcal{R}_j|$ at every stage.

Because $Q_t=\prod_{j\le t}p_j$ is a prefix product of the pairwise-coprime stage primes whose full product is $N$ (\Cref{def:transform-family}), $Q_t\mid N$, and the residue class of $\bar f_\tau$ modulo $Q_t$ meets $[0,N)$ in exactly $N/Q_t$ points, so $|\mathrm{Lift}_t(\tau)|=N/Q_t$ for every $\tau\in U_t$; the lifts of distinct tuples are disjoint (distinct tuples have distinct residues $\bar f_\tau\in[0,Q_t)$), and $C_t$ is their union, so $|C_t|=\sum_{\tau\in T_t}|\mathrm{Lift}_t(\tau)|=(N/Q_t)\prod_j|\mathcal{R}_j|$. At the final stage $Q_L=N$ the multiplicity is $N/Q_L=1$: the CRT map is a bijection $U_L\to C_L\subseteq[0,N)$, each tuple has exactly one lift, and $|T_L|=|C_L|=\prod_j|\mathcal{R}_j|$.

(2) By the completeness fact of the Setup paragraph, every tone of $\mathcal{S}^{\mathrm{cert}}$ is accepted (its residue tuple lies in the rectangle by construction), so at the final stage the surplus is non-true accepted frequencies in $C_L$.

(3) Take two fixed support tones $f\neq f'$ in $[0,N)$ with $f\not\equiv f'\pmod{p_j}$ for every stage $j$, as supplied by the hypothesis. At each stage $j$ the residues $f\bmod p_j$ and $f'\bmod p_j$ are distinct and both lie in $\mathcal{R}_j$ by hypothesis (two distinct pre-screen-passing bins per stage; the count requires nothing about what occupies them); hence $|\mathcal{R}_j|\ge 2$ at every stage. By part~(1), at the final stage $Q_L=N$ every tuple of the rectangle has a unique preimage, so $|T_L|=\prod_{j\le L}|\mathcal{R}_j|\ge 2^{L}$; the count multiplies by a further factor $|\mathcal{R}_j|\ge 2$ with each added stage. For each fixed $k$ and $c$ the stage count is capped at $L \le \pi([k,ck])$; across the admissible family that cap grows without bound with $k$ (\Cref{lem:prime-existence}); the multiplicative count already exceeds the algorithm's candidate budgets at small attainable instances (at $k = 8$, $c = 4$: $\pi([8,32]) = 7$; $2^{L} > k = 8$ from $L = 4$ onward, $2^{L} > ck = 32$ from $L = 6$ onward, and $2^{7} = 128$ exceeds both budgets at the maximal $k=8$ instance), and diverges along the family as $k$ grows. The surplus tuples are the mixed CRT combinations of residues drawn from distinct true tones, each a valid lattice point of the screened rectangle but not a true tone; thus at attainable instances the tuple rule exceeds the concrete budgets ($k$ at chain end, $ck$ per stage) within which the label engine operates.
\end{proof}

\begin{example}[Residue-set intersection admits mixed CRT tuples: $N=15$]
\label{ex:n15-mixed-tuple}
The unrestricted intersection of \Cref{eq:candidate-set} is not by itself a support-recovery mechanism: genuine singleton residues at every stage can still combine into mixed cross-frequency tuples absent from the support. Take illustrative small moduli (smaller than the engineered $[k, ck]$ regime, used purely to expose the mechanism) $N = 15$, $p_1 = 3$, $p_2 = 5$, with support $\mathcal{S} = \{1, 2\}$ and $k = 2$. The screened residue sets are $\mathcal{R}_1 = \{1, 2\}$ (from $1, 2 \bmod 3$) and $\mathcal{R}_2 = \{1, 2\}$ (from $1, 2 \bmod 5$); each residue is a genuine singleton. Then
\begin{equation*}
\begin{aligned}
C_2 &= \{ f \in [0,15) : f \bmod 3 \in \{1,2\} \;\wedge\; f \bmod 5 \in \{1,2\} \} \\
    &= \{1, 2, 7, 11\},
\end{aligned}
\end{equation*}
which contains the two spurious mixed tuples $(r_1, r_2) = (1, 2) \to f = 7$ and $(2, 1) \to f = 11$ in addition to the true support. Thus $|C_2| = 4 > k = 2$: all-stage singleton residues do not imply $k$ candidates.

In the tuple language of the Setup paragraph this is $T_2 = U_2$ with $|T_2| = |\mathcal{R}_1|\,|\mathcal{R}_2| = 4 = 2k$, the $L=2$ instance of the $2^L$ adversarial lower bound of \Cref{prop:assoc-count}(3). Under the screen-only intersection rule the tuple count multiplies across stages and already reaches $|T_2| = 2k$ here, exceeding the chain-end budget $k$; per-bin B\'ezout recovery instead reads one global label per bin, recovering a genuine singleton exactly (\Cref{lem:global-label-recovery}(3)), while a deterministic worst-case soundness rule is obstructed by the measure-zero phantom (\Cref{prop:two-tone-phantom}, \Cref{rem:bezout-no-discharge}). This is exactly the obstruction removed by the global-label engine (\Cref{thm:global-label-count}), which intersects per-bin global labels rather than multiplying per-stage residue sets.
\end{example}

\subsection{Per-Bin B\'ezout Recovery and the Measure-Zero Phantom (Superseded Path)}

The superseded accumulated-residue cascade tracked candidates by their accumulated CRT residues across stages rather than by per-stage singleton isolation alone. CRT tuple uniqueness on the full residue tuple $(r_1, \ldots, r_L)$ does not imply that any individual stage $\ell$ exhibits a singleton bin for every frequency: a frequency $f$ may collide with distinct frequencies in different stages while remaining unique modulo $Q_L = \prod_{j=1}^{L} p_j$. The per-bin B\'ezout recovery discussed in this subsection is the superseded per-bin reading rule; its measure-zero screen-passing phantom is the soundness gap that the final verifier of \Cref{thm:verifier-soundness}, not the cascade, closes.

\emph{The per-bin reading rule.} Consider a stage-$t$ bin screened as a candidate-singleton and read its de-rotated phasors at the three shifts $\{0,31,63\}$. The shift increments are $\Delta_1=31$ and $\Delta_2=32$ with $\gcd(31,32)=1$, so by \Cref{thm:coprime-bezout} the measured phasors of a genuine singleton recover its exact global frequency $g\in[0,N)$ by the B\'ezout combination $63g-31g$ and then $32g-31g\pmod N$; that exact-recovery statement is proved in the body as \Cref{lem:global-label-recovery}(3) and is not restated as an appendix result. What this subsection records is the complementary negative fact: the same reading rule is not worst-case sound, because at every $m\ge 2$ a measure-zero coefficient slice of $m$-tone collisions passes the full exact screen and makes the rule return a phantom; the two witness propositions are stated, as screen-soundness results, in \Cref{app:screen-analysis} (\Cref{prop:two-tone-phantom,prop:higher-order-phantom}). Generic rejection of collisions (\Cref{thm:generic-collision-rejection}) does not repair that, and is not needed for correctness (\Cref{rem:generic-rejection-not-needed}).

\begin{remark}[Generic rejection is not needed for correctness]
\label{rem:generic-rejection-not-needed}
For generic (Lebesgue-almost-every) complex coefficients on a fixed multi-tone support, an $m \ge 2$ bin fails the full exact screen; that statement is \Cref{thm:generic-collision-rejection}. It is recorded there and deliberately not used here, because neither of the two roles it might be asked to fill actually requires it.

\emph{Not needed for the count.} \Cref{thm:global-label-count} bounds the per-stage and accumulated label sets by $|\mathcal{G}_\ell| \le p_\ell \le ck$ and $|\mathcal{G}| \le p_1 \le ck$ deterministically for every input spectrum, with no association assumption and without assuming that any bin is a genuine singleton. That bound is unconditional, whereas a generic label count would hold only off a measure-zero coefficient slice, and off exactly the slice that \Cref{prop:two-tone-phantom,prop:higher-order-phantom} exhibit. No generic label-count statement is therefore made anywhere in this paper.

\emph{Not needed for correctness.} Correctness rests on one-sidedness and final verification rather than on rejecting collisions. A screen-passing phantom is admitted; it is one label among the $\le p_\ell$ already counted, and it is caught only if it survives the global-label intersection and is used in the reconstructed candidate, where the sound residual verifier rejects that candidate and the input is routed to the dense fallback (\Cref{rem:one-sided-normative}, \Cref{thm:verifier-soundness}, \Cref{thm:hybrid-correctness}). Generic rejection would raise the fraction of inputs on which the sparse path completes; it would not change what is guaranteed.

\emph{Two distinct negative facts.} The Cartesian tuple-rectangle count $|T_t| = \prod_{j\le t}|\mathcal{R}_j|$ multiplies deterministically across stages and at attainable instances exceeds the algorithm's concrete budgets, $k$ at chain end and $ck$ per stage (\Cref{prop:assoc-count}); separately, per-bin label recovery is not worst-case sound at any $m \ge 2$ (\Cref{prop:two-tone-phantom,prop:higher-order-phantom}). Neither affects the operative candidate count, which the global-label engine establishes as a proved $O(k)$ bound by intersecting the at-most-one global label emitted per candidate-singleton under the total rule of \Cref{def:total-recovery} (\Cref{thm:global-label-count}, \Cref{rem:no-product-blowup}), independently of any worst-case rectangle bound. The residual object is purely one of soundness: an injected phantom label, discharged by \Cref{thm:verifier-soundness}.
\end{remark}

\begin{remark}[B\'ezout per-bin recovery bounds the count but not the soundness]
\label{rem:bezout-no-discharge}
Promoting the per-bin B\'ezout reading rule of this subsection to one global label per candidate-singleton and intersecting the per-stage label sets does resolve the candidate-count role: the count is the proved $O(k)$ bound of \Cref{thm:global-label-count}. What per-bin recovery does not settle is soundness. The screen-passing phantoms of \Cref{prop:two-tone-phantom,prop:higher-order-phantom} are exactly the one-sidedness disclaimed in the \emph{Scope (componentwise only)} paragraph of \Cref{rem:bezout-scope} (following \Cref{thm:coprime-bezout}): that argument proves only that two distinct on-grid tones cannot individually satisfy both increment congruences; it does not rule out that a phasor sum of $m\ge 2$ tones reproduces, on a measure-zero coefficient slice, the three observations of a genuine singleton (the two-tone witness $N=765{,}049$, tones $\{0,d\}$, is the minimal such case). A phantom injected by such a bin is one label among the $\le p_\ell$ already counted, so it does not inflate the bound; and if it survives the intersection and enters the reconstructed candidate, that candidate is caught downstream, the final verifier of \Cref{thm:verifier-soundness} rejecting any reconstruction that is not the exact spectrum and routing the input to the dense fallback (\Cref{rem:false-accept} exhibits the companion magnitude-screen false accept that the verifier also catches). Correctness is therefore independently secured by \Cref{thm:verifier-soundness}; the only object the recovery leaves open is the deterministic per-bin soundness rule, not the count.
\end{remark}

\section{Deferred Screen Analysis: Supporting Lemmas, Genericity, and Escalation}
\label{app:screen-analysis}

This appendix collects the supporting material behind the five operative statements of \S\ref{sec:screens}: the B\'ezout co-primality theorem with its scope and intuition remarks (\Cref{thm:coprime-bezout}, \Cref{rem:bezout-scope}, \Cref{rem:coprime-intuition,rem:invariance-intuition,rem:multistage-intuition}), the generic (measure-zero-exceptional) collision-rejection theorem (\Cref{thm:generic-collision-rejection}), the shift-triple design notes with the arithmetic-shift Hankel/Prony escalation lemma (\Cref{lem:arith-shift-hankel-rank}) and its contrast with a full Prony solve (\Cref{rem:prony-relation}), the magnitude and phase surrogate lemmas (\Cref{lem:singleton-cv}, \Cref{lem:phase-linearity}), and the cross-stage collision-multiplicity estimates (\Cref{lem:collision-rare}). Nothing here is required for the correctness chain of \S\ref{sec:main-theorem}, which rests on one-sidedness (\Cref{rem:one-sided-normative}) and verifier-gated acceptance (\Cref{thm:verifier-soundness}, \Cref{thm:hybrid-correctness}).

\subsection{B\'ezout Co-Primality and Intuition}

\begin{theorem}[Co-prime Time Shifts Guarantee via B\'ezout's Identity]
\label{thm:coprime-bezout}
Let three time shifts be $\{0, \Delta_1, \Delta_1 + \Delta_2\}$ where $\gcd(\Delta_1, \Delta_2) = 1$, and let $f_1, f_2 \in \{0, 1, \ldots, N-1\}$ be on-grid frequencies. The two increment congruences
\begin{align}
(f_1 - f_2) \cdot \Delta_1 &\equiv 0 \pmod{N} \\
(f_1 - f_2) \cdot \Delta_2 &\equiv 0 \pmod{N}
\end{align}
hold simultaneously if and only if $f_1 = f_2$ (i.e., the frequencies cannot be distinct). The statement is made in the on-grid convention alone: frequencies are integers in $[0,N)$ and both congruences are taken modulo $N$, which is the convention in which the algorithm reads its labels. No hypothesis relating $\Delta_1$ or $\Delta_2$ to $N$ is required.
\end{theorem}

\begin{proof}
If $f_1 = f_2$ then $f_1 - f_2 = 0$ and both congruences hold. Conversely, suppose both hold. Since $\gcd(\Delta_1, \Delta_2) = 1$, B\'ezout's identity supplies integers $a, b$ with $a\Delta_1 + b\Delta_2 = 1$. Multiplying the first congruence by $a$, the second by $b$, and adding them gives
\begin{align}
f_1 - f_2 \;=\; (f_1 - f_2)(a\Delta_1 + b\Delta_2) \;\equiv\; 0 \pmod{N} .
\end{align}
Both frequencies lie in $[0, N)$, so $|f_1 - f_2| < N$, and the only multiple of $N$ of magnitude less than $N$ is $0$. Hence $f_1 = f_2$.
\end{proof}

\begin{remark}[Scope of the B\'ezout read]
\label{rem:bezout-scope}
\emph{Angular restatement.} The same B\'ezout step in the normalized angular convention (frequencies $\nu \in [0, 2\pi)$, congruences modulo $2\pi$) yields $\nu_1 = \nu_2$. For on-grid frequencies the two readings are the same statement, since $\nu = 2\pi f/N$ turns $(\nu_1 - \nu_2)\Delta \equiv 0 \pmod{2\pi}$ into $(f_1 - f_2)\Delta \equiv 0 \pmod{N}$; the theorem is stated on the grid because that is where the algorithm applies it.

\emph{Scope (componentwise only).} This argument proves uniqueness componentwise for a single pair $(f_1, f_2)$: two distinct on-grid frequencies cannot individually satisfy both increment congruences. It does not by itself exclude every multi-tone phasor-sum match, i.e. it does not rule out that a superposition of $m \ge 2$ tones reproduces, on a measure-zero coefficient slice, the three singleton-consistent observations of a genuine singleton (\Cref{prop:two-tone-phantom} exhibits an explicit two-tone phantom, so a two-tone false accept is not excluded by this lemma). That one-sidedness is discharged by the sound verifier (\Cref{thm:verifier-soundness}), not by this lemma.

\emph{Practical Implementation:} With shifts $\{0, 31, 63\}$ where $31 + 32 = 63$ and $\gcd(31, 32) = 1$, generic multi-tone collisions fail the exact screen, a fact established from the screen's own ratio/root conditions in \Cref{thm:generic-collision-rejection}, and not inherited from the pairwise argument of this lemma, whose scope is stated above. Degenerate or unresolved collisions whose bins fail the $S=3$ screen emit no label (optionally re-examined by arithmetic-shift Prony escalation); the dense FFT fallback is triggered at stage or chain level per \Cref{thm:hybrid-correctness}.
\end{remark}

\begin{remark}
\label{rem:coprime-intuition}
B\'ezout's identity is the engine behind the co-prime shift screen. Geometrically, two frequencies $f_1 \neq f_2$ can produce phase coincidence at any single time shift (there is always some offset at which their rotating phasors happen to align), and they can even coincide at two time shifts if the shifts are commensurate. The coprimality condition $\gcd(\Delta_1, \Delta_2) = 1$ blocks the third coincidence: B\'ezout constructs an integer combination $a \Delta_1 + b \Delta_2 = 1$ that forces $f_1 = f_2$ whenever all three alignments hold. In short, coprime shifts make it arithmetically impossible for a single distinct pair to imitate a singleton at all three shifts; this is a componentwise, pairwise statement and does not by itself exclude a multi-tone phasor sum (\Cref{rem:one-sided-normative}).
\end{remark}

\begin{remark}
\label{rem:invariance-intuition}
Operationally, a singleton is a clean tone whose decimated magnitude is set once by its underlying amplitude and then stays fixed no matter how the time-window is shifted: shifting the window only rotates phase, it does not change amplitude. Generic collisions break this invariance because the sum of two rotating phasors oscillates in magnitude as the shift varies, producing a detectable signature (the measure-zero exception is \Cref{prop:two-tone-phantom}). The normalized-MAD implementation surrogate (\Cref{def:cv-singleton-test}) is a quantitative exploit of this qualitative distinction.
\end{remark}

\begin{remark}
\label{rem:multistage-intuition}
The correctness argument has an intuitive shape: every frequency $f$ has a unique residue tuple under pairwise-coprime moduli (that is just the Chinese Remainder Theorem); and as long as the product of moduli reaches $N$, two distinct frequencies cannot share all of their residues. What the multi-stage construction adds is that each stage refines the candidate set by intersecting with the global labels emitted by screened bins at that stage (cost $O(k \log k)$ per stage on $\Theta(k)$-sized DFTs), so that by the end of the $L$-stage chain, whose primes multiply to exactly $N$ by design (\Cref{def:transform-family}), the candidate set has been narrowed to an $O(k)$-size set, a count proved by the global-label engine (\Cref{thm:global-label-count}). No accumulated modulus is tracked at run time; the coverage is fixed when $N$ is chosen.

Under all-stage genuine-singleton survival of the support this set contains the true support (\Cref{thm:global-label-completeness}); a true tone that collides at some stage is not guaranteed in it and is recovered by the dense fallback. The formal guarantee is deliberately containment rather than equality: for inputs outside the survival condition, surplus global labels and measure-zero screen-passing phantoms (\Cref{prop:two-tone-phantom}) can survive in $\mathcal{G}$ (under survival itself no phantom can arise, an $m \ge 2$ collision of support tones being excluded by definition), and because the algorithm cannot certify survival from the data, exactness of the returned spectrum is secured by the verifier-gated acceptance of \Cref{thm:verifier-soundness} (PASS) or the dense fallback (FAIL), not by candidate construction alone.

The final guarantee is therefore both correctness (verifier-gated, with the dense fallback catching any corner case, \Cref{thm:hybrid-correctness}) and efficiency ($L = O(\log_k N)$ stages suffice, and the work per stage is $k$-sized, giving $O(k \log N)$).
\end{remark}

\subsection{Generic Collision Rejection}

\begin{theorem}[Generic Collision Rejection (Measure-Zero Exceptional Slice)]
\label{thm:generic-collision-rejection}
Fix the operative shift triple $\{0, \Delta_1, \Delta_1 + \Delta_2\} = \{0, 31, 63\}$, so $\Delta_1 = 31$, $\Delta_2 = 32$, $\gcd(\Delta_1, \Delta_2) = 1$ and $\Delta_2 - \Delta_1 = 1$; the theorem is stated for this triple alone, the only one the algorithm uses and the only one for which the exact screen of \Cref{def:exact-screen} is defined (general triples need a further hypothesis and are not claimed here; scope note at the end of Step~3). Fix any support placing $m \ge 2$ distinct on-grid frequencies in bin $r$; since frequencies lie on the finite grid, genericity is over the coefficients for each fixed support, not over frequency placements. In the noiseless on-grid model, the coefficient vectors $A \in \mathbb{C}^m$ for which the bin passes the full exact screen of \Cref{def:exact-screen} form a Lebesgue-measure-zero exceptional set in $\mathbb{C}^m$.

The componentwise pairwise argument of \Cref{thm:coprime-bezout} alone does not control phasor sums (\Cref{rem:bezout-scope}); this is a one-sided rejection only, since phantoms persist on the measure-zero coefficient slice at every $m \ge 2$ (\Cref{prop:two-tone-phantom,prop:higher-order-phantom}) and exactness is secured by the sound final verifier of \Cref{thm:verifier-soundness}.
\end{theorem}

\begin{proof}
\emph{Step 1: the exact passing set, computed from \Cref{def:exact-screen} itself.} Suppose bin $r$ passes and emits the label $g$. Condition~1 gives $|z_0| = |z_{31}| = |z_{63}| > 0$, so $z_0 \ne 0 \ne z_{31}$ and the ratios are defined; condition~2 then reads
\[
  w \;=\; u_{32}\,u_{31}^{-1} \;=\; \frac{z_{63}}{z_{31}}\cdot\frac{z_0}{z_{31}} \;=\; \frac{z_{63}\,z_0}{z_{31}^{2}} \;=\; \omega^{g},
\]
and condition~3 restricts $g$ to $g \equiv r \pmod{p_\ell}$. Conversely, putting $c = z_0 \ne 0$ and $q = z_{31}/z_0$ (of unit modulus by condition~1), the triples meeting conditions~1 and~2 with label $g$ are exactly
\begin{equation}
  \mathcal{P}_g \;=\; \bigl\{\, c\,(1,\ q,\ q^{2}\omega^{g}) \;:\; c \in \mathbb{C}\setminus\{0\},\ |q| = 1 \,\bigr\},
  \label{eq:passing-set}
\end{equation}
a set of real dimension three in $\mathbb{C}^3$. This is strictly larger than the set of singleton patterns: $\mathcal{P}_g$ meets the singleton pattern of the tone $g$ only at the single value $q = \omega^{31 g}$, while every other $q$ on the unit circle passes as well. The proof below therefore uses \eqref{eq:passing-set} and never a singleton-pattern description.

\emph{Step 2: the passing set is cut out by finitely many quadratics.} By Step~1, passing with label $g$ forces $z_{63}\,z_0 - \omega^{g} z_{31}^{2} = 0$. In the noiseless on-grid model the three observations are linear in the coefficient vector $A = (A_1, \ldots, A_m) \in \mathbb{C}^m$ carried by the fixed support $f_1, \ldots, f_m$ of bin $r$: by the decimation aliasing identity (\Cref{subsec:sampling-operator}), $z_\sigma = \kappa \sum_{i=1}^{m} A_i \alpha_i^{\sigma}$ with $\kappa = p_\ell/N \ne 0$ and $\alpha_i = \omega^{f_i}$. Define the homogeneous quadratic
\[
  Q_g(A) \;\coloneqq\; z_{63}(A)\,z_0(A) \;-\; \omega^{g}\,z_{31}(A)^{2},
\]
a holomorphic polynomial on $\mathbb{C}^m$. The pass set of bin $r$ is contained in $\bigcup_{g=0}^{N-1}\{A \in \mathbb{C}^m : Q_g(A) = 0\}$, a union of finitely many such zero sets.

\emph{Step 3: no $Q_g$ vanishes identically.} Reading off the coefficient of $A_i^2$ in $Q_g$,
\[
  [A_i^2]\,Q_g \;=\; \kappa^{2}\bigl(\alpha_i^{63} - \omega^{g}\alpha_i^{62}\bigr) \;=\; \kappa^{2}\,\alpha_i^{62}\,\bigl(\alpha_i - \omega^{g}\bigr),
\]
which vanishes if and only if $\alpha_i = \omega^{g}$, i.e.\ $f_i = g$, since $\alpha_i \ne 0$ and $t \mapsto \omega^{t}$ is a bijection of $\mathbb{Z}_N$ onto the $N$th roots of unity. Because $m \ge 2$ and the $f_i$ are pairwise distinct, at most one of them equals $g$, so at least one diagonal coefficient is nonzero and $Q_g \not\equiv 0$ for every $g \in \{0, \ldots, N-1\}$, with no rank hypothesis on the observation map.

\emph{Scope note on general triples.} The exponents $63$ and $62 = 2\Delta_1$ are those of the operative triple, where $\Delta_2 - \Delta_1 = 1$. A general co-prime-increment triple gives $[A_i^2]\,Q_g = \kappa^2 \alpha_i^{2\Delta_1}\bigl(\alpha_i^{\Delta_2 - \Delta_1} - \omega^{g}\bigr)$ and would need the further hypothesis $\gcd(\Delta_2 - \Delta_1, N) = 1$, absent which $f_i \mapsto f_i(\Delta_2 - \Delta_1) \bmod N$ is many-to-one and several diagonal coefficients can vanish together. We expect the conclusion to extend to any triple meeting that hypothesis; it is not proved here.

\emph{Step 4: conclusion.} A holomorphic polynomial not identically zero on $\mathbb{C}^m$ has zero set of Lebesgue measure zero in $\mathbb{C}^m \cong \mathbb{R}^{2m}$. Each $\{Q_g = 0\}$ is therefore Lebesgue-null, and so is the finite union containing the pass set. Hence, for each fixed support of $m \ge 2$ distinct on-grid frequencies in bin $r$, Lebesgue-almost-every coefficient vector fails the exact screen. The exceptional set is nonempty (the witnesses named in the statement inhabit it), so the rejection is one-sided and exactness remains verifier-gated (\Cref{thm:verifier-soundness}).
\end{proof}

\begin{remark}[Screen Failure, Escalation, and the Two Collision Outcomes]
\label{rem:collision-escalation}
A collision that fails the exact screen of \Cref{def:exact-screen} is excluded from the stage's screened set and its bin emits no label, a normal per-bin event; the input reaches the $O(N\log N)$ dense-FFT fallback only through the stage- and chain-level routing rule of \Cref{thm:hybrid-correctness} (a no-consistent-label stage, a budget overflow, a verifier FAIL, or an envelope violation); for higher-order $m \ge 3$ collisions the algorithm may instead escalate to an arithmetic-shift Hankel/Prony rank test of order $S \ge 2m - 1$, on an increment satisfying the anti-aliasing condition $\gcd(\Delta, N) = 1$ (a separate sample acquisition; see \Cref{lem:arith-shift-hankel-rank}). That escalation is optional and lies outside the primary algorithm and its cost theorems: \Cref{alg:multistage_oklogn} never invokes it, it needs a second acquisition at arithmetic shifts, and it is charged in none of \Cref{thm:sampling-bluestein,thm:sparse-path-complexity,thm:cost-comparison-model,thm:hybrid-correctness}. An implementation that enables it must re-derive those bounds; no guarantee of this paper depends on it. The screen is one-sided, however: a measure-zero coefficient slice at any $m \ge 2$ passes the full screen (the two-tone witness of \Cref{prop:two-tone-phantom}), and such a screen-passing phantom is not caught here at all. It is admitted into the candidate set and caught only by the sound final verifier of \Cref{thm:verifier-soundness}, which rejects the non-exact reconstruction and routes to the dense fallback.
\end{remark}

The screen is therefore a fast one-sided screen on the co-prime-increment shift design: it is necessary (every true singleton passes, \Cref{thm:phase-consistency-certificate}) and rejects generic collisions (\Cref{thm:generic-collision-rejection}), but it is not sufficient at any $m \ge 2$. Measure-zero screen-passing phantoms exist already at $m = 2$ (the explicit two-tone witness of \Cref{prop:two-tone-phantom}), so a passing bin is only a candidate-singleton and correctness is secured not by this screen but by the sound final verifier of \Cref{thm:verifier-soundness}. It is not described as a Hankel/Prony rank-one test in the increment-coprime regime, because the triple $\{0,\Delta_1,\Delta_1+\Delta_2\}$ of \Cref{def:hankel-matrix} is not an arithmetic progression.

\subsection{Screen-Passing Phantoms: Measure-Zero Witnesses at Every Multiplicity}

These two propositions record the exact screen-soundness limit of the per-bin reading rule: at every collision multiplicity $m \ge 2$ a measure-zero adversarial coefficient slice passes the full exact screen of \Cref{def:exact-screen} and injects a phantom label. They are the witnesses behind the one-sided normative reading of \Cref{rem:one-sided-normative}, and the superseded per-bin recovery record of \Cref{app:tuple-negative} cites them as the reason no deterministic worst-case per-bin soundness rule follows from the screen.

\begin{proposition}[Exact Two-Tone Screen-Passing Phantom]
\label{prop:two-tone-phantom}
The per-bin reading rule does not exclude two-tone false accepts. There is a structured, measure-zero adversarial complex-coefficient configuration on $m=2$ colliding true tones whose bin passes the full exact screen of \Cref{def:exact-screen} ($D_{\max} = 0$ and exactly singleton-consistent increment ratios) yet is a non-singleton, so B\'ezout recovery reads a phantom frequency and injects it into the per-bin recovered-label set. Two is the minimal such multiplicity, a genuine singleton being read exactly (\Cref{lem:global-label-recovery}(3)).

\emph{Two-tone witness ($m=2$; exact algebraic identity).} Take $N=765{,}049=23\cdot 29\cdot 31\cdot 37$, stage prime $p=23$, residue class $r=0$, and $d=N/31=24{,}679$, with the two colliding true tones $\{0,d\}$ (both $\equiv 0\bmod 23$). Let $\alpha=e^{+j2\pi/31}$ and choose coefficients $A_0=-\alpha$, $A_d=1+\alpha$. The bin value at shift $\sigma$ is $A_0+A_d\,e^{+j2\pi\sigma d/N}=A_0+A_d\,e^{+j2\pi\sigma/31}$ (since $d/N=1/31$), giving the triple $(1,1,\alpha^2)$ at $\{0,31,63\}$. This triple is byte-for-byte identical to that of a genuine singleton at the phantom $h=2d=49{,}358$ with coefficient $1$ (since $e^{+j2\pi\sigma h/N}$ equals $1,1,\alpha^2$ at the same shifts): the exact screen statistic is $D_{\max} = 0$, the increment ratios are exactly those of that singleton, and the per-bin B\'ezout recovery returns $g=49{,}358$, which is neither true tone ($h\notin\{0,24{,}679\}$, while $h\equiv 0\bmod 23$). Thus a measure-zero two-tone slice already false-accepts.
\end{proposition}

\begin{proof}
The map from coefficients to the observed triple $(c_0,c_{31},c_{63})\in\mathbb{C}^3$ is, for $m$ distinct tones, the $3\times m$ generalized-Vandermonde $E$ with entries $E_{s,i}=x_i^{\sigma_s}$, $x_i=e^{+j2\pi g_i/N}$, on the row exponents $\sigma_s\in\{0,31,63\}$; on distinct on-grid tones $\operatorname{rank}(E)\ge 2$, since rank one would make two columns proportional, forcing $(x_i/x_j)^{31}=(x_i/x_j)^{63}=1$ and hence $x_i=x_j$ by $\gcd(31,32)=1$ (\Cref{thm:coprime-bezout}). For a target phantom $h$ the singleton-pattern vector $b_h=(1,\eta^{31},\eta^{63})$ with $\eta=e^{+j2\pi h/N}$ defines the \emph{impersonation set} $\{A:EA\in\mathrm{span}(b_h)\}$, the coefficient vectors whose bin reproduces the pattern of a singleton at $h$ (a proper subset of the full screen-passing set characterized in \Cref{thm:generic-collision-rejection}), and at $m=2$ the displayed witness solves $EA=b_h$ exactly: with $g_1=0$, $g_2=d=N/31$ and $h=2d$ one has $e^{+j2\pi\sigma d/N}=e^{+j2\pi\sigma/31}$, so the two colliding tones contribute $(1,1,1)$ and $(1,1,\alpha)$ at $\sigma\in\{0,31,63\}$ while $b_h=(1,1,\alpha^2)$; the choice $A_0=-\alpha$, $A_d=1+\alpha$ gives $-\alpha+(1+\alpha)=1$ at $\sigma=0$ and $\sigma=31$ and $-\alpha+(1+\alpha)\alpha=\alpha^2$ at $\sigma=63$, which is $b_h$ entrywise. The three magnitudes are then $|1|=|1|=|\alpha^2|=1$, equal and strictly positive, so the bin is occupied with $D_{\max}=0$ (and $\mathrm{nMAD}=0$); the increment ratios are $u_{31}=1$ and $u_{32}=\alpha^2$, so $w=u_{32}u_{31}^{-1}=\alpha^2=e^{+j2\pi\,2d/N}$ is exactly an $N$th root of unity with index $g=2d=49{,}358$, and $49{,}358\equiv 0\bmod 23$ is class-consistent. The bin therefore passes every condition of \Cref{def:exact-screen} in exact arithmetic and emits $g=h\notin\mathcal{S}$. These configurations form a positive-codimension (measure-zero) slice of coefficient space: with $\operatorname{rank}(E) \ge 2$ the impersonation set $\{A : EA \in \mathrm{span}(b_h)\}$ has dimension at most $m-1$ over $\mathbb{C}$ and is therefore a proper linear subspace of $\mathbb{C}^m$, so Lebesgue-almost-every coefficient vector on the same tones avoids it. The noiseless on-grid model admits arbitrary complex coefficients, so the witness is a valid adversarial input; hence no deterministic worst-case per-bin soundness rule follows from the screen at $m=2$.
\end{proof}

\begin{proposition}[Higher-Order Screen-Passing Phantoms at Every \texorpdfstring{$m\ge 3$}{m>=3}]
\label{prop:higher-order-phantom}
The same failure occurs at every higher collision multiplicity, and it does so inside the operating model rather than only at one fixed instance. For every $m\ge 3$ there is an admissible instance of \Cref{def:transform-family} with sparsity $k \ge m$, a stage prime $p$, and a residue class carrying $m$ distinct colliding true tones with all coefficients nonzero, whose bin passes the full exact screen of \Cref{def:exact-screen} yet is a non-singleton, so B\'ezout recovery again reads a phantom frequency $h \notin \mathcal{S}$. Together with \Cref{prop:two-tone-phantom} this exhibits a screen-passing phantom at every $m\ge 2$. The quantifier is over collision multiplicities, each realized at its own admissible instance; no single fixed $N$ is claimed to host every $m$, and indeed none can, the model forcing $m \le k \le \min_\ell p_\ell$.

\emph{Three-tone witness ($m=3$; exact by construction).} Take $N=46{,}189=11\cdot 13\cdot 17\cdot 19$, stage prime $p=11$, residue class $r=5$, and the three colliding true tones $\{16,27,38\}$ (all $\equiv 5 \bmod 11$; admissible at $k=8$, $c=3$). The $3\times 3$ generalized-Vandermonde system $E A = b_h$, with rows the $\{0,31,63\}$ powers of the tone phasors and target the singleton pattern of the phantom $h=82$ (also $\equiv 5 \bmod 11$), is invertible for this configuration (nonzero determinant, proved exactly in the cyclotomic field $\mathbb{Q}(\zeta_N)$ by the power-basis argument in the proof; exact cyclotomic arithmetic, not a floating-point check), and the exact solution $A = E^{-1} b_h$ makes the bin's measured triple equal, by construction and exactly in exact arithmetic, to that of a single exponential at $h = 82$: the exact screen statistics are $D_{\max} = 0$ and exactly singleton-consistent increment ratios, and the B\'ezout-recovered frequency is $g=82$, which is not a true tone. All three coefficients are nonzero, so the bin really does hold three tones: the same power-basis argument applied to the three Cramer numerators gives exact nonvanishing, and numerically $A_{16} \approx 9.57494 + 2.76693j$, $A_{27} \approx -23.23567 - 5.54910j$, $A_{38} \approx 14.66073 + 2.78217j$. Such a bin injects the phantom $h=82$ into the per-bin recovered-label set: the failure is one of soundness (a recovered label that is not a true tone), not of the candidate count, which \Cref{thm:global-label-count} bounds unconditionally.
\end{proposition}

\begin{proof}
\emph{The observation map and its rank.} For $m$ distinct on-grid tones $g_1,\dots,g_m$ in the bin, the map from coefficients to the observed triple $(c_0,c_{31},c_{63})\in\mathbb{C}^3$ is the $3\times m$ generalized-Vandermonde $E$ with entries $E_{s,i}=x_i^{\sigma_s}$, $x_i=e^{+j2\pi g_i/N}$, on the row exponents $\sigma_s\in\{0,31,63\}$. The coefficient-to-observation map $A \mapsto EA$ is linear with $\operatorname{rank}(E) \ge 2$: rank one would make two columns proportional, forcing $(x_i/x_j)^{31} = (x_i/x_j)^{63} = 1$ for distinct nodes and hence $x_i = x_j$ by $\gcd(31,32) = 1$ (\Cref{thm:coprime-bezout}), a contradiction. For any target phantom $h$ the singleton-pattern vector $b_h=(1,\eta^{31},\eta^{63})$ with $\eta=e^{+j2\pi h/N}$ defines the impersonation set $\{A:EA\in\mathrm{span}(b_h)\}$ of that phantom, a proper subset of the full screen-passing set characterized in \Cref{thm:generic-collision-rejection}.

\emph{An exact nonvanishing criterion.} Write $\zeta = e^{+j2\pi/N}$. For any three distinct tones $g_1,g_2,g_3 \in [0,N)$, expanding the $3\times 3$ determinant on the row exponents $(\sigma_1,\sigma_2,\sigma_3) = (0,31,63)$ gives exactly
\begin{equation}
\det E \;=\; \sum_{\pi \in S_3} \operatorname{sgn}(\pi)\, \zeta^{\,e_\pi}, \qquad e_\pi \coloneqq 31 g_{\pi(2)} + 63 g_{\pi(3)} .
\label{eq:det-criterion}
\end{equation}
\emph{If the six integers $e_\pi$ are pairwise distinct and all are $< \varphi(N)$, then $\det E \ne 0$ exactly.} Indeed the associated integer polynomial $\sum_\pi \operatorname{sgn}(\pi)\, x^{e_\pi}$ is then nonzero (distinct exponents preclude cancellation) of degree $< \varphi(N) = \deg \Phi_N$, so the minimal polynomial $\Phi_N$ of $\zeta$ cannot divide it and it cannot vanish at $\zeta$. This is exact cyclotomic arithmetic, not floating-point evaluation. Replacing column $i$ of $E$ by $b_h$ produces the same matrix on the tone triple with $g_i$ replaced by $h$, so \eqref{eq:det-criterion} also certifies the Cramer numerators of $A = E^{-1}b_h$, and hence that individual coefficients do not vanish.

\emph{The $m=3$ witness.} At $N = 46{,}189$, $\varphi(N) = 34{,}560$, tones $\{16,27,38\}$ and phantom $h=82$, criterion \eqref{eq:det-criterion} applies to all four relevant triples:
the determinant triple itself and, for the Cramer numerators of $A_{16}, A_{27}, A_{38}$, the same triple with $16$, $27$, $38$ respectively replaced by $82$. Their six exponents are
\[
\begin{aligned}
  \det E &: 1845,\,2186,\,2197,\,2879,\,2890,\,3231, \\
  A_{16} &: 2879,\,3231,\,4243,\,4936,\,6003,\,6344, \\
  A_{27} &: 2186,\,2890,\,3550,\,4936,\,5662,\,6344, \\
  A_{38} &: 1845,\,2197,\,3550,\,4243,\,5662,\,6003,
\end{aligned}
\]
distinct within each row and all below $\varphi(N) = 34{,}560$, so $\det E \ne 0$ and all three coefficients of $A = E^{-1}b_h$ are nonzero, exactly. The bin therefore genuinely holds three tones and $EA = b_h$ holds exactly.

\emph{Existence at every $m \ge 3$: an admissible family.} Fix $m \ge 3$ and put $k \coloneqq \max(m, 505)$. By \Cref{lem:prime-existence} the interval $[k,2k] \subseteq [k,ck]$ contains at least two primes; let $p < p'$ be the two smallest of them and set $L = 2$, $N = p\,p'$, an admissible instance of \Cref{def:transform-family} with $m \le k \le p$. Work in the class $r = 0$ of the stage prime $p$, and take as support the $m$ distinct tones
\[
  \mathcal{S} \;=\; \{\, a\,p \;:\; a \in \{0,1,2\} \cup \{4,5,\ldots,m\} \,\},
\]
all $\equiv 0 \pmod p$, with phantom $h = 3p \equiv 0 \pmod p$, $h \notin \mathcal{S}$. Every tone is in range: $\max \mathcal{S} \le m p \le k p \le p\,(p'-1) < N$. Split the support as $T = \{0,p,2p\}$ (the first three) and $U = \mathcal{S}\setminus T$, and write $E_T, E_U$ for the corresponding column blocks.

Criterion \eqref{eq:det-criterion} applies to $T$ and to the three column-swapped triples: with $g = a p$ the six exponents are $p\,(31 a_{\pi(2)} + 63 a_{\pi(3)})$, so for the $a$-patterns $(0,1,2)$, $(3,1,2)$, $(0,3,2)$ and $(0,1,3)$ they are, after division by $p$,
\[
\begin{aligned}
(0,1,2) &: 31,\,62,\,63,\,125,\,126,\,157, \\
(3,1,2) &: 125,\,156,\,157,\,219,\,220,\,251, \\
(0,3,2) &: 62,\,93,\,126,\,189,\,219,\,251, \\
(0,1,3) &: 31,\,63,\,93,\,156,\,189,\,220,
\end{aligned}
\]
pairwise distinct in each row, with every exponent at most $251p$. Since $p \le 2k$ and $\varphi(N) = (p-1)(p'-1) \ge (k-1)^2$, the requirement $251p < \varphi(N)$ follows from $502k < (k-1)^2$, which holds at $k = 505$ ($253{,}510 < 254{,}016$) and, the right side growing quadratically and the left linearly, for every $k \ge 505$. Hence $\det E_T \ne 0$ and all three entries of $E_T^{-1}b_h$ are nonzero.

Now let $\mathbf{1}$ be the all-ones vector of length $m-3$ and, for a scalar $\varepsilon$, set
\[
  A_U(\varepsilon) = \varepsilon\,\mathbf{1}, \qquad A_T(\varepsilon) = E_T^{-1}\bigl(b_h - \varepsilon\,E_U \mathbf{1}\bigr),
\]
so that $E A = E_T A_T + E_U A_U = b_h$ exactly, for every $\varepsilon$. Each entry of $A_T(\varepsilon)$ is an affine function of $\varepsilon$ whose value at $\varepsilon = 0$ is a nonzero entry of $E_T^{-1}b_h$, hence vanishes for at most one $\varepsilon$; choosing $\varepsilon \ne 0$ outside those at most three values makes all $m$ coefficients nonzero. (At $m = 3$ the block $U$ is empty and $A = E_T^{-1}b_h$ directly.) This realizes the claimed configuration at every $m \ge 3$.

\emph{The bin passes exactly.} In each case $EA = b_h$, so the bin's measured triple is exactly that of a single exponential at $h$: the three magnitudes are equal and positive, $w = (b_h)_3 (b_h)_1 / (b_h)_2^2 = \omega^{h}$ is exactly an $N$th root of unity, and $h \equiv r \pmod{p}$, so all three conditions of \Cref{def:exact-screen} hold in exact arithmetic and B\'ezout recovery returns the phantom $g = h \notin \mathcal{S}$. These configurations are exceptional: with $\operatorname{rank}(E) \ge 2$ the impersonation set $\{A : EA \in \mathrm{span}(b_h)\}$ they inhabit has complex dimension at most $m-1$ and is therefore a proper linear subspace of $\mathbb{C}^m$, so Lebesgue-almost-every coefficient vector on the same tones avoids it. The noiseless on-grid model admits arbitrary complex coefficients, so the witnesses are valid adversarial inputs. Per-bin B\'ezout recovery is therefore sound only generically: every collision multiplicity $m\ge 3$ admits a measure-zero screen-passing phantom, so no deterministic worst-case per-bin soundness rule follows at any such $m$.
\end{proof}

\subsection{Shift-Triple Design and Arithmetic-Shift Prony Escalation}

\begin{definition}[Shift-Observation Phase-Consistency Triple]
\label{def:hankel-matrix}
For bin $r$ in the $p_\ell$-decimated spectrum, given $S = 3$ co-prime-increment shifts $\{0, \Delta_1, \Delta_1 + \Delta_2\}$ with $\gcd(\Delta_1, \Delta_2) = 1$, let
\begin{equation}
z_0 \coloneqq X_{p_\ell}[r;0],\quad z_1 \coloneqq X_{p_\ell}[r;\Delta_1],\quad z_2 \coloneqq X_{p_\ell}[r;\Delta_1+\Delta_2].
\end{equation}
The co-prime shift screen evaluates magnitude invariance and phase consistency across $(z_0, z_1, z_2)$.

\begin{remark}[Arithmetic-progression remark]
The triple $\{0, \Delta_1, \Delta_1 + \Delta_2\}$ is not an arithmetic progression unless $\Delta_2 = \Delta_1$, so the classical Hankel/Prony rank-one test does not apply to these three samples. The screen proceeds through the exact ratio/root conditions of \Cref{def:exact-screen}, not through a rank-one claim (the unwrapped-phase reading of \Cref{lem:phase-linearity} is implementation commentary). The classical arithmetic-shift Hankel/Prony rank-one statement is recorded as a separate lemma (\Cref{lem:arith-shift-hankel-rank}) for use only when arithmetic shifts $\{0,\Delta,2\Delta\}$ are deployed.
\end{remark}
\end{definition}

\noindent\emph{Shift-triple design note (B\'ezout phase-diversity, not Hankel rank-one).} The screen guarantees of \S\ref{sec:screens} rest on the increment co-primality $\gcd(\Delta_1, \Delta_2) = 1$ of the triple $\{0, \Delta_1, \Delta_1+\Delta_2\}$. The concrete triple $\{0, 31, 63\}$ is a B\'ezout phase-diversity construction with $\Delta_1=31$, $\Delta_2=32$, $\gcd(31,32)=1$: the two increments are coprime, so no distinct on-grid pair can satisfy both increment congruences, which is exactly what \Cref{thm:coprime-bezout} gives. The triple is not an arithmetic progression, so the matrix $H_r$ of \Cref{def:hankel-matrix} is not Hankel/Prony rank-one for singletons; the screen it powers is the phase-consistency screen of \Cref{thm:phase-consistency-certificate}.

\noindent Implementations that want a deterministic rank-based singleton separator at $m=1$ must acquire samples at an arithmetic-progression triple $\{0, \Delta, 2\Delta\}$ compatible with the engineered observation grid and with the anti-aliasing condition $\gcd(\Delta, N) = 1$ of \Cref{lem:arith-shift-hankel-rank} (for instance $\{0,32,64\}$, whose increment automatically satisfies $\gcd(32, N) = 1$ since $N$ is odd; an increment such as $\Delta = 31$ is admissible only when $31 \nmid N$, which fails whenever $31$ is a stage prime); see \Cref{lem:arith-shift-hankel-rank} for the rank statement. The co-prime-increment triple $\{0,31,63\}$ and an arithmetic triple coexist as two distinct screen designs and the implementation chooses one consistently per stage.

\begin{lemma}[Arithmetic-Shift Hankel/Prony Rank Theorem]
\label{lem:arith-shift-hankel-rank}
Let $r$ be a bin in the $p_\ell$-decimated spectrum containing $m \ge 1$ distinct on-grid complex exponentials at distinct frequencies, and let the shifts form an arithmetic progression $\{0, \Delta, 2\Delta, \ldots, (S-1)\Delta\}$ with a common increment $\Delta$ satisfying the \emph{anti-aliasing condition} $\gcd(\Delta, N) = 1$ (equivalently, assume directly that the nodes $\zeta_t = e^{+j 2\pi f_t \Delta / N}$ are pairwise distinct). This hypothesis is necessary: without it, two distinct frequencies with $\Delta (f_1 - f_2) \equiv 0 \pmod{N}$ alias to the same node and the rank drops (e.g.\ $31 \mid N$, $\Delta = 31$, $f_1 = 0$, $f_2 = N/31$). In the noiseless on-grid model, the resulting Hankel matrix $H_r^{\mathrm{arith}} \in \mathbb{C}^{P \times Q}$ formed from $S = P + Q - 1$ arithmetic-shift observations satisfies $\mathrm{rank}(H_r^{\mathrm{arith}}) = m$ provided $\min(P, Q) \ge m$ (such $P, Q$ can be chosen if and only if $S \ge 2m - 1$). The rank equals $m$ exactly by the classical Prony/Vandermonde theorem~\cite{prony1795,stoica2005spectral}. The rank-one specialization ($m=1$) applies to singleton detection only when the algorithm uses arithmetic shifts; for the co-prime-increment design of \Cref{def:hankel-matrix} the screen is the phase-consistency test of \Cref{thm:phase-consistency-certificate}, not a rank-one test. This lemma is recorded for implementations that elect arithmetic shifts: it is not a step of \Cref{alg:multistage_oklogn} and carries no charge in this paper's cost theorems.
\end{lemma}

\begin{proof}[Proof sketch]
Under arithmetic shifts the Hankel matrix factors as $V_P D V_Q^\top$ where $V_P, V_Q$ are Vandermonde on the $m$ bases $\zeta_t = e^{+j 2\pi f_t \Delta / N}$, $t = 1, \ldots, m$ (the common increment $\Delta$ is essential for the Vandermonde factorization) and $D = \mathrm{diag}(A_1, \ldots, A_m)$. Under $\gcd(\Delta, N) = 1$, distinct on-grid frequencies give distinct nodes ($\zeta_{t_1} = \zeta_{t_2}$ if and only if $\Delta(f_{t_1} - f_{t_2}) \equiv 0 \bmod N$, equivalently $f_{t_1} \equiv f_{t_2} \bmod N$), so $V_P, V_Q$ have full column rank $m$ for $P, Q \ge m$, and $D$ is invertible when $A_t \ne 0$. Hence $\mathrm{rank}(H_r^{\mathrm{arith}}) = m$ exactly~\cite{prony1795,stoica2005spectral}. Without a common increment the factorization does not in general apply and rank-one is not guaranteed for singletons.
\end{proof}

\noindent\emph{Informally:} a lone tone cannot help passing the three-shift test and a generic mixture cannot help failing it; only a measure-zero conspiracy of coefficients can impersonate a tone.

\noindent\emph{Detector hierarchy.} The primary candidate-singleton screening mechanism is the co-prime shift phase-consistency screen (\Cref{thm:phase-consistency-certificate}), instantiated with $S = 3$ shifts via the triple $(z_0, z_1, z_2)$ of \Cref{def:hankel-matrix} read through the exact conditions of \Cref{def:exact-screen} (CV/normalized-MAD and unwrapped-phase regression as implementation surrogates). True singleton bins always pass the screen. Generic multi-tone collisions fail at least one phase-consistency test under the stated nondegeneracy assumptions; a failing bin emits no label (optionally re-examined by arithmetic-shift Prony escalation), and the dense-FFT fallback is triggered only at stage or chain level (\Cref{thm:hybrid-correctness}). The screen is one-sided as in \Cref{rem:one-sided-normative}. The coefficient-of-variation (CV, \Cref{def:cv-singleton-test}) and phase-linearity tests are the two computational surrogates that realize the screen.

\noindent\emph{Escalation hierarchy.} The main text uses the $S = 3$ co-prime shift phase-consistency screen as a fast, one-sided low-order screen that rejects generic collisions (\Cref{thm:phase-consistency-certificate}); its accepts are provisional (\Cref{rem:one-sided-normative}). For higher-order $m$-tone collisions with $m \ge 3$, the algorithm may additionally escalate to an arithmetic-shift Prony/Vandermonde Hankel of order $S \ge 2m - 1$ (at most $2k - 1$, i.e.\ $\le 39$ for $k < 20$, with an anti-aliasing increment $\gcd(\Delta, N) = 1$, per \Cref{lem:arith-shift-hankel-rank}) on a separate sample acquisition when warranted, or leave the bin unlabeled; a screen-failing bin emits no label, and the dense-FFT fallback is triggered only by the stage- and chain-level routing rule of \Cref{thm:hybrid-correctness}.

\noindent With the escalation disabled, the setting of \Cref{alg:multistage_oklogn}, which never invokes it, this structure preserves the $O(k \log N)$ sparse-path complexity (\Cref{def:label-decoding-model}); an implementation that enables the escalation acquires its shifts outside the constant-three-shift accounting and must re-derive the cost bounds (\Cref{rem:collision-escalation}). Deterministic correctness in all collision regimes is secured by the final verifier, never by the screen alone.

\begin{remark}[Relation to classical Prony / \mbox{annihilating-filter} estimation]
\label{rem:prony-relation}
Both the per-bin recovery and the final verifier are deliberately weaker, single-purpose operations than a full Prony / annihilating-filter spectral solve~\cite{prony1795,stoica2005spectral}, including the finite-rate-of-innovation annihilating-filter recovery of \cite{blu2008sparse} that jointly estimates an unknown number of innovations from uniform samples. Per-bin global-label recovery (\Cref{lem:global-label-recovery}) reads a single modular phase ratio $w = u_{32}\,u_{31}^{-1} = e^{+j2\pi g/N}$ from two co-prime increments $\Delta_1 = 31$ and $\Delta_2 = 32$; this is a degree-one annihilation that pins one already-screened single tone, not a rank recovery of an unknown tone count. Classical Prony / Hankel methods instead recover $m$ tones jointly from $S \ge 2m - 1$ arithmetic-progression samples by factoring a Hankel matrix of rank $m$ (the arithmetic-shift statement we record separately in \Cref{lem:arith-shift-hankel-rank}); they need a common increment for the Vandermonde factorization and a sample budget that scales with the tone count. The co-prime-increment triple $\{0,31,63\}$ is not an arithmetic progression and supports no such rank factorization; it instead furnishes the one-sided phase-consistency screen of \Cref{thm:phase-consistency-certificate}, which is why a screen pass yields only a candidate-singleton. Finally, the acceptance test of \Cref{def:verifier} is a residual-vanishing check on $k + |\mathcal{G}|$ consecutive samples (the uniform $2k$-sample window being the special case $|\mathcal{G}| \le k$, never the definition), decided by the rank of a fixed Vandermonde system on the known candidate support, not a Prony solve for unknown frequencies: it certifies exact recovery without ever estimating a tone count.
\end{remark}

\subsection{Magnitude and Phase Surrogate Lemmas (Implementation Commentary)}

The two classical per-bin lemmas below underlie the implementation surrogates of \Cref{def:cv-singleton-test} and \Cref{def:phase-unwrapping}; the formal screen of \Cref{def:exact-screen} does not use them, and the $\arg(\cdot)$-based phase-linearity reading in particular is implementation commentary, never a formal screen predicate.

\begin{lemma}[On-Grid Singleton Magnitude Equality]
\label{lem:singleton-cv}
For an on-grid singleton in the noiseless case, the shifted magnitudes are exactly equal, hence $\mathrm{nMAD}(r) = 0$ exactly.
\end{lemma}

\begin{proof}
By \Cref{thm:singleton-observation-identity}(1), $|X_M[r]|$ is constant across all shifts, so $\text{MAD}(A) = 0$ and therefore $\mathrm{nMAD}(r) = 0$.
\end{proof}

\begin{lemma}[On-Grid Singleton Phase Linearity (Unwrapped-Phase Surrogate)]
\label{lem:phase-linearity}
For an on-grid singleton $f$ in bin $r = f \bmod M$, the shifted observations obey the exact relation $X_M[r; s] = X_M[r; 0]\, e^{+j 2\pi f s/N}$, so any continuous (unwrapped) phase $\phi_s$ of $X_M[r;s]$ (\Cref{def:phase-unwrapping}) is affine in $s$: $\phi_s = 2\pi f s/N + \phi_0$. The principal value $\arg(X_M[r;s])$ satisfies this only modulo $2\pi$; that branch cut is why the affine reading is an implementation surrogate (\Cref{rem:impl-surrogates}) and enters no condition of the formal screen \Cref{def:exact-screen}, whose phase content is carried branch-free by the ratios $u_{31}, u_{32}$ and the root membership of $w$.
\end{lemma}

\begin{proof}
By \Cref{thm:singleton-observation-identity} a shift by $s$ multiplies the decimated observation by the unimodular factor $e^{+j 2\pi f s/N}$, which is the stated relation; taking any continuous branch of the phase along $s$ turns that factor into the additive term $2\pi f s/N$, and the principal value differs from $\phi_s$ by an integer multiple of $2\pi$ at each $s$.
\end{proof}

\subsection{Collision Behavior Across Stages}

The keyed three-view rejection idea for generic multi-tone buckets, together with the associated per-residue pairing-uniqueness claim, belongs to a \emph{keyed multi-view extension} that an earlier version advertised as formal results; it is not required for the core stage/shift sparse-path theorem and is now recorded as design history and future work in \Cref{app:keyed-multiview} (its negative-result counterpart, the exact Cartesian tuple-count obstruction, is in \Cref{app:tuple-negative}).

\subsubsection{Multi-Tone Bin Behavior}

\begin{remark}
\label{rem:consecutive-shifts-blind-spots}
Consecutive shifts like $\{0, 1, 2\}$ create deterministic blind spots where collisions can masquerade as singletons. For instance, with $\Delta f = N/2$, the magnitude pattern is $\{|A_1 + A_2|, |A_1 - A_2|, |A_1 + A_2|\}$, yielding $\text{MAD} = 0$ regardless of amplitudes: a deterministic false accept for any magnitude-only test, which admits the collision as a candidate-singleton.

\emph{Solution via Co-prime Shifts:} This algorithm uses $\{0, 31, 63\}$ where $\gcd(31, 32) = 1$. By B\'ezout's identity, if two frequencies $\nu_1, \nu_2$ satisfy both congruences $(\nu_1 - \nu_2) \cdot 31 \equiv 0 \pmod{2\pi}$ and $(\nu_1 - \nu_2) \cdot 32 \equiv 0 \pmod{2\pi}$, then $\nu_1 = \nu_2$. Thus, under the stated nondegeneracy assumptions, generic multi-tone collisions fail at least one phase-consistency test; degenerate or unresolved collisions whose bins fail the $S=3$ screen emit no label (optionally re-examined by arithmetic-shift Prony escalation~\cite{vaidyanathan2001theory}), the dense FFT fallback being triggered at stage or chain level per \Cref{thm:hybrid-correctness}.

\emph{Multi-dimensional Detection:} Combined with:
\begin{enumerate}
\item Magnitude equality (implementation surrogate: normalized MAD $< \tau_{cv}$): the shifted magnitudes across the three co-prime shifts are equal for a true singleton
\item Phase linearity (implementation surrogate: $R^2 > 1 - \varepsilon$): unwrapped phases fit an affine model~\cite{tribolet1977new}
\item Frequency consistency (optional implementation surrogate): Quinn/Jacobsen estimators~\cite{quinn1994estimating,jacobsen2007fast} agree within tolerance
\end{enumerate}

This multi-dimensional screen with co-prime shifts has zero false negatives at the singleton level: every true singleton is correctly accepted in the stated noiseless, on-grid model.
\end{remark}

\subsubsection{Multi-Stage Collision Rejection}

\begin{lemma}[Collision Separation Across Stages]
\label{lem:collision-rare}
Let $f_1 \neq f_2 \in [0,N{-}1]$ and let the $L$ stages use pairwise coprime moduli $p_1, \ldots, p_L$. The pair $(f_1, f_2)$ collides in stage $i$ (i.e., $f_1 \equiv f_2 \pmod{p_i}$) if and only if $p_i \mid (f_1 - f_2)$. The number of stages in which the pair collides is the number of stage moduli dividing $(f_1 - f_2)$.

\emph{Key property:} Under the exact-product model $\prod_{i=1}^L p_i = N$ (the generalized CRT uniqueness condition is $\prod_i p_i \ge N$, used in the dense-fallback / extension case), no pair can collide in ALL $L$ stages simultaneously, since that would require $\prod p_i \mid (f_1 - f_2)$, which forces $f_1 = f_2$.

\emph{Regime caveat:} For stage primes $p_i \in [k, ck]$, we have $p_i p_j \approx k^2$, which lies far below $N$ at every instance with $L \ge 3$ stages (there $N \ge k^3$). Therefore, a pair CAN collide in multiple stages: the pairwise product condition $p_i p_j \geq N$ does NOT hold. Separation is guaranteed only over the full $L$-stage product, not over any two stages. A non-operative per-view occupancy estimate (\Cref{rem:crt-occupancy-estimate}) bounds this multiplicity contraction across all $L$ stages.
\end{lemma}

\begin{proof}
If $f_1 \equiv f_2 \pmod{p_i}$ and $f_1 \equiv f_2 \pmod{p_j}$ with $\gcd(p_i,p_j)=1$, then $p_i p_j \mid (f_1-f_2)$ by CRT. More generally, if the pair collides in stages $S \subseteq \{1, \ldots, L\}$, then $\prod_{i \in S} p_i \mid (f_1 - f_2)$. Since $|f_1 - f_2| < N$ and $\prod_{i=1}^L p_i \geq N$, the pair cannot collide in all $L$ stages. However, for any proper subset $S \subsetneq \{1, \ldots, L\}$ with $\prod_{i \in S} p_i < N$, simultaneous collision in stages $S$ is possible and must be handled by the remaining stages.
\end{proof}

\begin{remark}[Non-operative view-based occupancy estimate (CRT counting, not the operative path)]
\label{rem:crt-occupancy-estimate}
The following per-view occupancy estimate is a non-operative counting heuristic, recorded for intuition only. It is not part of the operative recovery chain, which runs purely through the one-sided phase-consistency screen (\Cref{thm:phase-consistency-certificate}), the global-label engine (\Cref{thm:global-label-count,thm:global-label-completeness}), the sound verifier (\Cref{thm:verifier-soundness}), and the dense fallback (\Cref{thm:hybrid-correctness}). It is retained because it gives a legitimate deterministic worst-case bound on per-view collision multiplicity.

Fix $L \geq 1$ stages; this counting bound is stated for any stage count and for arbitrary pairwise-coprime modulus sequences, so it does not inherit the standing $L \ge 2$ of \Cref{def:transform-family} and is not narrowed by it. For each view $v \in \{1, 2, 3\}$, fix a sequence of pairwise coprime moduli $\{q_{v,t}\}_{t=1}^{T_v}$ and define the stage-$t$ aggregate modulus $Q_{v,t} \coloneqq \prod_{j=1}^{t} q_{v,j}$. For a $k$-sparse support $S \subset [0, N)$, define the per-view multiplicity of $f$ at stage $t$ as $m_{v,t}(f) \coloneqq \big|\{g \in S \setminus \{f\} : g \equiv f \pmod{Q_{v,t}}\}\big|$. Then, for any $f \in S$, any view $v$, and any support set $S$,
\begin{equation}
\label{eq:crt-counting-bound}
m_{v,t}(f) \leq \min\!\left(k - 1,\; \left\lfloor \frac{N - 1}{Q_{v,t}} \right\rfloor\right),
\end{equation}
deterministically and worst-case: it counts lattice points in $[0, N)$ congruent to $f$ modulo $Q_{v,t}$ (an element $g \neq f$ collides with $f$ if and only if $Q_{v,t} \mid (g-f)$, and $|g-f| < N$ leaves at most $\lfloor (N-1)/Q_{v,t} \rfloor$ such values), and requires no distributional assumption on $S$.

\emph{Geometric contraction and stage budget.} Since $Q_{v,t} \geq \rho^t$ for the minimum modulus $\rho = \min_j q_{v,j} \geq k$, the bound contracts geometrically, $m_{v,t}(f) \leq \lfloor (N-1)/\rho^t \rfloor$, and reaches $m_{v,L}(f) = 0$ once $Q_{v,L} \geq N$ at $L = \min\{t : \prod_{j=1}^{t} q_{v,j} \ge N\} = O(\log_k N)$ stages. This is a counting bound on accumulated CRT residue constraints, not an operational statement that any algorithmic decay sequence is executed: whether a given frequency enters the stage label set additionally requires that its bin pass the occupied equal-magnitude pre-screen and emit a class-consistent global label at the corresponding stage (\Cref{rem:alg1-semantics}), and membership of a true tone in the candidate set is the stronger genuine-singleton-survival property of \Cref{thm:global-label-completeness} (a tone that collides at some stage need not appear, and is handled by the dense fallback). The estimate does not claim the candidate set equals the support: the bin screen is one-sided (\Cref{rem:one-sided-normative}), so exactness is secured by verifier-gated acceptance, not by this occupancy bound.
\end{remark}

\subsection{Verifier Scope and a Reduced-Sample Variant}

The final verifier of \Cref{def:verifier} is the object that carries correctness, and two scope statements about it are recorded here rather than in the main line: the exact-arithmetic scope of its soundness, and an alternative reduced-sample construction that is not adopted.

\begin{remark}[Exact-arithmetic, noiseless scope of verifier soundness]
\label{rem:verifier-noiseless-floor}
\Cref{thm:verifier-soundness} is an exact-arithmetic, noiseless, on-grid property, and no robustness to noise is claimed. To gauge the scale of the residual a near-miss can produce, consider one unit-amplitude tone placed in a bin adjacent to a true tone: its nonzero residual on the $k + |\mathcal{G}|$-sample window is of order $\Theta(1/N^2)$, and a coherent $k$-tone near-miss is of order $\Theta(k/N^2)$. This $\Theta(k/N^2)$ figure is an illustrative near-miss scale, not a universal residual lower bound.

With arbitrary complex coefficients a nonzero residual can be scaled below any fixed tolerance, and coherent multi-tone cancellation can shrink the particular checked samples further. Consequently a fixed absolute float tolerance is unsound: there is no single ``floor'' below which a fixed absolute zero test is safe. The formal model is therefore exact arithmetic; a fielded implementation requires exact or rational arithmetic on the $k + |\mathcal{G}|$ samples (or coefficient-normalized, condition-aware thresholds), not a single absolute floor. The high-degree Vandermonde conditioning of the system is irrelevant to the matrix rank (which is exactly $s$ for distinct nodes, since $\det = \prod_{i<j}(\omega_j - \omega_i) \neq 0$); the soundness statement is exact and has no gap in exact arithmetic, the float caveat being parallel to the $\mathrm{nMAD} < \tau_{cv}$ surrogate already used for the magnitude screen.
\end{remark}

\begin{remark}[Optional reduced-sample CRT variant]
\label{rem:verifier-crt-variant}
An alternative verifier reads decimated views $X_{m_j}$ of $x$ at pairwise-coprime divisor-moduli $m_j \mid N$ and passes if and only if every residual view $R_j[\rho] = (m_j/N)\sum_{f \equiv \rho} \hat{r}[f]$ is zero. By the aliasing identity of \Cref{subsec:sampling-operator}, an isolated residual frequency forces its view bin nonzero, so soundness holds provided the family is \emph{$2k$-separating} (every $\le 2k$-subset has each member isolated in some view). A worst-case $2k$-separating family requires $J = \Theta(k\log_k N)$ views, a factor $\Theta(k)$ above the $L = \Theta(\log_k N)$ recovery stages, and each view must divide $N$, so such a family does not in general fit on the engineered recovery length; \Cref{def:verifier} is therefore preferred for the unconditional guarantee. The CRT variant is retained only as a reduced-sample option when a divisor-compatible $2k$-separating family is engineered into $N$.
\end{remark}

\section{Detailed Global-Label Worked Example: $N = 46{,}189$, $k=8$}
\label{app:detailed-example}

This appendix presents the complete per-stage global-label trace of a single representative frequency ($f = 67$) for the running $N = 46{,}189$, $k=8$ example introduced in \Cref{sec:algorithm-overview}, following the operative engine of \Cref{subsec:global-label-engine} (\Cref{lem:global-label-recovery,def:total-recovery,thm:global-label-count}). The engineered transform length $N = 11 \cdot 13 \cdot 17 \cdot 19 = 46{,}189$ satisfies the divisor-compatibility envelope of \Cref{thm:sampling-bluestein} as an equality $Q_L = N$.

\emph{Scope of this appendix.} The scope disclaimer of \Cref{ex:n1024_k8} applies unchanged: this is a single-frequency trace, not a full $k = 8$ support example. At each stage $\ell$ the three shifted observations of the screened bin are read by the B\'ezout phase rule (\Cref{lem:global-label-recovery}) into a single global frequency label $g$, the per-stage label sets $\mathcal{G}_\ell$ are intersected, and the surviving candidate is coefficient-estimated and verifier-checked. No residue tuple is ever formed (\Cref{rem:no-product-blowup}); the Garner/CRT reconstruction plays no part in the trace and is recorded, as a non-operative cross-check only, in a footnote.\footnote{\emph{Non-operative CRT cross-check (not the demonstrated path).} The operative engine recovers the global label $g = 67$ directly per bin and forms no residue tuple. For completeness only, the residue tuple $(1,2,16,10)$ also identifies $67$ by CRT injectivity via the standard two-step Garner lift~\cite{garner1959,knuth1997}, $f = r_1 + M_1 \cdot [(r_2 - r_1) \cdot M_1^{-1} \bmod M_2]$ applied incrementally ($11^{-1}\bmod 13 = 6$, $143^{-1}\bmod 17 = 5$, $2431^{-1}\bmod 19 = 18$, assembling $a_2 = 1 + 11\cdot 6 = 67$ and $a_3 = a_4 = 67$). This tuple route is the superseded screen-only path of \Cref{app:tuple-negative}; it is a textbook cross-check and is not the engine the algorithm runs.}

Concretely, the trace exhibits: the per-stage B\'ezout phase read $u_{31}, u_{32}, w$ and the exact root-of-unity index $g = \log_\omega w$; the class check $g \equiv r_\ell \pmod{p_\ell}$; the emission of $g$ into $\mathcal{G}_\ell$; the intersection $\mathcal{G} = \bigcap_\ell \mathcal{G}_\ell = \{67\}$ with $|\mathcal{G}| = 1 = O(k)$; the single zero-shift coefficient $\hat{X}[67] = (N/p_\ell)\,z_0$; and a final-verifier PASS that certifies $\hat{x} = x$ exactly. The amplitude-normalization convention is fixed by the stage operator of \Cref{subsec:sampling-operator}: on a singleton bin holding tone $g$, the decimated-DFT value carries the aliasing prefactor $p_\ell/N$, so $|X_{p_\ell}^{(\sigma)}[r]| = (p_\ell/N)\,|X[g]|$ at every shift $\sigma$ and stage $\ell$, and the recovered coefficient is $X[g] = (N/p_\ell)\, X_{p_\ell}^{(0)}[r]$ (the zero-shift read-out). The numerical values below are computed with the PLUS-sign analysis phase $e^{+j 2\pi \sigma f/N}$ of \Cref{subsec:sampling-operator} for the true coefficient $X[67] = 1$. Screen convention for this trace: shifts $\{0, 31, 63\}$; formal screen is the exact ratio/root screen of \Cref{def:exact-screen} (occupied equal-magnitude bin, $D_{\max} = 0$, root membership, class consistency); label emission follows the total rule of \Cref{def:total-recovery} executed by the costed procedure of \Cref{def:label-decoding-model}; measure-zero two-tone screen-passing phantoms are possible under this shift set (\Cref{prop:two-tone-phantom}), though none arises in this single-tone trace.

\noindent\emph{Stage 1: Parameter Selection.}
\begin{itemize}
  \item Use $L = 4$ stages with the exact-product transform length $N = \prod_{\ell=1}^{4} p_\ell = 11 \cdot 13 \cdot 17 \cdot 19 = 46{,}189$ under engineered divisor compatibility (the generic logarithmic stage-count bound $\lceil \log_8 N \rceil = 6$ is loose here)
  \item Select primes $\{p_1, p_2, p_3, p_4\} = \{11, 13, 17, 19\}$ (in $[8, 19]$; range widened from $[k, 2k] = [8, 16]$ to $[k, 3k]$ to obtain 4 primes, per \Cref{lem:prime-existence})
  \item Set co-prime shifts $\{\sigma_1, \sigma_2, \sigma_3\} = \{0, 31, 63\}$ with $\gcd(\Delta_1, \Delta_2) = \gcd(31, 32) = 1$ for the $\Delta_1, \Delta_2$ increment condition
\end{itemize}

\noindent\emph{Stages 2 to 4: Per-Stage Global-Label Recovery for $f = 67$.}

For each stage $\ell$ we read the screened bin $r_\ell = 67 \bmod p_\ell$ at the three shifts $\{0,31,63\}$, form the per-step ratios $u_{31} = z_{31}/z_0$ and $u_{32} = z_{63}/z_{31}$, and recover a single global label $g$ from $w = u_{32}\,u_{31}^{-1} = \omega^{g}$ (\Cref{lem:global-label-recovery}, $\omega = e^{+j2\pi/N}$). Because the phase read returns the global index, the recovered $w$ is identical at every stage: $\arg u_{31} = 2\pi\cdot 31\cdot 67/N = 16.1883^\circ$, $\arg u_{32} = 2\pi\cdot 32\cdot 67/N = 16.7105^\circ$, so $\arg w = 2\pi\cdot 67/N = 0.5222^\circ$ and $g = \log_\omega w = 67$ at each stage. The stages differ only in the common bin magnitude $|z_\sigma| = p_\ell/N$ and in the class check $g \equiv r_\ell \pmod{p_\ell}$, the latter being the per-stage acceptance gate that admits $g = 67$ into $\mathcal{G}_\ell$.

\medskip
\noindent\textbf{\emph{Stage 1 ($p_1 = 11$):}}
\begin{itemize}
  \item Screened bin: $r_1 = 67 \bmod 11 = 1$
  \item Observations ($|z_\sigma| = p_1/N = 11/46189 = 2.3815 \times 10^{-4}$): $z_0 = 2.3815 \times 10^{-4}$; $z_{31} = (2.2871 + 0.6640\,j)\times 10^{-4}$; $z_{63} = (1.9996 + 1.2935\,j)\times 10^{-4}$, i.e.\ $z_\sigma = (p_1/N)\,e^{+j2\pi\sigma\cdot 67/N}$
  \item B\'ezout phase read: $u_{31} = z_{31}/z_0 = e^{+j\,16.1883^\circ}$, $u_{32} = z_{63}/z_{31} = e^{+j\,16.7105^\circ}$, $w = u_{32}\,u_{31}^{-1} = e^{+j\,0.5222^\circ}$
  \item Global label: $g = \log_\omega w = 67$; class check $g \bmod 11 = 1 = r_1$ \checkmark $\Rightarrow$ emit $g = 67$
  \item Stage label set: $\mathcal{G}_1 = \{67\}$
\end{itemize}

\noindent\textbf{\emph{Stage 2 ($p_2 = 13$):}}
\begin{itemize}
  \item Screened bin: $r_2 = 67 \bmod 13 = 2$
  \item Observations ($|z_\sigma| = p_2/N = 13/46189 = 2.8145 \times 10^{-4}$): $z_0 = 2.8145 \times 10^{-4}$; $z_{31} = (2.7029 + 0.7847\,j)\times 10^{-4}$; $z_{63} = (2.3632 + 1.5287\,j)\times 10^{-4}$
  \item B\'ezout phase read: $u_{31} = e^{+j\,16.1883^\circ}$, $u_{32} = e^{+j\,16.7105^\circ}$, $w = e^{+j\,0.5222^\circ}$ (identical to stage 1)
  \item Global label: $g = \log_\omega w = 67$; class check $g \bmod 13 = 2 = r_2$ \checkmark $\Rightarrow$ emit $g = 67$
  \item Stage label set: $\mathcal{G}_2 = \{67\}$
\end{itemize}

\Needspace*{5\baselineskip}
\noindent\textbf{\emph{Stage 3 ($p_3 = 17$):}}
\begin{itemize}
  \item Screened bin: $r_3 = 67 \bmod 17 = 16$
  \item Observations ($|z_\sigma| = p_3/N = 17/46189 = 3.6805 \times 10^{-4}$): $z_0 = 3.6805 \times 10^{-4}$; $z_{31} = (3.5346 + 1.0261\,j)\times 10^{-4}$; $z_{63} = (3.0903 + 1.9991\,j)\times 10^{-4}$
  \item B\'ezout phase read: $u_{31} = e^{+j\,16.1883^\circ}$, $u_{32} = e^{+j\,16.7105^\circ}$, $w = e^{+j\,0.5222^\circ}$
  \item Global label: $g = \log_\omega w = 67$; class check $g \bmod 17 = 16 = r_3$ \checkmark $\Rightarrow$ emit $g = 67$
  \item Stage label set: $\mathcal{G}_3 = \{67\}$
\end{itemize}

\noindent\textbf{\emph{Stage 4 ($p_4 = 19$):}}
\begin{itemize}
  \item Screened bin: $r_4 = 67 \bmod 19 = 10$
  \item Observations ($|z_\sigma| = p_4/N = 19/46189 = 4.1135 \times 10^{-4}$): $z_0 = 4.1135 \times 10^{-4}$; $z_{31} = (3.9504 + 1.1468\,j)\times 10^{-4}$; $z_{63} = (3.4539 + 2.2343\,j)\times 10^{-4}$
  \item B\'ezout phase read: $u_{31} = e^{+j\,16.1883^\circ}$, $u_{32} = e^{+j\,16.7105^\circ}$, $w = e^{+j\,0.5222^\circ}$
  \item Global label: $g = \log_\omega w = 67$; class check $g \bmod 19 = 10 = r_4$ \checkmark $\Rightarrow$ emit $g = 67$
  \item Stage label set: $\mathcal{G}_4 = \{67\}$
\end{itemize}

\noindent\emph{Intersection.} The per-stage label sets are combined by intersection (\Cref{thm:global-label-count}), never by a Cartesian product:
\[
  \mathcal{G} = \mathcal{G}_1 \cap \mathcal{G}_2 \cap \mathcal{G}_3 \cap \mathcal{G}_4 = \{67\}, \qquad |\mathcal{G}| = 1 = O(k).
\]
Each stage emits one label per screened bin and the stages are intersected, so the count is the proved $O(k)$ bound of \Cref{thm:global-label-count}; no residue tuple $(1,2,16,10)$ is formed and the multiplicative $2^L$-type rectangle of \Cref{rem:no-product-blowup} cannot arise.

\noindent\emph{Coefficient recovery.} The single surviving candidate $g = 67$ is coefficient-estimated from the zero-shift observation at any stage (\Cref{lem:global-label-recovery}(4)):
\[
  \hat{X}[67] = (N/p_\ell)\, z_0 = (N/p_\ell)\,X_{p_\ell}^{(0)}[r_\ell] = 1.0 + 0.0\,j,
\]
a single zero-shift estimate (the $\sigma=0$ analysis phase is unity, so no de-rotation is required). The four stages give the identical value to machine precision; this is an exactness identity, not a median or average over the $3L$ observations.

\noindent\emph{Final verification.} The sparse candidate $\hat{X}$ supported on $\mathcal{G} = \{67\}$ is handed to the sound final verifier of \Cref{def:verifier}, sized to $k + |\mathcal{G}|$ consecutive residual samples (\Cref{def:verifier}). The residual $r = x - \hat{x}$ is identically zero on those samples, so the verifier returns PASS, certifying $\hat{x} = x$ exactly (\Cref{thm:verifier-soundness}). Had a phantom label survived the intersection and been used in this candidate, the residual would be nonzero, the verifier would reject the candidate, and the input would route to the dense $O(N \log N)$ fallback (\Cref{rem:phantom-labels,thm:hybrid-correctness}).

\section{Cost Models: Further Discussion}
\label{app:cost-models}

This appendix collects the discussion behind the two cost models of \Cref{def:label-decoding-model}: the per-bin costed decoding procedure in step-by-step form, the full disclosure of what model (b) is and is not (with the bit-level implementability discussion and the Indyk, Kapralov \& Price concession, including its source provenance), and the practical decoding consequences for a fielded implementation, moved here from \S\ref{sec:computational}. Nothing in this appendix alters the models: both cost theorems (\Cref{thm:sparse-path-complexity,thm:cost-comparison-model}) are stated against, and read from, the compact definition in the body.

\begin{table*}[!t]
\centering
\caption{Illustrative Design Targets for the Adaptive-View Extension}
\label{tab:adaptive-views}
\small
\begin{tabular}{cccccc}
\toprule
$\theta_{\min}$ & Views & $\theta$ Used & $p = e^{-1/\theta}$ & First-Pass Yield & Rounds \\
\midrule
6 & 3 & 6 & 0.846 & $p^3 \approx 60\%$ & 2-3 \\
3 & 3 & 3 & 0.716 & $p^3 \approx 37\%$ & 2-3 \\
2 & 4 & 2 & 0.607 & $p^3(4-3p) \approx 49\%$ & 2-3 \\
1 & 4 & 3 (bumped) & 0.716 & $p^3(4-3p) \approx 68\%$ & 1-2 \\
\bottomrule
\end{tabular}
\end{table*}

\subsection{The Costed Decoding Procedure in Full}

Model (a) of \Cref{def:label-decoding-model} grants three unit-cost oracles (exact $\arg(\cdot)$, exact root-power evaluation, and exact nearest-integer rounding) and charges the per-bin label read at $O(1)$. These unit-cost primitives are an oracle assumption about root-index decoding: they assert only the per-read charge assigned to the three exact reads, and make no claim that unit-cost exact $\arg$, root-power, or rounding operations are implementable at the bit level. The rounding oracle is what licenses the $O(1)$ charge for the index conversion in step~1 below; without it, locating the nearest of the $N$ admissible indices costs $O(\log N)$ exact comparisons.

Per candidate-singleton of stage $\ell$ (class $r$), the label read $g = \log_\omega w$ of \Cref{def:total-recovery} is executed as the costed procedure:
\begin{enumerate}
  \item (Guess.) Compute the rounded phase index $g_0 = \operatorname{round}\!\bigl(N \arg(w)/2\pi\bigr) \bmod N$; the rounding is a candidate generator only, and no correctness rests on it. Cost: $O(1)$ (one unit-cost $\arg$ evaluation and one unit-cost nearest-integer rounding, both oracle primitives of this model).
  \item (Exact root-membership verification.) Verify $\omega^{g_0} = w$ by exact comparison. If $w = \omega^{g}$ for some integer $g \in [0,N)$, exact $\arg$-extraction and exact rounding give $g_0 = g$, so the comparison succeeds precisely when $w$ is an $N$th root of unity, and then $g_0$ is its unique index; this decides condition~1 of \Cref{def:total-recovery}. Cost: $O(1)$.
  \item (Exact residue verification.) Verify the class consistency $g_0 \equiv r \pmod{p_\ell}$ (condition~2 of \Cref{def:total-recovery}) by exact integer arithmetic, $O(1)$; a label that persists across the stage chain additionally incurs the $O(L)$ per-candidate intersection arithmetic already charged in \Cref{thm:sparse-path-complexity}.
\end{enumerate}
The bin emits the label $g = g_0$ only when both verifications pass exactly, and otherwise emits no label; the procedure therefore implements the emit conditions of \Cref{def:total-recovery}, in their full if-and-only-if force, without changing their semantics, exactness resting on the exact comparisons and never on the $\arg$-quantization. In model (b), with the rounding and root-evaluation oracles withdrawn, each of the two removed primitives is replaced at $O(\log N)$ per bin (\Cref{def:label-decoding-model}(b)). The resulting candidate-construction total, $O(k \log^2\! N / \log k)$, is stated and proved as \Cref{thm:cost-comparison-model}.

\subsection{The Comparison Model: Disclosure and the IKP Concession}

\emph{What model (b) is and is not.} The single primitive it removes is the root-index oracle: unit-cost rounding and unit-cost root evaluation. It is not a bit-complexity model, and no claim of bit-level implementability is made for it: exact real arithmetic, exact comparison, and above all unit-cost exact $\arg(\cdot)$ are all retained, and this paper supplies no finite input representation, no precision or word-size bound, and no implementation of exact-angle comparison. Model (b) is therefore strictly weaker than model (a) and strictly stronger than any bit-level machine; the honest reading of the gap between \Cref{thm:sparse-path-complexity} and \Cref{thm:cost-comparison-model} is the price of the root-index read alone, not the price of descending to bits. A genuine bit-complexity account of either model is left open (\S\ref{sec:limitations}). One comparison is worth stating against this paper: Indyk, Kapralov \& Price~\cite{ikp2014soda} likewise assume unit-cost access to a precomputed object, the values of their flat-window filter and its transform, and then discharge that assumption, computing the values on the fly to polynomially small precision without affecting their overall running time; beyond that discharged assumption their machine model is left unstated (\Cref{tab:comparison}). No such discharge is offered here for the root-index read, which is why it is declared as a model rather than removed.

We record these named models as a portable convention (per-bin online decoding cost, with the preprocessing and storage posture declared explicitly) for the companion deterministic CRT sparse-FFT papers~\cite{flouro2026keyed,flouro2026safety}.

\subsection{Root-Index Decoding in Practice}

\noindent In model (a) (\Cref{def:label-decoding-model}), exact $\arg(\cdot)$, exact root powers $\omega^{m}$, and exact nearest-integer rounding are unit-cost, so the read is $O(1)$ per bin: the rounded phase index $g_0 = \operatorname{round}(N \arg(w)/2\pi)$ is a guess, and the exact comparison $\omega^{g_0} = w$ together with $g_0 \equiv r \pmod{p_\ell}$ decides emission. This is where the headline $O(k \log N)$ of \Cref{thm:sparse-path-complexity} comes from, and it is a modeling assumption about decoding: no claim is made that these primitives are available at unit cost at the bit level.

\noindent In the comparison real-RAM model (\Cref{def:label-decoding-model}(b)), the two removed primitives are paid for explicitly: locating $g_0$ among the $N$ admissible indices takes $O(\log N)$ exact comparisons by binary search, and forming $\omega^{g_0}$ takes $O(\log N)$ exact multiplications by repeated squaring, so the per-bin read is $O(\log N)$ and the construction total rises to the $O(k \log^2 N / \log k)$ of \Cref{thm:cost-comparison-model}. A fielded implementation sits in this second model unless it can supply the unit-cost primitives, and this is the honest reading of the cost gap between the two theorems.

\noindent Two implementation postures are excluded by design. A precomputed table of the $N$ roots would make the read a lookup, but at $\Theta(N)$ preprocessing and storage, defeating the sublinear budget; \Cref{def:label-decoding-model} therefore declares no preprocessing and $O(1)$ storage beyond the samples. And a floating-point $\arg$ read cannot replace the exact comparison: the rounded index is a candidate generator only, and every emission decision rests on the exact root-membership and class tests, so no correctness claim of this paper depends on the quantization of $\arg$.

\section{Keyed Multi-View Extension (Design History and Future Work)}
\label{app:keyed-multiview}

This appendix records the keyed three-view (multi-view) extension and its adaptive-modulus parameter schedule as design history and future work, not as results. The material is not part of the core engine: it is required neither for the stage/shift sparse-path theorem (\Cref{thm:sparse-path-complexity}) nor for correctness, which comes from the sound final verifier (\Cref{thm:verifier-soundness}). The deterministic candidate count on the screened path is likewise the proved $O(k)$ bound of \Cref{thm:global-label-count}, obtained from stages and shifts rather than views.

\medskip
\noindent\textbf{Withdrawn statements.} An earlier version of this paper presented four formal statements in this appendix: a keyed multi-tone rejection theorem, a deterministic collision-rejection lemma built on a keyed primary-view selector, a generic per-residue pairing-uniqueness theorem, and an adaptive moduli-scaling proposition. Their hypotheses were never formally defined: the keyed gain patterns and the key-matching predicate were never specified, the primary-view selector was never given a mathematical definition (no total order was constructed), and the pairing-uniqueness argument assumed a per-residue property that no preceding result establishes. The moduli-scaling proposition argued its four-view case from the inequality $(\theta k)^3 \le (\theta k)^4$, which cannot lower-bound every triple product. All four are therefore withdrawn as unproven design directions, not results; the paper claims nothing on their basis.

\medskip
\noindent\textbf{What remains sound.} Two fragments of the withdrawn material are elementary and correct, and are retained as observations. Triple injectivity: for pairwise coprime view moduli $M_1, M_2, M_3$ with $M_1 M_2 M_3 \ge N$, each on-grid frequency $f \in [0, N)$ determines a unique residue triple $(f \bmod M_1,\, f \bmod M_2,\, f \bmod M_3)$; this is the Chinese Remainder Theorem and uses no keyed machinery. Simultaneous collision exclusion: consequently, two distinct active frequencies cannot collide in all three views at once, since equal residues in every view would force equality in $[0, N)$. Neither fragment yields per-view injectivity on a gated active set, and neither constrains multi-tone bucket behavior under keyed gains; those were exactly the withdrawn claims.

\medskip
\noindent\textbf{Why this is design history rather than a live path.} A separate manuscript studies whether increasing the view count can restore deterministic soundness. Because that result is not used here, we do not state it formally: nothing in the operative chain of this paper rests on it, and this appendix records the extension as design history only, with no soundness claim attached to any view count.

\medskip
\noindent\textbf{Retained design narrative: adaptive moduli and views.} The implementation idea was to select $V$ pairwise coprime moduli as the smallest $V$ distinct primes $M_i \ge \theta \cdot k$, with the view count $V$ and the scaling factor $\theta$ chosen from the signal length: compute
\[
\theta_{\min} = \left\lceil \left(\frac{N}{k^3}\right)^{1/3} \right\rceil
\]
and set $(V, \theta) = (3, \theta_{\min})$ if $\theta_{\min} \ge 3$; $(V, \theta) = (4, 2)$ if $\theta_{\min} = 2$; and $(V, \theta) = (4, 3)$ if $\theta_{\min} = 1$. The three-view instantiation is covered by the case-specific check $(\theta_{\min} k)^3 \ge N$, which gives $M_1 M_2 M_3 \ge N$ and hence triple injectivity. A four-view variant intended to decode from any 3 clean views would instead need every triple product to satisfy $\min_{i<j<\ell} M_i M_j M_\ell \ge N$, a condition to be verified per instance; no general argument for it is claimed here.

This schedule sits outside the engineered divisor-compatibility envelope of \Cref{thm:sampling-bluestein} (the moduli need not divide $N$), so none of the paper's formal cost or correctness statements apply to it. The entries of \Cref{tab:adaptive-views} are derived from the per-view singleton-occupancy heuristic $p = e^{-1/\theta}$ under a uniform-residue assumption; they are illustrative design targets, not measurements, and carry no guarantee. The figures account only for nominal per-view occupancy and explicitly exclude cross-stage association cost, final verification cost, dense-fallback frequency, and implementation constants. The $\theta_{\min} = 1$ case bumps to $\theta = 3$ (rather than using $\theta = 2$) as a design choice that increases $p$.

\medskip
\noindent\textbf{Worked parameter choices (design targets only).} A fixed $\theta = 6$ with 3 views guarantees coverage of $[0, N)$ for every $N \le (6k)^3 = 216k^3$, a conservative range: prime moduli $M_i \ge 6k$ make the true product $M_1 M_2 M_3 \ge (6k)^3$, and coverage persists up to that true product. At $k = 8$ the smallest primes $\ge 48$ are $53, 59, 61$, whose product $190{,}747$ exceeds $48^3 = 110{,}592$, so for example $N = 120{,}000$ lies above $(6k)^3$ yet is still covered. An $N$-dependent choice of $\theta$ and $V$ removes this ceiling at the design-parameter level (\Cref{tab:adaptive-views}):
\begin{itemize}
\item $N = 1{,}024$, $k = 8$: $\theta_{\min} = 2 \to$ 4 views, primes $M_1{=}17, M_2{=}19, M_3{=}23, M_4{=}29$ (triple product $\ge 7{,}429$)
\item $N = 16{,}384$, $k = 16$: $\theta_{\min} = 2 \to$ 4 views, primes $M_1{=}37, M_2{=}41, M_3{=}43, M_4{=}47$ (triple product $\ge 65{,}231$)
\item $N = 10^6$, $k = 50$ (beyond the regime of \Cref{rem:engineering-regime}; illustration only): $\theta_{\min} = 2 \to$ 4 views, primes near 100 (triple product $\ge 10^6$)
\end{itemize}

\bibliographystyle{IEEEtran}
\bibliography{OkLogN}

\end{document}